\documentclass[12pt,a4paper]{article}

\usepackage{amssymb,amsmath,geometry,ulem,graphicx,caption,color,setspace,natbib,array,hyperref,amsthm,float,bm,tikz,url,enumerate}

\usepackage{subcaption}

\usepackage[T1]{fontenc}
\usepackage{newpxmath}
\usepackage{newpxtext}
\usepackage{threeparttable}
\usepackage{booktabs}
\usetikzlibrary{patterns}

\usepackage{pgfplots} 
\pgfplotsset{compat=1.17}
\usepgfplotslibrary{fillbetween}
\usetikzlibrary{arrows.meta}
\usepackage[title]{appendix}

\allowdisplaybreaks

\theoremstyle{plain}\newtheorem{theorem}{Theorem}
\theoremstyle{plain}\newtheorem{proposition}{Proposition}
\theoremstyle{plain}\newtheorem{corollary}{Corollary}
\theoremstyle{plain}\newtheorem{lemma}{Lemma}
\theoremstyle{plain}

\theoremstyle{definition}
\newtheorem{definition}{Definition}

\newtheorem{example}{Example}

\newtheorem{observation}{Observation}

\newcommand{\supp}[1]{\text{supp}(#1)}

\newcommand{\proofpart}[1]{\textbf{#1.}\enspace\ignorespaces}

\usetikzlibrary{positioning}

\makeatletter
\let\@makefntextOring\@makefntext
\def\@makefntext#1{\@makefntextOring{\baselineskip=15pt#1}}
\makeatother

\definecolor{navyblue}{rgb}{0.0, 0.0, 0.5}
\definecolor{green}{rgb}{0.2, 0.5, 0.2}

\hypersetup{
    colorlinks=true,
    citecolor=navyblue,
    linkcolor=navyblue,
    urlcolor=navyblue,
}

\usepackage{titlesec}

\usepackage{etoolbox}
\patchcmd{\thanks}{#1}{\protect\doublespacing}{}{}

\DeclareFontShape{OT1}{cmr}{m}{n}{<->cmr10}{}
\usepackage{fix-cm}

\newcommand{\circledtext}[1]{%
  \tikz[baseline=(char.base)]{
    \node[
      shape=circle,
      draw,
      inner sep=0pt,
      minimum size=1.6em
    ] (char) {\scriptsize #1};
  }%
}

\title{\fontsize{19}{23}\selectfont Information Aggregation and Social Networks: Responsiveness and Overturning\thanks{\protect 
We thank Xiaoyu Chen, Benjamin Golub, Masaki Miyashita, Ryo Shirakawa, and Yuichi Yamamoto for their constructive comments.
All remaining errors are our own.
We acknowledge the financial support from the JSPS KAKENHI Grant 26KJ1194 (Noguchi) and 24KJ0100 (Sato).
}}
\author{%
\makebox[\textwidth][c]{%
\large
Shinpei Noguchi\thanks{%
Graduate School of Economics, Hitotsubashi University, 2-1 Naka, Kunitachi, Tokyo 186-8601, Japan.
Email: \url{ed225008@g.hit-u.ac.jp}.%
}
\hspace{1em}
Hiroto Sato\thanks{%
Department of Economics, Nagoya University, Furo-cho, Chikusa-ku, Nagoya 464-8601, Japan.
Email: \url{sato.hiroto.s9@f.mail.nagoya-u.ac.jp}.%
}
\hspace{1em}
Konan Shimizu\thanks{%
Faculty of Economics, Keio University, 2-15-45 Mita, Minato-ku, Tokyo 108-8345, Japan.
Email: \url{shimizu-konan@keio.jp}.%
}
\hspace{1em}
Tohya Sugano\thanks{%
Graduate School of Economics, The University of Tokyo, 7-3-1 Hongo, Bunkyo-ku, Tokyo 113-0033, Japan.
Email: \url{sugano-tohya1011@g.ecc.u-tokyo.ac.jp}.%
}%
}%
}
\date{\today}

\begin{document}

\begin{titlepage}
\maketitle
\begin{abstract}
This paper studies how network structure affects information aggregation in social learning. Agents sequentially choose actions based on private signals and observations of neighbors’ actions. Comparing a focal agent’s expected payoff across networks at a finite period, we show that a network is uniquely optimal for some informational environment if and only if the focal agent observes every predecessor. This characterization reveals that which network performs best depends critically on the informational environment. We then revisit two important implications of the characterization through transparent constructions: the {\it star network} can uniquely outperform every alternative under binary signals, while the {\it complete network} can do so with richer signals. These constructions highlight a trade-off between the {\it responsiveness effect} and the {\it overturning effect}: sparse networks facilitate information aggregation by preserving the responsiveness of actions to private signals, whereas dense networks facilitate information aggregation by revealing extreme information that overturns existing public beliefs.
\end{abstract}

\bigskip

\setcounter{page}{0}
\thispagestyle{empty}
\end{titlepage}


\newcolumntype{C}[1]{>{\centering\arraybackslash}p{#1}}

\setlength{\abovedisplayskip}{5pt}
\setlength{\belowdisplayskip}{5pt}


\section{Introduction} \label{sec:introduction}
Many economic decisions are made under uncertainty. 
In making these decisions, individuals often learn from others.
For example, people deciding whether to adopt a new technology may infer its quality from earlier adoption choices.\footnote{This example is inspired by \citet{acemoglu2011bayesian}; see also \citet{foster1995learning} and \citet{munshi2004social} for studies of social learning in technology adoption. For an organizational application involving a manager who relies on sequential recommendations from experts, see \citet{slezak2000effect}.} 
What individuals learn from others depends on whom they can observe. 
Social networks shape these observation patterns and thus influence how dispersed information is transmitted and aggregated. 
How do network structures affect information aggregation? What kinds of information can different networks effectively aggregate? 

Understanding these questions has become increasingly important as the digitization of economic activity has made others’ behavior more widely observable and reshaped network structures. 
Experimental evidence further illustrates the practical relevance of social observation. 
On the one hand, social influence can weaken the wisdom of crowds or bias aggregate evaluations \citep{lorenz2011social,muchnik2013social}.
On the other hand, social observation can improve collective accuracy and consumer satisfaction \citep{becker2017network,cai2009observational}. 
These findings motivate assessing the consequences of recent transformations in social observation for information aggregation.

To investigate these questions, we consider a simple model of social learning on networks. 
A finite number of agents move sequentially and choose one of two actions, a high action and a low action. 
The unknown state is either high or low, and agents may hold different priors. Each agent receives payoff one if her action matches the state and zero otherwise.
Before choosing an action, each agent receives a private signal drawn from an information structure that depends on the underlying state. 
In addition, each agent observes the actions taken by a subset of preceding agents, as specified by a network structure. 

In this model, we evaluate the efficiency of information aggregation across network structures by comparing a focal agent's expected payoff at a given period. 
In a technology adoption setting, this payoff measures how reliably earlier adoption decisions reveal the technology's quality to a later decision maker. 
In an organization, the focal agent may instead be a manager whose decision, informed by experts' sequential recommendations, determines the organization's ultimate outcome. 

We first focus on which network structures can be most effective at aggregating information.
An \emph{informational environment} consists of profiles of priors and information structures.
We call a network \emph{rationalizable} if there exists an informational environment under which it gives the focal agent a strictly higher expected payoff than every alternative network.
Theorem \ref{thm:characterization} shows that a network is rationalizable if and only if the focal agent observes every predecessor.
Necessity follows because she can always ignore an additional observation; sufficiency implies that, once she observes all predecessors, any pattern of links among them can be uniquely optimal.
Thus, which network is most effective depends critically on the informational environment.
As an extension, Proposition \ref{prop:welfare-characterization} shows that the same characterization applies when networks are evaluated by social welfare, defined as any fixed positively weighted sum of agents' expected payoffs.

An implication is that greater network connectivity has an ambiguous effect on social information aggregation.
An additional link is beneficial to the observing agent, who can ignore the additional action.
However, by changing how she combines private and social information, the link also changes what her own action reveals to later agents.
Consequently, the broader access to others' behavior enabled by more connected networks can worsen social information aggregation: a more connected network need not transmit more information.
This gap between the private and social value of connectivity reflects an informational externality and highlights the potential importance of shaping information flows.

Theorem \ref{thm:characterization} shows that different networks can be optimal under different informational environments, but it does not reveal the economic mechanisms that connect an informational environment to the network that is optimal under it.
To make these mechanisms more transparent, we focus on informational environments in which agents share a symmetric common prior and a common information structure, and revisit two important networks at opposite ends of connectivity.

We first revisit the {\it star network}, one of the sparsest networks in terms of connections among predecessors: the focal agent observes the actions taken by all preceding agents, while the other agents do not observe anyone else and rely only on their own private signals.
Theorem \ref{thm: star} provides a simple binary information structure under which the expected payoff of the focal agent in the star network is strictly higher than that in any other network.

The intuition behind Theorem \ref{thm: star} is that disconnected networks can preserve the responsiveness of agents’ actions to their own private information.
When an agent observes predecessors’ actions, she may put more weight on these observations and ignore her own information.
Thus, connected networks may generate information cascades.
By limiting the informational interdependence generated through social learning, disconnected networks prevent society from placing excessive weight on early actions and thereby mitigate information cascades.
We refer to this mechanism as the {\it responsiveness effect}.\footnote{
This terminology is related to \citet{ali2018herding,ali2018role}, who use the term responsiveness as a property of the decision problem, capturing the extent to which optimal actions vary with beliefs. In our setting, by contrast, responsiveness is a property of the network structure, capturing the extent to which equilibrium actions remain sensitive to agents’ own private signals under a given network.
Note that because our model has binary actions, our decision problem is not responsive in Ali’s sense.}
This mechanism is related to mechanisms such as the "sacrificial lambs effect," which has been extensively discussed in the literature.
The star network performs well when maintaining responsiveness to dispersed private information is relatively important for information aggregation.

The responsiveness effect alone may suggest that the star network is always optimal, but Theorem \ref{thm:characterization} also implies that the opposite extreme can be uniquely optimal under a suitable informational environment.
We make this complementary implication transparent for the {\it complete network}, in which every agent is connected to all preceding agents.
This network has been widely studied in the social learning literature.
Theorem \ref{thm: complete} provides an interpretable information structure under which the expected payoff of the focal agent in the complete network is strictly higher than that in any other network. 
Proposition \ref{prop: pareto dominate} further shows that, under the same information structure, the complete network weakly improves every agent's expected payoff and strictly improves the focal agent's payoff relative to any other network.

The intuition behind Theorem \ref{thm: complete} is not simply that more connections provide agents with more observations, but that they make richer patterns of actions interpretable by later agents.
The key feature of the information structure constructed in Theorem \ref{thm: complete} is that, among the signals indicating the high state and among those indicating the low state, some are highly informative while others are less informative.
Consequently, an agent who receives a highly informative signal may overturn the inference generated by a long sequence of identical actions.
To see this concretely, consider a history in which all agents have taken the high action up to some point, but some agent then takes the low action.
This behavior suggests that the agent's signal is informative enough to overturn the history by choosing the low action.
Since this overturning becomes part of the public history under the complete network, subsequent agents can draw this inference.
Suppose that at some later point another agent overturns the history again by taking the high action.
This suggests that she received an even more informative signal indicating the high state.
Repeating this reasoning, successive overturnings enable later agents to infer increasingly informative signals.
The complete network is especially effective in making this inference possible because every action becomes part of the public history observed by all subsequent agents.
We call this mechanism the {\it overturning effect}.\footnote{This terminology is related to \citet{smith2000pathological}, who emphasize the role of sufficiently strong signals in overturning accumulated public beliefs.}
In environments where this effect is strong, dense observation can improve information aggregation by enabling a prevailing inference to be corrected by overturning rather than merely amplified.

Thus, Theorems \ref{thm: star} and \ref{thm: complete} make implications of Theorem \ref{thm:characterization} transparent and reveal how the relationship between informational environment and network structure operates through the responsiveness effect and the overturning effect.
These mechanisms also provide a theoretical lens for contrasting experimental phenomena documented in the literature, even if those phenomena arise through different underlying forces. 
The responsiveness effect offers one mechanism through which social observation can worsen information aggregation \citep[see, e.g.,][]{lorenz2011social,muchnik2013social}, whereas the overturning effect offers one mechanism through which social observation can improve information aggregation \citep[see, e.g.,][]{becker2017network,cai2009observational}. 
Our analysis thus shows how analogous patterns can arise from observational learning under suitable informational environments and network structures.

What features of the information structure determine which of these two effects dominates?
This question is central because the responsiveness effect favors sparse networks, whereas the overturning effect favors dense networks.
Understanding this trade-off is therefore essential for determining when sparse networks aggregate information more effectively than dense networks, and vice versa.
Section \ref{sec: star vs complete} addresses this question by comparing networks within a class that includes both the star and complete networks, under specific families of information structures.

In Section \ref{sec: mixture}, we focus on a tractable class of information structures in which, for each state, there are two types of signals indicating that state: ordinary signals and conclusive signals that perfectly reveal the state.
This environment clarifies the role of highly informative signals by focusing on the extreme one, and yields a sharp characterization of the trade-off between the responsiveness effect and the overturning effect.
We compare a class of networks in which an initial group of agents (guinea pigs) act independently and subsequent agents observe all preceding actions as in \citet{sgroi2002optimizing}.
Proposition \ref{prop: star vs complete} establishes that there exists a sequence of cutoffs in the probability of conclusive signals: for an odd number of agents, the optimal network moves from the star network through intermediate networks to the complete network as this probability increases.
It also shows that, as the accuracy of the ordinary signal falls, each cutoff decreases, meaning that networks with fewer guinea pigs become optimal at lower probabilities of conclusive signals.

In Section \ref{sec: numerical}, we next numerically examine whether the insight from the theoretical analysis in Section \ref{sec: mixture} extends beyond that specific environment.
We consider four parameterized families of private-belief distributions: truncated normal, symmetric Beta, Beta mixture, and contaminated log-likelihood-ratio normal.
The latter two are mixture distributions that generate W-shaped densities with substantial mass both near the prior and near the endpoints.
The numerical analysis points to the same insight as the theoretical results: networks with fewer guinea pigs tend to perform better when highly informative signals occur with larger probability and when ordinary signals are less accurate.

These analyses identify a key determinant of which side of the trade-off prevails.
When highly informative signals are rare and ordinary signals are accurate enough, the value of preserving responsiveness tends to dominate; when highly informative signals are more likely and ordinary signals are less accurate, the value of making highly informative signals identifiable through overturning histories becomes more important.
Thus, the relative performance of networks is governed not only by network connectivity itself, but also by how the information structure distributes probability between highly informative and less informative signals.
In particular, it depends crucially on the relative thickness of the tails of the private signal distribution.

\subsection{Literature}
Our model belongs to the classical literature on social learning pioneered by \citet{banerjee1992simple}, \citet{bikhchandani1992theory}, and \citet{smith2000pathological}.\footnote{See the recent survey by \citet{bikhchandani2024information} for a broad overview of the social learning literature.}
A central insight of this literature is that the long-run consequences of social learning, whether information cascades arise or asymptotic learning occurs, depend critically on the informativeness of private signals, particularly on the tail behavior of the signal distribution.\footnote{Related contributions on learning speed include \citet{hann2018speed} and \citet{rosenberg2019efficiency}.}

Beginning with \citet{sgroi2002optimizing}, \citet{gale2003bayesian}, and \citet{ccelen2004observational}, a large literature has developed on social learning in more general networks (see \citet{golub2016} for a survey).
Important contributions include \citet{acemoglu2011bayesian}, \citet{lobel2015information}, \citet{mossel2015strategic}, \citet{arieli2019multidimensional,arieli2021general}, \citet{dasaratha2023learning}, and \citet{kartik2024beyond}.
A central focus of this literature is whether agents eventually learn the relevant state and how information diffuses through networks in the long run, under different assumptions regarding the timing of actions and the evolution of the state.

Within this literature, our paper is most closely related to work that compares the efficiency of information aggregation across network structures.
\citet{sgroi2002optimizing} studies the optimal use of "guinea pigs" in social learning and shows that withholding some agents from observing predecessors can improve subsequent information aggregation.
Under Gaussian information structures, \citet{dasaratha2019aggregative} focus on generation networks in which agents observe multiple neighbors but not their common predecessors, thereby confounding the sources of information. 
They quantify how such information confounding reduces the efficiency of information aggregation under each network.
Complementing these studies, we show in a general framework that the relative performance of networks depends critically on the informational environment.
We clarify the underlying mechanism by identifying a trade-off between a responsiveness effect closely related to the sacrificial-lamb logic emphasized by \citet{sgroi2002optimizing} and an overturning effect related to the role of strong signals emphasized by \citet{smith2000pathological}.
Our contribution is to integrate these previously dispersed insights into a unified framework, showing that the apparent advantages of sparse and dense networks arise from a common underlying trade-off and that the informational environment determines which force dominates.

Our paper is also close in spirit to a broad literature that asks how network structure shapes information aggregation in models of non-Bayesian learning. 
Early contributions include \citet{ellison1993rules,ellison1995word}, \citet{bala1998learning,bala2001conformism}, and \citet{demarzo2003persuasion} followed by \citet{golub2010naive,golub2012homophily}, \citet{jadbabaie2012nonbayesian}, and \citet{molavi2018theory}, among others.\footnote{Related work examines misspecified inference in social learning \citep{eyster2010naive,bohren2016informational,frick2020misinterpreting,bohren2021learning,frick2023belief,arieli2025hazards}.} 
Most closely, \citet{dasaratha2020naive} show that dense networks can amplify the overcounting of correlated information; see \citet{dasaratha2021experiment} for experimental evidence. 
By contrast, our paper establishes that, even with Bayesian agents, network density has an ambiguous effect on information aggregation because of a trade-off between the responsiveness and overturning effects arising from an informational externality.

\section{Illustrative Example}
This section uses a simple three-agent environment to illustrate the paper's main results and the trade-off between the responsiveness and overturning effects. 
Readers who prefer to proceed directly to the formal analysis may skip this section and turn to Section \ref{sec: model}.

We compare agent $3$'s expected payoff across networks.
Since agents move sequentially, agent $1$ has no predecessor, agent $2$ may or may not observe agent $1$, and agent $3$ may observe any subset of agents $1$ and $2$.
There are therefore eight possible networks.
Among these, we focus in particular on two networks shown in Figure \ref{fig:N3network_comparison}.

\begin{figure}[ht]
\centering
\begin{subfigure}{0.45\textwidth}
\centering
\begin{tikzpicture}[scale=1]
\node (A1) at (0,0) {$\circledtext{1}$};
\node (A2) at (0,-1.6) {$\circledtext{2}$};
\node (A3) at (3,-0.8) {$\circledtext{3}$};

\draw[->, very thick] (A1) -- (A3);
\draw[->, very thick] (A2) -- (A3);
\end{tikzpicture}
\caption{Star network}
\label{fig:N3star}
\end{subfigure}
\hfill
\begin{subfigure}{0.45\textwidth}
\centering
\begin{tikzpicture}[scale=1]
\node (A1) at (0,0) {$\circledtext{1}$};
\node (A2) at (2,0) {$\circledtext{2}$};
\node (A3) at (4,0) {$\circledtext{3}$};

\draw[->, very thick] (A1) -- (A2);
\draw[->, very thick] (A1) to[bend left=60] (A3);
\draw[->, very thick] (A2) -- (A3);
\end{tikzpicture}
\caption{Complete network}
\label{fig:N3complete}
\end{subfigure}
\caption{Star and complete networks}
\label{fig:N3network_comparison}
\end{figure}
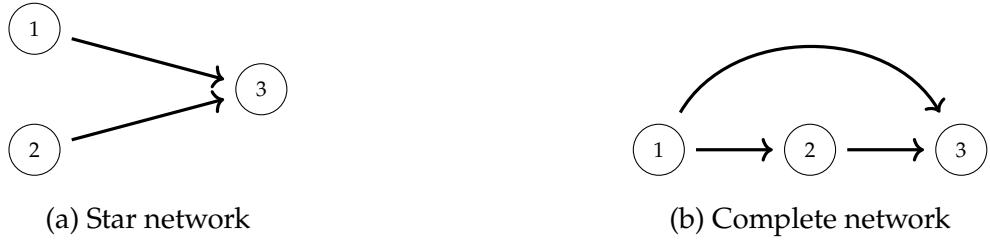

The first is the {\it star network}, in which agent $3$ observes both predecessors, while agent $2$ does not observe agent $1$.
In this network, agents $1$ and $2$ act without observing any previous action, so agent $3$ observes two actions that are not influenced by earlier social observations.
Its formal definition is provided in Definition \ref{def:star} in Section \ref{sec: results}.

The second is the {\it complete network}, in which every agent observes all predecessors.
In this network, agent $2$ observes agent $1$, and agent $3$ observes both previous actions.
Its formal definition is provided in Definition \ref{def:complete} in Section \ref{sec: results}.
This is the canonical network studied in the social learning literature.

The remaining six networks, shown in Figure \ref{fig:N3dominated}, do not allow agent $3$ to observe both predecessors. Each is weakly dominated under any information structure by the star or complete network, because agent $3$ can ignore any additional observations.
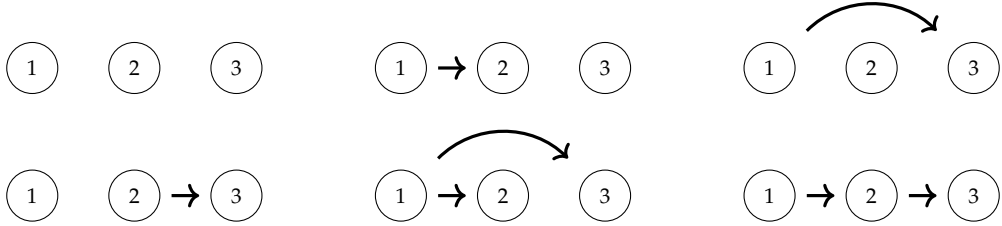
\begin{figure}[ht]
\centering
\begin{tikzpicture}[scale=0.75]
\node (A1) at (0,0) {$\circledtext{1}$};
\node (A2) at (1.8,0) {$\circledtext{2}$};
\node (A3) at (3.6,0) {$\circledtext{3}$};
\node (B1) at (6.5,0) {$\circledtext{1}$};
\node (B2) at (8.3,0) {$\circledtext{2}$};
\node (B3) at (10.1,0) {$\circledtext{3}$};
\draw[->, very thick] (B1) -- (B2);
\node (C1) at (13,0) {$\circledtext{1}$};
\node (C2) at (14.8,0) {$\circledtext{2}$};
\node (C3) at (16.6,0) {$\circledtext{3}$};
\draw[->, very thick] (C1) to[bend left=45] (C3);

\node (D1) at (0,-2.25) {$\circledtext{1}$};
\node (D2) at (1.8,-2.25) {$\circledtext{2}$};
\node (D3) at (3.6,-2.25) {$\circledtext{3}$};
\draw[->, very thick] (D2) -- (D3);
\node (E1) at (6.5,-2.25) {$\circledtext{1}$};
\node (E2) at (8.3,-2.25) {$\circledtext{2}$};
\node (E3) at (10.1,-2.25) {$\circledtext{3}$};
\draw[->, very thick] (E1) -- (E2);
\draw[->, very thick] (E1) to[bend left=45] (E3);
\node (F1) at (13,-2.25) {$\circledtext{1}$};
\node (F2) at (14.8,-2.25) {$\circledtext{2}$};
\node (F3) at (16.6,-2.25) {$\circledtext{3}$};
\draw[->, very thick] (F1) -- (F2);
\draw[->, very thick] (F2) -- (F3);
\end{tikzpicture}
\caption{Networks in which agent $3$ does not observe both predecessors.}
\label{fig:N3dominated}
\end{figure}

As a benchmark, consider a hypothetical setting in which agents observe their neighbors' signal realizations rather than their actions.
In that setting, the star network and the complete network provide the same information to agent $3$.
In what follows, we show that the comparison becomes nontrivial when agents observe actions rather than signals.
This highlights the fact that the informational content of an observed history depends on how earlier agents formed their actions, which is itself determined by the network.

We first focus on a binary information structure as follows:
\begin{example}\label{example: binary}
    Suppose that the state $\theta$ is either $L$ or $H$, each occurring with equal probability under the common prior, and that it is optimal for an agent to choose action $1$ if her posterior belief assigns probability strictly larger than $1/2$ to $\theta=H$, and action $0$ otherwise.\footnote{Here, we assume that the agent chooses action $0$ when she is indifferent between the two actions.}
    Consider a binary information structure $\Pi^{B}=(\{l,h\},\pi^B)$ with
    \[
    \pi^B(h\mid H)=\pi^B(l\mid L)=9/10, \qquad \pi^B(l\mid H)=\pi^B(h\mid L)=1/10.\]
     
    In the star network, agents $1$ and $2$ choose their actions solely based on their own private signals, so each of them chooses action $1$ if and only if her signal is $h$. 
    Thus, $a_1$ and $a_2$ perfectly reveal $s_1$ and $s_2$, respectively, and agent $3$ chooses her action as if she directly observed $s_1$, $s_2$, and $s_3$.
    
    In the complete network, agent $2$ observes $a_1$ and can infer $s_1$. 
    Thus, she chooses action $1$ if and only if both $s_1$ and $s_2$ are $h$, so agent $2$'s action pools $s_2=h$ together with $s_2=l$ into the same action $0$ when $s_1=l$. 
    Hence, with probability $1/2$, agent $3$ cannot recover $s_2$ from the history $(a_1,a_2)$ in the complete network.
    In particular, when $(s_1,s_2,s_3)=(l,h,h)$, action $1$ maximizes agent $3$'s expected payoff given this signal profile, and agent $3$ indeed chooses action $1$ under the star network, whereas she follows the predecessors and chooses action $0$ under the complete network.
    Thus, under $\Pi^B$, agent 3 is better off in the star network than in the complete network. \qed
\end{example}

In Example \ref{example: binary}, the star network perfectly reveals the predecessors' private signals because the private signals and the actions have a one-to-one relationship.
By contrast, under the complete network, the information cascade may immediately occur beginning from agent $2$.
Thus, there is informational loss in aggregation of private signals.
In this sense, compared to the complete network, the star network can preserve the responsiveness of predecessors' actions to their own private signals relatively well by eliminating the interdependence across the predecessors' actions.
We refer to this effect as the {\it responsiveness effect} throughout the paper.

By the above argument, one may conjecture that the star network is always better than the complete network.
This conjecture is false.
Now, we consider a slight modification of the previous binary information structure as follows:
\begin{example}\label{example: ternary}
    Consider the same decision problem as in Example \ref{example: binary}, and information structure $\Pi^\ast=(\{l,h,h^\ast\},\pi^\ast)$ defined by
   \begin{alignat*}{3}
    &\pi^\ast(l\mid H)=1/10,
    &&\qquad \pi^\ast(h\mid H)=4/5,
    &&\qquad \pi^\ast(h^\ast\mid H)=1/10,\\
    &\pi^\ast(l\mid L)=2/5,
    &&\qquad \pi^\ast(h\mid L)=3/5,
    &&\qquad \pi^\ast(h^\ast\mid L)=0.
    \end{alignat*}
    Under this information structure, $h^\ast$ perfectly reveals that $\theta=H$, so an agent chooses action $1$ if she knows that some agent receives $h^\ast$.
    Also, we can verify that agent $3$ chooses action $0$ if she receives $l$ herself and knows that both agents $1$ and $2$ receive either $h$ or $h^\ast$.\footnote{
    Since the prior assigns equal probability to each state, action 0 is optimal if the likelihood ratio of her observation does not exceed 1.
    The following calculation shows that this is indeed the case:
    \[
    \left( \frac{\pi^\ast(\{h,h^\ast \}\mid H)}{\pi^\ast(\{h,h^\ast \}\mid L)} \right)^2 
     \frac{\pi^\ast(l\mid H)}{\pi^\ast(l\mid L)} 
     =
     \left( \frac{9/10}{3/5} \right)^2 \frac{1/10}{2/5}
     <1.
    \]
    }

    In the star network, each of agents 1 and 2 chooses action 1 if and only if her private signal is $h$ or $h^\ast$. 
    Thus, the history reveals whether the agents receive $l$ and never reveals $h^\ast$. 
    If $(a_1, a_2)=(1,1)$, agent 3 chooses action 1 unless she receives $l$. If agent 3 observes at least one action 0, she follows and chooses action 0 unless she receives $h^\ast$.

    In the complete network, agent $2$ observes $a_1$ and can infer whether $s_1$ is $l$ or not.
    When $a_1 =1$, she chooses action 1 unless $s_2$ is $l$; when $a_1=0$, she chooses action 0 unless $s_2$ is $h^\ast$.
    Thus, $a_2$ reveals whether $s_2$ is $l$ or not in the former case, and whether it is $h^\ast$ or not in the latter case.
    Therefore, when $a_1=1$, the history $(a_1,a_2)$ provides the same information to agent 3 and she chooses the action in the same way as in the star network. 
    However, when $a_1=0$, the history provides different information from the star-network case. When $(a_1,a_2)=(0,1)$, agent 3 infers that $s_2=h^\ast$, which perfectly reveals that $\theta = H$, and thus she chooses the optimal action, action 1, regardless of her signal.
    When $(a_1,a_2)=(0,0)$, agent 3 chooses the action in the same way as in the star network.
    Thus, under $\Pi^\ast$, agent 3 is better off in the complete network than in the star network. \qed
\end{example}

Under the ternary information structure in Example \ref{example: ternary}, the complete network can outperform the star network because actions may reveal the strength of private information through overturning. 
Specifically, the switch in the observed history from the low action to the high action thus reveals strong evidence in favor of the high state, because the agent takes the high action despite observing a history that suggests the low state.
This inference is unavailable in the star network. 
Since agent $2$ does not observe agent $1$, the high action only indicates that her private signal is positive; it does not distinguish the ordinary signal $h$ from the strong signal $h^\ast$. 
Thus, while the star network preserves the responsiveness of predecessors' actions to their own signals, the complete network can make decisive signals identifiable through the pattern of actions. 
We refer to this force as the {\it overturning effect}.

These examples clarify why neither sparse nor dense observation is uniformly optimal. 
With binary signals, the main problem is that social observation makes later actions less responsive to private signals; the star network avoids this problem by keeping predecessors' actions independent and directly interpretable. 
With ternary signals, however, social observation can help identify strong private information: when a later action overturns an earlier one, the pattern of actions reveals that the later agent must have received a sufficiently strong signal. 
Thus, the same interdependence among actions that weakens responsiveness can also make strong information identifiable and transmissible. 
The comparison between the star and complete networks therefore reflects a trade-off between the responsiveness effect and the overturning effect, and which network performs better depends on the underlying information structure.


\section{Model} \label{sec: model}
There are $N$ agents, labeled $1,2,\dots,N$, where $N\geq 3$.
They choose their actions sequentially.
Let $\Theta=\{L,H\}$ be the state space.
Each agent $i$ assigns prior $\mu_i\in(0,1)$ to state $H$.
Let $A=\{0,1\}$ be the action set for each agent. 
The periods are discrete, indexed by $t=0,1,\dots, N$, and each agent $i$ takes an action at period $i$ from the action set $A$.
Each agent faces a common decision problem $\mathcal{D}=(A,u)$, where $u:A\times \Theta\to\mathbb{R}$ is a payoff function such that $u(a,\theta)=1$ if $(a,\theta)=(0,L)$ or $(1,H)$ and $u(a,\theta)=0$ otherwise.

A {\it network} $\mathcal{N}=(\mathcal{N}_i)_{i=1}^N$ specifies, for each agent, which predecessors' actions she observes.
Specifically, before choosing an action, agent $i$ observes the actions taken by agents in $\mathcal{N}_i\subseteq \{1,2,\dots,i-1\}$.\footnote{Thus, we focus on deterministic networks rather than stochastic networks.}

Each agent $i$ also receives a {\it private signal} $s_i\in S_i$, drawn from the {\it information structure} $\Pi_i=(S_i,\pi_i)$, where $S_i$ is the set of signals and $\pi_i:\Theta\to\Delta(S_i)$ is the disclosure rule.
Conditional on the state, private signals are independent: $\mathbb{P}(s_1,\ldots,s_N\mid\theta)=\prod_{i=1}^N\pi_i(s_i\mid\theta)$.
Write $\mu=(\mu_i)_{i=1}^N$ for the prior profile and $\Pi=(\Pi_i)_{i=1}^N$ for the profile of information structures. We refer to $(\mu,\Pi)$ as the {\it informational environment}. The prior profile, the information structures, and the network are common knowledge, and agents update from their respective priors according to Bayes' rule.

Given agent $i$'s information structure $\Pi_i=(S_i,\pi_i)$, for each $s_i\in S_i$, we define $LR_i(s_i)=\pi_i(s_i\mid H)/\pi_i(s_i\mid L)$ as the likelihood ratio induced by private signal $s_i$.

The timing of the game is as follows.
In period $0$, nature selects the true state, which then remains fixed for the duration of the game.
In each period $i$, agent $i$ first observes the (partial) history, consisting of the actions of agents in $\mathcal{N}_i$.
Furthermore, agent $i$ receives a private signal $s_i\in S_i$ drawn from $\pi_i$. Then, based on these observations, the agent chooses an action from $A$.

Given the network $\mathcal{N}$ and the informational environment $(\mu,\Pi)$, the strategy of agent $i$ is denoted by $\sigma_i: A^{|\mathcal{N}_i|}\times S_i \to \Delta(A)$. When $\sigma_i$ is a pure strategy, we identify it with the corresponding map $\sigma_i: A^{|\mathcal{N}_i|}\times S_i \to A$. By a slight abuse of notation, given a pure strategy profile $(\sigma_1, \dots, \sigma_i)$, we write agent $i$'s action under signal profile $(s_1, \dots, s_i) \in \prod_{j=1}^iS_j$ as $\sigma_i(s_1, \dots, s_i)$, which is defined, for each $i\geq 2$, recursively by,
\begin{align*}
    \sigma_i(s_1, \dots, s_i) := \sigma_i\bigl((\sigma_j(s_1, \dots, s_j))_{j \in \mathcal{N}_i},\, s_i\bigr).
\end{align*}

Given the network $\mathcal{N}$, informational environment $(\mu,\Pi)$, and the strategy profile $\bm{\sigma}=(\sigma_{i})_{i =1}^N$, let $\alpha_{\leq i}^{\theta}(\mathcal{N},\mu,\Pi,\bm{\sigma})\in \Delta(A^{i})$ denote the distribution of actions taken by agents in $1,2,\dots,i$ when the state is $\theta$, that is,
\begin{align*}
 \alpha_{\leq i}^{\theta}(\bm{a}\mid \mathcal{N},\mu,\Pi,\bm{\sigma})=\sum_{(s_{1},\dots,s_{i}) \in \prod_{k=1}^iS_k}\prod_{k=1}^{i}\sigma_{k}(a_{k}\mid (a_j)_{j\in \mathcal{N}_k},s_{k})\pi_k(s_{k}\mid \theta).
\end{align*}
Then, let $\alpha_{i}^{\theta}(\mathcal{N},\mu,\Pi,\bm{\sigma})\in \Delta(A)$ be the distribution of actions taken by agent $i$ when the state is $\theta$, that is,
\begin{align*}
 \alpha_{i}^{\theta}(a\mid \mathcal{N},\mu,\Pi,\bm{\sigma})=\sum_{(a_{1}',\dots,a_{i-1}')\in A^{i-1}}\alpha_{\leq i}^{\theta} (a_{1}',\dots,a_{i-1}',a \mid  \mathcal{N},\mu,\Pi,\bm{\sigma}),
\end{align*}
where $\alpha_{i}^{\theta}(\mathcal{N},\mu,\Pi,\bm{\sigma})$ does not depend on the strategies of agents after $i$.
We say that the strategy profile $\bm{\sigma}^{*}$ is a Bayes-Nash equilibrium (hereafter referred to simply as an equilibrium) under $(\mathcal{N},\mu,\Pi)$ if $\sigma_i$ is a map from $A^{| \mathcal{N}_i|}\times S_i$ to $\Delta(A)$ and
\begin{align*}
   \mathbb{E}_{\mu_i}\left[\sum_{a\in A} \alpha_{i}^{\theta}(a\mid\mathcal{N},\mu,\Pi,\bm{\sigma}^*)u(a,\theta)\right]\geq \mathbb{E}_{\mu_i}\left[\sum_{a\in A} \alpha_{i}^{\theta}(a\mid\mathcal{N},\mu,\Pi,\bm{\sigma}^*_{-i},\sigma_i)u(a,\theta)\right],
\end{align*}
for all $\sigma_{i}$ and $i$.

Since each agent’s payoff does not depend on the actions of subsequent agents, an equilibrium can be obtained by considering the expected payoff maximization problem in each period separately. Therefore, an equilibrium always exists. 
However, it is not necessarily unique. 
This is because the equilibrium depends on which action an agent chooses when both actions are optimal. 
To eliminate the issue of equilibrium selection, throughout this paper, we assume the following tie-breaking rule: whenever an agent is indifferent between the two actions in terms of expected payoff, the agent chooses action $0$.

Under this assumption, the equilibrium is uniquely determined and is pure strategy. 
We denote agent $i$'s equilibrium strategy under $(\mathcal{N},\mu,\Pi)$ by $\sigma_i^\ast(\cdot,s_i\mid \mathcal{N},\mu,\Pi)$ and her expected payoff by
\begin{align*}
V_i(\mathcal{N},\mu,\Pi):=(1-\mu_i)\alpha_i^L(0\mid\mathcal{N},\mu,\Pi,\bm{\sigma}^\ast)+\mu_i\alpha_i^H(1\mid\mathcal{N},\mu,\Pi,\bm{\sigma}^\ast).
\end{align*}

\section{Main Analysis} \label{sec: results}
This section analyzes how network structures affect the efficiency of information aggregation.
In particular, we evaluate each network by agent $N$'s expected payoff.
This criterion can be interpreted in two ways.
First, in environments with only a finite number of agents, agent $N$ represents the ultimate decision maker who can potentially benefit from all information generated by the preceding agents.
Thus, her payoff captures how effectively each network transmits and aggregates that information.
Second, in environments with infinitely many agents, agent $N$ can instead be viewed as a representative finite horizon.
Under this interpretation, we can compare the speed and quality of information aggregation across networks period by period.
This contrasts with the existing literature, which typically focuses on asymptotic learning outcomes.

\subsection{Benchmark}
For the benchmark, we now consider the hypothetical case in which agents observe their neighbors' private signal realizations directly.\footnote{Here, we assume that agents observe only the private signals of their neighbors, but not their actions. If agents could observe both private signals and past actions, the argument below would no longer apply.}
\begin{observation}
    Suppose that agents observe the private signal realizations of their neighbors rather than their past actions.
    If two networks $\mathcal{N}$ and $\mathcal{N}'$ satisfy $\mathcal{N}_{N}=\mathcal{N}'_{N}$, then they induce the same expected payoff for agent $N$.
\end{observation}
Thus, in the observable-signal benchmark, once agent $N$'s neighborhood is fixed, the network among her predecessors is irrelevant for her payoff.
This is because the signals observed by agent $N$ do not depend on how her predecessors are connected to one another.
In contrast, in the observable-action setting, agent $N$ must infer private signals through actions chosen on the basis of each predecessor's private signal and the history she observes.
These actions are chosen for the predecessors' own payoffs but also transmit information to agent $N$.
The network among the predecessors therefore shapes this informational externality.

\subsection{Rationalizable Networks}
To investigate the implications of this informational externality, we begin by asking which networks can be effective under some informational environment.

\begin{definition}
A network $\widehat{\mathcal{N}}$ is \emph{rationalizable} if there exists an informational environment $(\mu,\Pi)$ such that $V_N(\widehat{\mathcal{N}},\mu,\Pi)>V_N(\mathcal{N},\mu,\Pi)$ for every $\mathcal{N}\neq\widehat{\mathcal{N}}$.
\end{definition}
We use \emph{rationalizable} in a strict sense throughout: the target must be the unique maximizer, rather than merely a weak maximizer, of agent $N$'s expected payoff.\footnote{
If rationalizability were instead defined using a weak maximizer, every network would be rationalizable. Consider an informational environment in which agent $N$'s private signal alone fully reveals the state. Under this informational environment, agent $N$ obtains payoff one under every network, and hence every network weakly maximizes her expected payoff.
}
Note also that the same informational environment $(\mu,\Pi)$ is used in every comparison with the fixed target network $\widehat{\mathcal{N}}$.

The following theorem provides the complete characterization.
\begin{theorem}
\label{thm:characterization}
A network $\widehat{\mathcal{N}}$ is rationalizable if and only if $\widehat{\mathcal{N}}_N=\{1,\ldots,N-1\}$.
\end{theorem}

Theorem \ref{thm:characterization} establishes a sharp boundary between links entering agent $N$ and links among her predecessors.
Every rationalizable network must contain all links entering agent $N$, since additional observations cannot reduce her payoff.
Beyond this requirement, however, rationalizability imposes no restriction on the predecessor network: every pattern of links among agents $1,\ldots,N-1$ can be made uniquely optimal by a suitable informational environment.
The following corollary sharpens this latter conclusion by isolating the effect of a single link among the predecessors.

\begin{corollary}
\label{cor:predecessor-link}
Fix $1\leq j<i\leq N-1$ and a network $\mathcal{N}^-$ with $j\notin\mathcal{N}^-_i$ and $\mathcal{N}^-_N=\{1,\ldots,N-1\}$.
Let $\mathcal{N}^{+}$ be defined by $\mathcal{N}^+_i=\mathcal{N}^{-}_i\cup\{j\}$ and $\mathcal{N}^+_{i^\prime}=\mathcal{N}^{-}_{i^\prime}$ for all $i^\prime \neq i$.
Then, there exist informational environments $(\mu^{+},\Pi^{+})$ and $(\mu^{-},\Pi^{-})$ such that $V_N(\mathcal{N}^{+},\mu^{+},\Pi^{+})>V_N(\mathcal{N}^{-},\mu^{+},\Pi^{+})$ and $V_N(\mathcal{N}^{-},\mu^{-},\Pi^{-})>V_N(\mathcal{N}^{+},\mu^{-},\Pi^{-})$.
\end{corollary}
Corollary \ref{cor:predecessor-link} shows that the effect of a link among the predecessors on agent $N$'s payoff can be either positive or negative, depending on the informational environment.
Although the link is weakly valuable to the agent who gains the observation, it changes how her private information and observation are encoded in her action and may therefore make that action more or less informative to later agents.
Thus, the link's private value to the observing agent and its informational effect on later agents need not be aligned.

We next outline the proof of Theorem~\ref{thm:characterization}.
Necessity follows immediately because agent $N$ can ignore additional observations.
The substantive challenge is sufficiency, whose proof proceeds in two stages.
First, with agent $N$ observing all predecessors, we construct the informational environment inductively over agents $i=2,\ldots,N-1$.
At each step, we choose agent $i$'s prior and two weak \emph{probe signals} that induce different actions at a suitable history under the target neighborhood but are pooled under every alternative neighborhood.
Rare \emph{test signals} for agent $N$ make this distinction strictly valuable without reversing earlier payoff comparisons.
Second, additional rare signals make every link entering agent $N$ strictly valuable while preserving these comparisons.
The following example illustrates the main idea of the first stage in more detail.
\begin{example}
\label{ex:four-agent-shield}
Consider the four-agent network in Figure~\ref{fig:four-agent-shield}, which is referred to as a {\it shield network} by \citet{eyster2014extensive}.
We construct an informational environment under which this target network gives agent $4$ a strictly higher expected payoff than both the complete network and the star network.
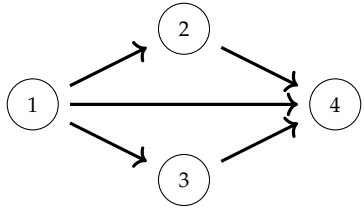
\begin{figure}[ht]
\centering
\begin{tikzpicture}
\node (1) at (0,0) {$\circledtext{1}$};
\node (2) at (2,1) {$\circledtext{2}$};
\node (3) at (2,-1) {$\circledtext{3}$};
\node (4) at (4,0) {$\circledtext{4}$};

\draw[->, very thick] (1) -- (2);
\draw[->, very thick] (1) -- (3);
\draw[->, very thick] (1) -- (4);
\draw[->, very thick] (2) -- (4);
\draw[->, very thick] (3) -- (4);
\end{tikzpicture}
\caption{Four-agent target network}
\label{fig:four-agent-shield}
\end{figure}
We specify the informational environment $(\mu,\Pi)$ as follows.
Agent $1$ has prior $\mu_1$ and a binary information structure $\Pi_1=(\{l_1,h_1\},\pi_1)$ defined by
\begin{alignat*}{3}
&\mu_1=1/2,
&&\quad \pi_1(h_1\mid H)=\pi_1(l_1\mid L)=2/3,
&&\quad \pi_1(l_1\mid H)=\pi_1(h_1\mid L)=1/3.
\end{alignat*}
She therefore follows her own signal, choosing the low action after $l_1$ and the high action after $h_1$.

Agents $2$ and $3$ have priors $\mu_2$ and $\mu_3$ and binary information structures $\Pi_i=(\{p_i^-,p_i^+\},\pi_i)$ for $i=2,3$, defined by
\begin{alignat*}{4}
&\mu_2=1/3,
&&\quad \pi_2(p_2^-\mid H)=1/4,
&&\quad \pi_2(p_2^+\mid H)=3/4,
&&\quad \pi_2(p_2^-\mid L)=\pi_2(p_2^+\mid L)=1/2,\\
&\mu_3=1/3,
&&\quad \pi_3(p_3^-\mid H)=9/20,
&&\quad \pi_3(p_3^+\mid H)=11/20,
&&\quad \pi_3(p_3^-\mid L)=\pi_3(p_3^+\mid L)=1/2.
\end{alignat*}
In the terminology of the proof, these are probe signals.
They are deliberately weak: without observing any predecessor, each agent chooses the low action after either signal realization.
After observing agent $1$'s high action alone, however, each chooses different actions after her two probe signals.

Agent $4$ has prior $\mu_4$ and an information structure $\Pi_4=(\{t_2,t_3,t^{-},t^{+}\},\pi_4)$ with
\begin{alignat*}{4}
&\mu_4=1/2,
&&\quad \pi_4(t_2\mid H)=1/20,
&&\quad \pi_4(t_3\mid H)=1/30,
&&\quad \pi_4(t^{+}\mid H)=11/12,\\
&
&&\quad \pi_4(t_2\mid L)=1/10,
&&\quad \pi_4(t_3\mid L)=1/10,
&&\quad \pi_4(t^{-}\mid L)=4/5.
\end{alignat*}
Here $\pi_4(t^{+}\mid L)=\pi_4(t^{-}\mid H)=0$.
The conclusive signals $t^{+}$ and $t^{-}$ perfectly reveal the state, and thus agent $4$ chooses the correct action under every network whenever either signal is realized.
Thus, differences in her expected payoff across networks arise only when she receives one of the relatively rare test signals $t_2$ and $t_3$.
These test signals make identifying the corresponding probe signals valuable to agent $4$ at suitable histories, as illustrated below.

Under the star network, neither agent 2 nor agent 3 observes anyone, so both cascade to the low action. Thus, agent 4 cannot identify $s_2$ or $s_3$ under the star network.

Under the complete network, agent 2 observes agent 1's action. 
After observing agent 1's low action, she cascades to the low action, but after observing agent 1's high action, she follows her own signal. 
Thus, agent 4 can identify $s_2$ when $s_1=h_1$. 
However, agent 4 cannot identify $s_3$. 
This is because, as can be verified, when agent 3 observes the history $(0,0)$ or $(1,0)$, she cascades to the low action, and when she observes $(1,1)$, she cascades to the high action.\footnote{Agent 2 chooses the low action regardless of her signal whenever she observes $a_1=0$, so the history $(0,1)$ never arises.}

By contrast, under the target network, agents 2 and 3 each observe only agent 1's action, and both follow their own signal when $a_1=1$. 
Therefore, when agent 1 chooses the high action, agent 4 obtains strictly more information from the history under the target network than under either the star or the complete network: she can identify both probe signals under the target network, whereas she cannot identify at least agent 3's probe signal under either rival network. Suppose $(s_1,s_2,s_4)=(h_1,p_2^+,t_3)$. 
Agent 4 can then identify that $s_1=h_1$ and $s_2=p_2^+$, and her optimal action depends on whether $s_3=p_3^-$ or $p_3^+$: identifying agent 3's signal at this event therefore gives her a strictly higher expected payoff. 
Also, when agent 1 chooses the low action, agents 2 and 3 cascade to the low action under every network, so agent 4 obtains the same information under all three networks. 
Therefore, agent 4's expected payoff under the target network is strictly higher than under both rival networks.

Two general observations underlie this comparison. 
First, under the target network agent 4 obtains more information than under any other network. In particular, there is a history at which she can identify every probe signal, whereas no such history exists under any other network. 
Second, at such a history, identifying agent $i$'s probe signal strictly raises agent 4's expected payoff when she receives the corresponding test signal $t_i$. 
By the same argument using these two observations, we can verify that the target network achieves a strictly higher expected payoff for agent 4 than any network under which agent 4 observes all three predecessors.

Showing that the target network gives agent 4 a strictly higher expected payoff than a network in which she does not observe all three predecessors is more technically involved, and requires perturbing the information structure specified above. 
Appendix \ref{app:detailed-four-agent} constructs such an information structure explicitly and shows that, under it, the target network gives agent 4 a strictly higher expected payoff than every other network. 
The proof behind Theorem \ref{thm:characterization} generalizes the two observations above and this perturbation argument to establish the same result for an arbitrary number of agents and any target network under which agent $N$ observes all predecessors. \qed
\end{example}

Theorem \ref{thm:characterization} and Corollary \ref{cor:predecessor-link} evaluate network performance in terms of agent $N$'s expected payoff.
We now turn from agent $N$'s payoff to social welfare. Fix $\lambda=(\lambda_1,\ldots,\lambda_N)\in\mathbb{R}_{++}^N$ and define $W_\lambda(\mathcal{N},\mu,\Pi):=\sum_{i=1}^N\lambda_iV_i(\mathcal{N},\mu,\Pi)$.
This welfare criterion includes discounted welfare as a special case, with $\lambda_i=\delta^{i-1}$ for some $\delta\in(0,1)$.
We define the welfare counterpart of rationalizability as follows.
\begin{definition}
A network $\widehat{\mathcal{N}}$ is \emph{socially rationalizable} if there exists an informational environment $(\mu,\Pi)$ such that $W_\lambda(\widehat{\mathcal{N}},\mu,\Pi)>W_\lambda(\mathcal{N},\mu,\Pi)$ for every $\mathcal{N}\neq\widehat{\mathcal{N}}$.
\end{definition}
Social rationalizability is defined relative to the fixed welfare weights $\lambda$, although we suppress this dependence in the terminology.
The rationalizing informational environment may depend on $\lambda$; a single environment need not rationalize the network under all positive welfare weights.
The following proposition provides the welfare counterpart of Theorem \ref{thm:characterization} and shows that the characterization is independent of the choice of $\lambda$.

\begin{proposition}
\label{prop:welfare-characterization}
A network $\widehat{\mathcal{N}}$ is socially rationalizable if and only if $\widehat{\mathcal{N}}_N=\{1,\ldots,N-1\}$.
\end{proposition}
Thus, even when every agent's payoff receives positive weight, social rationalizability imposes no restriction on links among predecessors, provided that agent $N$ observes all of them.

Beyond Theorem \ref{thm:characterization}, sufficiency requires an additional argument because changing an intermediate agent's neighborhood may improve that agent's own payoff.
For each intermediate agent, the construction identifies a rare history at which the target neighborhood reveals information to agent $N$ that every alternative neighborhood pools.
The possible payoff loss of the intermediate agent is of smaller order than agent $N$'s payoff gain.
These events are introduced at successively smaller scales, so the comparison associated with the earliest difference between the networks dominates all effects on later agents.

\subsection{Star and Complete Networks}
Theorem \ref{thm:characterization} highlights the relationship between informational environments and rationalizable networks, but its general construction is designed to accommodate arbitrary predecessor networks rather than to isolate particular economic mechanisms.
We therefore revisit two important implications for benchmark networks at opposite ends of connectivity.
Using comparatively transparent constructions, we show how the informational environment supports each network and identify the forces that favor sparse and dense observation structures, respectively.

In the analysis of these two benchmark networks, all agents share the prior $1/2$ and a common information structure $\Pi=(S,\pi)$, and their signals are conditionally independent across agents.
With a slight abuse of notation, we write
\[
    V_i(\mathcal{N},\Pi):=V_i\bigl(\mathcal{N},(1/2)_{j=1}^N,(\Pi)_{j=1}^N\bigr)
\]
and denote the likelihood ratio of signal $s\in S$ by $LR(s)=\pi(s\mid H)/\pi(s\mid L)$.

The two benchmark networks are the {\it star network} and the {\it complete network}.
We first focus on a sparse network, which we call a star network, defined as follows.
\begin{definition} \label{def:star}
A network $\mathcal{N}$ is called a {\it star network}, denoted by $\mathcal{N}^S$, if it satisfies
$\mathcal{N}_i=\emptyset$ for $i=1,2,\dots,N-1$ and
$\mathcal{N}_N=\{1,2,\dots,N-1\}$.
\end{definition}
A star network is the most dispersed network in the sense that no agent observes the actions of others except for agent $N$.
This network shares the spirit of the ``sacrificial lambs'' or ``guinea pigs'' mechanism, in that agents $1,\dots,N-1$ act solely on their private signals and thereby generate information for agent $N$.
Figure \ref{fig:star} shows the star network.

We next focus on the most connected network, which we call a complete network, defined as follows.
\begin{definition} \label{def:complete}
A network $\mathcal{N}$ is called a {\it complete network}, denoted by $\mathcal{N}^C$, if it satisfies
$\mathcal{N}_i=\{1,2,\dots,i-1\}$ for all $i=1,2,\dots,N$.
\end{definition}
In contrast to the star network, a complete network is the most connected network in the sense that each agent observes all preceding actions.
Consequently, information can be aggregated and transmitted throughout the entire sequence of agents, although information cascades may also arise, as suggested by the social learning literature.
Figure \ref{fig:complete} shows the complete network.

\begin{figure}[t]
\centering
\begin{subfigure}{0.48\textwidth}
\centering
\begin{tikzpicture}[node distance=1.0cm and 1.5cm]
\centering
\node (R1)  at (0,2)   {$\circledtext{1}$};
\node (R2)  at (0,1)   {$\circledtext{2}$};
\node (Dots) at (0,0) {$\vdots$};
\node (RN2) at (0,-1) {$\circledtext{N-2}$};
\node (RN1) at (0,-2) {$\circledtext{N-1}$};
\node (RN)  at (4,0) {$\circledtext{N}$};

\draw[->, very thick] (R1)  -- (3.5,0.3);
\draw[->, very thick] (R2)  -- (3.5,0.1);
\draw[->, very thick] (RN2) -- (3.5,-0.1);
\draw[->, very thick] (RN1) -- (3.5,-0.3);

\end{tikzpicture}
\caption{Star network}
\label{fig:star}
\end{subfigure}
\hfill
\begin{subfigure}{0.48\textwidth}
\centering
\begin{tikzpicture}[node distance=1.0cm and 1.6cm]
\centering

\node (1) at (0,0) {$\circledtext{1}$};
\node (2) at (2,0) {$\circledtext{2}$};
\node (dots) at (3,0) {$\dots$};
\node (Nm1) at (4,0) {$\circledtext{N-1}$};
\node (N) at (6,0) {$\circledtext{N}$};

\draw[->, very thick] (1) -- (2);
\draw[->, very thick] (Nm1) -- (N);

\draw[->, very thick] (1) to[bend left=-50] (Nm1);
\draw[->, very thick] (1) to[bend left=60] (N);
\draw[->, very thick] (2) to[bend left=-30] (3.5,-0.3);
\draw[->, very thick] (2) to[bend left=30] (5.5,0.3);

\end{tikzpicture}
\caption{Complete network}
\label{fig:complete}
\end{subfigure}

\caption{Star and complete networks}
\label{fig:network_comparison}

\end{figure}
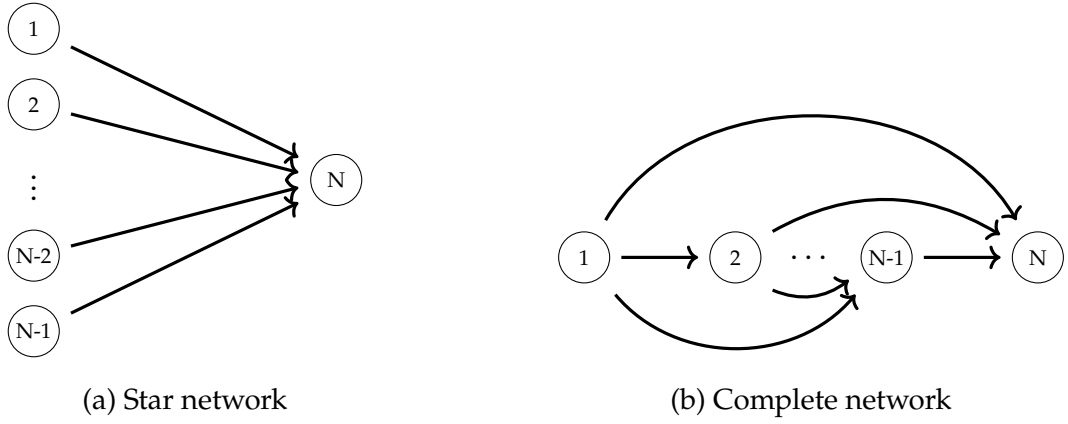

Theorem \ref{thm:characterization} implies that the star network is rationalizable. The following result recovers this implication through a simple binary information structure that makes the underlying mechanism transparent.
\begin{theorem}\label{thm: star}
There exists a common binary information structure $\Pi^B$ such that
$V_N(\mathcal{N}^S,\Pi^B) > V_N(\mathcal{N},\Pi^B)$
for any $\mathcal{N}\neq\mathcal{N}^S$.
\end{theorem}
The intuition behind Theorem \ref{thm: star} is as follows.
Under binary signals, an agent who does not observe others' actions chooses an action that perfectly reveals her private signal.
Consequently, in a star network, agent $N$ can infer the private signals of all preceding agents from their actions.
By contrast, under binary signals, when predecessors are connected, information cascades immediately arise, causing actions to become less responsive to private signals.
This phenomenon appears even in the canonical social learning environments of \citet{banerjee1992simple} and \citet{bikhchandani1992theory}.

This feature gives rise to what we call the {\it responsiveness effect}: as predecessors are less connected, informational interdependence becomes weaker, which preserves the responsiveness of predecessors' actions to their own private information.
This responsiveness prevents society from placing excessive weight on early actions and thereby mitigates information cascades.
As a result, actions remain informative about underlying signals, facilitating information aggregation.
Therefore, when preserving responsiveness to dispersed private information is sufficiently important, the star network outperforms more connected networks.

The responsiveness effect may suggest that the star network is optimal for all information structures. Theorem \ref{thm:characterization}, however, implies that the complete network is also rationalizable. The following result recovers this complementary implication through an information structure that makes a different mechanism transparent.

\begin{theorem}\label{thm: complete}
There exists a common information structure $\Pi^\ast$ such that
$V_N(\mathcal{N}^C,\Pi^\ast) > V_N(\mathcal{N},\Pi^\ast)$
for any $\mathcal{N}\neq\mathcal{N}^C$.
\end{theorem}
The construction of the information structure behind Theorem \ref{thm: complete} is based on four signals, $\{l^\ast,l,h,h^\ast\}$.
The signals $h$ and $h^\ast$ are positive signals that favor state $H$, while $l$ and $l^\ast$ are negative signals that favor state $L$.
We interpret $l$ and $h$ as ordinary signals, and $l^\ast$ and $h^\ast$ as strong signals.
The informativeness of these signals is ordered in a lexicographic manner, in the order of $l^\ast, h^\ast, l, h$, in the following sense.
Consider a hypothetical situation in which agent $N$ can observe all $N$ signals directly.
Even if $N-1$ of them are $h$ and only one is $l$, the agent puts more weight on $l$ and takes action $0$.
The same relationship holds between $l$ and $h^\ast$, and between $h^\ast$ and $l^\ast$: a single more informative signal outweighs any number of less informative ones.
The signals are ordered in the opposite way in terms of the probability of each signal: $h$, $l$, $h^\ast$, $l^\ast$.

Under such an information structure, since agents observe only actions, not predecessors' signals directly, particular action histories have sharp interpretations.
Consider the complete network.
A history consisting only of action $1$ is interpreted, with high probability, as being generated by signal $h$ alone.
Suppose that an agent overturns this history by choosing action $0$.
This overturning suggests that the agent receives signal $l$ with high probability.
Suppose that a subsequent agent then takes action $1$.
This overturning suggests that the agent receives good news more informative than signal $l$, namely signal $h^\ast$. 
Since all subsequent agents observe this overturning, they can draw this inference.
Finally, suppose that a further subsequent agent takes action $0$.
By the same reasoning, this overturning suggests that the agent receives signal $l^\ast$.
In this way, through the overturning of actions, the history successively reveals the presence of more informative signals.
Note that this overturning argument relies on the property of the complete network that every agent observes all past actions.

This feature generates what we call the {\it overturning effect}: as predecessors are more connected, later agents are more likely to detect the presence of informative signals through the inference that an agent must have received a signal informative enough to overturn past actions.
Such overturnings in actions correct histories that would otherwise induce incorrect inferences, leading to more efficient information aggregation.
This correction is possible only when an overturning in past actions becomes part of the public history.
Therefore, when detecting such informative signals is sufficiently important, the complete network outperforms less connected networks.

The advantage of the complete network is not limited to the particular agent on whom we have focused.
Under the same information structure, the complete network Pareto dominates any other network.
\begin{proposition}\label{prop: pareto dominate}
    There exists an information structure $\Pi^\ast$ such that $V_{j}(\mathcal{N}^C,\Pi^\ast) \geq  V_{j}(\mathcal{N},\Pi^\ast)$ for any $j \in \{1,\dots,N-1 \}$ and $V_{N}(\mathcal{N}^C,\Pi^\ast) >  V_{N}(\mathcal{N},\Pi^\ast)$ under any network $\mathcal{N}\neq \mathcal{N}^C$. 
\end{proposition}
Proposition \ref{prop: pareto dominate} strengthens the implication of Theorem \ref{thm: complete} by extending the comparison from a single agent's payoff to a welfare comparison for the entire society. 
Under the same information structure, the complete network weakly improves the payoff of every agent and strictly improves the payoff of agent $N$ relative to any alternative network. 
Roughly speaking, this is because the subnetwork consisting of the first $N'$ agents in the complete network is itself the complete network of $N'$ agents.
Hence, under any utilitarian social welfare criterion that assigns a positive weight to agent $N$, the complete network yields strictly higher social welfare than any alternative network. 
In environments in which highly informative signals play a central role, extensive connectivity enables these signals to influence the decisions of many agents, thereby enhancing social welfare throughout the network.


\section{Trade-off between Responsiveness and Overturning} \label{sec: star vs complete}
In this section, we broaden our focus to transparent and tractable classes of information structures in which both the responsiveness effect and the overturning effect matter.
By comparing networks in a class that contains the star and complete networks as its endpoints, we identify which features of the information structures determine which effect dominates.

For $M\in\{1,\ldots,N-1\}$, the network is
\begin{align*}
    \mathcal N_i^M=\begin{cases}
        \emptyset, & i\leq M,\\
        \{1,\ldots,i-1\}, & i>M.
    \end{cases}
\end{align*}
The first $M$ agents (guinea pigs) act independently; subsequent agents observe all preceding actions.
Thus, $\mathcal N^1=\mathcal N^C$ and $\mathcal N^{N-1}=\mathcal N^S$.

Signals are independent conditional on $\theta$, and the network and information structure are common knowledge.
Throughout this section, we fix $\mu_i=1/2$ for all $i$.
The objective is $V_N(\mathcal N^M,\Pi)=(1/2)\alpha_N^L(0)+(1/2)\alpha_N^H(1)$.

\subsection{Mixtures of Binary and Conclusive Signals} \label{sec: mixture}
We first consider a simple and tractable class of information structures that captures the central intuition behind the trade-off.
The class starts from a symmetric binary information structure and enriches it with the possibility of conclusive signals.
The binary component provides the simplest environment in which the responsiveness effect is important: when agents' actions cease to reflect their own binary signals, information aggregation deteriorates.
The added conclusive signals represent an extreme form of highly informative signals that are central to the overturning effect: once such a signal is inferred from an action, it can overturn an incorrect inference generated by the preceding history.
By varying the frequency of conclusive signals and the accuracy of the original binary component, this class transparently captures the balance between the two effects and allows us to identify the optimal number of guinea pigs.

Formally, let $\Pi^{\varepsilon,p}=(\{l^{*},l,h,h^\ast\},\pi)$ be the information structure defined by
\[
\begin{alignedat}{2}
& \pi(l \mid L) = \pi(h \mid H) = (1-\varepsilon)p, &\qquad
& \pi(h \mid L) = \pi(l \mid H) = (1-\varepsilon)(1-p), \\
& \pi(l^{*} \mid L) = \pi(h^{*} \mid H) = \varepsilon, &\qquad
& \pi(h^{*} \mid L) = \pi(l^{*} \mid H) = 0,
\end{alignedat}
\]
where $\varepsilon\in(0,1)$ and $p\in(1/2,1)$.
Signals $l^{*}$ and $h^{*}$ are conclusive signals that perfectly reveal the state.
Under both states, a conclusive signal occurs with probability $\varepsilon$.
Otherwise, either of the two ordinary signals, $l$ or $h$, occurs.
These signals are symmetric with accuracy $p$.

Let $M^*\in\{1,\ldots,N-1\}$ denote an optimal number of guinea pigs, maximizing agent $N$'s expected payoff $V_N(\mathcal N^M,\Pi^{\varepsilon,p})$.
\begin{proposition} \label{prop:independent-optimum}\label{prop: star vs complete}
For each $p\in(1/2,1)$, there are unique cutoffs $\varepsilon_k(p)$, indexed by integers $k\geq0$, with $1>\varepsilon_0(p)>\varepsilon_1(p)>\cdots>0$.
For every odd $N\geq3$, let $K=(N-1)/2$.\footnote{For even $N$, the same classification holds with $K=(N-2)/2$, and the conclusion in (ii) uses $M=N-2$ instead of $M=N-1$; all other conclusions remain unchanged. The star network, $M=N-1$, is then strictly dominated by $M=N-2$ by Lemma \ref{lem:independent-payoff}.}
Except at $\varepsilon_0(p),\ldots,\varepsilon_{K-1}(p)$, the optimum is unique and given by
\begin{align*}
    M^*=\begin{cases}
        1, & \varepsilon>\varepsilon_0(p),\\
        2k, & \varepsilon_k(p)<\varepsilon<\varepsilon_{k-1}(p),\quad 1\leq k<K,\\
        2K, & 0<\varepsilon<\varepsilon_{K-1}(p).
    \end{cases}
\end{align*}
The cases $M=1$ and $M=2K=N-1$ correspond to the complete and star networks, respectively.
At $\varepsilon=\varepsilon_0(p)$, precisely $M=1,2$ are optimal; at $\varepsilon=\varepsilon_k(p)$ for $1\leq k<K$, precisely $M=2k,2k+2$ are optimal.
Moreover,
\begin{enumerate}[(i)]
    \item Each $\varepsilon_k(p)$ is strictly increasing in $p$, with $\lim_{p\searrow1/2}\varepsilon_k(p)=0$ and $\lim_{p\nearrow1}\varepsilon_k(p)=1$.
    In particular, $\varepsilon_0(p)=p(2p-1)/[p+2(1-p)^2]$.
    \item For fixed $p$, $\lim_{k\to\infty}\varepsilon_k(p)=(2p-1)^2/[1+4(1-p)^2]$.
    If $\varepsilon$ is at or below this limit, $M=N-1$ is uniquely optimal for every odd $N\geq3$.
    Above this limit, the optimal set remains unchanged as $N$ increases through odd integers beyond a finite threshold.
\end{enumerate}
\end{proposition}

Proposition \ref{prop:independent-optimum} characterizes the optimal number of guinea pigs by a sequence of cutoffs.
This cutoff characterization highlights the key trade-off between the responsiveness effect and the overturning effect.
When conclusive signals are rare, preserving independent responses to ordinary signals is more valuable, and hence the star network is optimal.
As conclusive signals become more frequent, allowing more agents to observe preceding actions becomes more valuable, and the optimum passes through intermediate networks before reaching the complete network.

Part (i) shows how accuracy of ordinary signals shifts this balance.
Increasing $p$ raises each cutoff: more accurate ordinary signals make preserving independent responses more valuable, and thus a larger frequency of conclusive signals is required to favor fewer guinea pigs.
Figure \ref{fig:independent-optimum(parameter_region)} illustrates the cutoffs for $N=5$.
\begin{figure}[ht]
\centering
\begin{tikzpicture}
\begin{axis}[
    width=9cm, height=7cm,
    xmin=1/2, xmax=1, ymin=0, ymax=1,
    xtick=\empty,
    ytick=\empty,
    axis x line=bottom,
    axis y line=left,
    axis line style={draw=none},
    clip=false,
    every axis plot/.append style={line width=1pt},
    declare function={
        cutoffzero(\prob)=\prob*(2*\prob-1)/(\prob+2*(1-\prob)^2);
        cutoffone(\prob)=1-12*(1-\prob)/(2-4*(1-\prob)+sqrt((4*(1-\prob)-2)^2+24*(1-\prob)*(1-(1-\prob)^2))+12*(1-\prob)^2);
    },
]

\addplot[draw=none, name path=cutzero, domain=0.5:1, samples=201] {cutoffzero(x)};
\addplot[draw=none, name path=cutone, domain=0.5:1, samples=201] {cutoffone(x)};
\path[name path=top] (axis cs:0.5,1) -- (axis cs:1,1);
\path[name path=bottom] (axis cs:0.5,0) -- (axis cs:1,0);
\addplot[draw=none, fill=gray!50] fill between[of=cutzero and top];
\addplot[draw=none, fill=gray!80] fill between[of=cutzero and cutone];
\addplot[draw=none, fill=gray!10] fill between[of=cutone and bottom];

\addplot[black, very thick, domain=0.5:1, samples=201] {cutoffzero(x)};
\addplot[black, very thick, domain=0.5:1, samples=201] {cutoffone(x)};

\draw[very thick] (axis cs:1.0,0) -- (axis cs:1.0,1.0);
\draw[very thick] (axis cs:0.5,1.0) -- (axis cs:1.0,1.0);
\draw[->, very thick] (axis cs:0.5,0) -- (axis cs:1.05,0);
\draw[->, very thick] (axis cs:0.5,0) -- (axis cs:0.5,1.1);
\node[right] at (axis cs:1.05,0) {$p$};
\node[above] at (axis cs:0.5,1.1) {$\varepsilon$};
\node[below] at (axis cs:0.52,0) {$1/2$};
\node[left] at (axis cs:0.5,0.025) {$0$};
\node[below] at (axis cs:1,0) {$1$};
\node[left] at (axis cs:0.5,1) {$1$};

\node[align=center] at (axis cs:0.66,0.80) {Complete network\\$M=1$ is optimal};
\node[font=\small, rotate=40, inner sep=1pt] at (axis cs:0.78,{(cutoffzero(0.78)+cutoffone(0.78))/2}) {$M=2$};
\node[align=center] at (axis cs:0.87,0.18) {Star network\\$M=4$ is optimal};

\end{axis}
\end{tikzpicture}
\caption{Parameter regions for the optimal number of guinea pigs}
\label{fig:independent-optimum(parameter_region)}\label{fig:star_vs_complete(parameter_region)}
\begin{threeparttable}
\begin{tablenotes}
    \footnotesize
    \item[] \textit{Notes:}
    The figure shows the optimal network within $\{\mathcal N^M:1\leq M\leq4\}$ for $N=5$.
    The horizontal axis is the accuracy of ordinary signals $p$; the vertical axis is the frequency of conclusive signals $\varepsilon$.
    The solid curves are $\varepsilon_0(p)$ and $\varepsilon_1(p)$, defined by $J_1=J_0$ and $J_2=J_1$, respectively.
    The dark-gray region above $\varepsilon_0(p)$ has the complete network ($M=1$) as its unique optimizer.
    The diagonally hatched, medium-gray region between $\varepsilon_1(p)$ and $\varepsilon_0(p)$ has $M=2$ as its unique optimizer.
    The light-gray region below $\varepsilon_1(p)$ has the star network ($M=4$) as its unique optimizer.
    On either solid curve, the two adjacent choices of $M$ are both optimal.
    The endpoints are shown as limits; the parameter domain is $p\in(1/2,1)$ and $\varepsilon\in(0,1)$.
\end{tablenotes}
\end{threeparttable}
\end{figure}
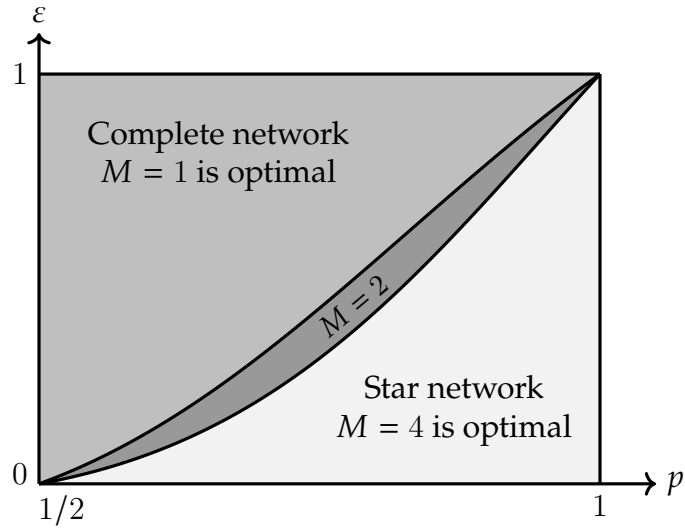

Part (ii) shows that the trade-off does not vanish as the number of agents grows.
At or below the limiting cutoff, the star network is uniquely optimal for every odd $N$.
Above the limit, the optimal set of guinea-pig counts eventually remains unchanged as $N$ increases.

\subsection{Numerical Exercises} \label{sec: numerical}\label{sec:guinea-pigs-numerical}
Section \ref{sec: mixture} provides a sharp cutoff result within a deliberately simple information structure consisting of mixtures of binary signals with conclusive signals.
We now ask whether a similar pattern appears in richer, smooth families of information structures.
In this subsection, we numerically compare the networks in $\{\mathcal N^M:1\leq M\leq N-1\}$ under four families of private-belief distributions.
The truncated normal family provides unimodal densities, while the symmetric Beta family ranges from unimodal to U-shaped densities. The other two families are mixture distributions that generate W-shaped densities with substantial mass near both the prior and the endpoints.

With a slight abuse of notation, we relabel the private signal $s$ as the private belief induced by $s$, that is, $s=\pi(s\mid H)/[\pi(s\mid H)+\pi(s\mid L)]$.
The numerical exercises use four unconditional densities of private beliefs on $[0,1]$.
For each unconditional density $f$, Bayes plausibility and the common prior $1/2$ determine the corresponding conditional densities as $f_H(s)=2s f(s)$ and $f_L(s)=2(1-s)f(s)$.
First, the truncated normal density is
\[
f^{\text{N}}_{\sigma}(s)\coloneq\frac{\exp\!\left(-\frac{(s-1/2)^2}{2\sigma^2}\right)}{\int_0^1 \exp\!\left(-\frac{(t-1/2)^2}{2\sigma^2}\right)dt}.
\]
Second, the symmetric Beta density is
\[
f^{\text{B}}_{\alpha}(s)\coloneq\frac{s^{\frac{1}{\alpha}-1}(1-s)^{\frac{1}{\alpha}-1}}{\int_0^1 t^{\frac{1}{\alpha}-1}(1-t)^{\frac{1}{\alpha}-1}dt}.
\]
Third, the Beta-mixture density is
\[
f^{\text{BM}}_{\varepsilon}(s)\coloneq (1-\varepsilon)\,\mathrm{Beta}(s;10,10)+\frac{\varepsilon}{2}\,\mathrm{Beta}(s;0.3,5)+\frac{\varepsilon}{2}\,\mathrm{Beta}(s;5,0.3).
\]
Fourth, the contaminated log-likelihood-ratio normal density is 
\[
f^{\text{CLN}}_{\varepsilon}(s)\coloneq\frac{(1-\varepsilon)\phi(\log(s/(1-s));0,0.5^2)+\varepsilon\phi(\log(s/(1-s));0,4^2)}{s(1-s)},
\]
where $\phi(\cdot;\mu,\tau^2)$ denotes the normal density with mean $\mu$ and variance $\tau^2$.
In the truncated normal family, a larger $\sigma$ makes private beliefs more dispersed.
In the symmetric Beta family, a larger $\alpha$ makes private beliefs more dispersed, with a U-shaped density when $\alpha>1$.
In the two mixture families, a larger $\varepsilon$ increases the probability of private beliefs close to $0$ or $1$.
The shapes of these densities are summarized in Figure \ref{fig:density-shapes}.
\begin{figure}[ht]
    \centering
    \captionsetup{skip=4pt}
    \IfFileExists{guinea_pigs_densities.png}{
        \includegraphics[width=\linewidth,trim=0 14bp 0 0,clip]{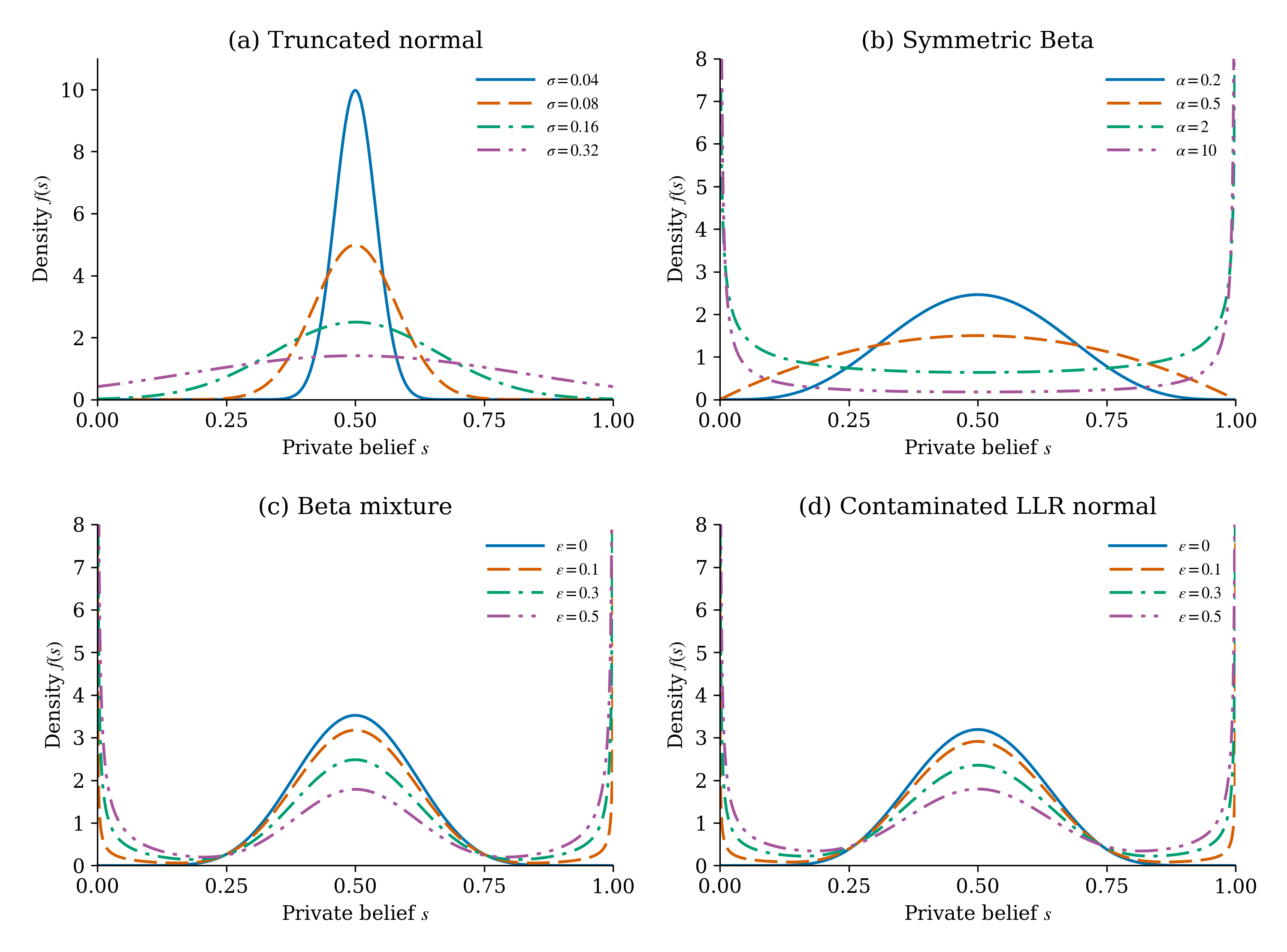}
    }{
        \fbox{\texttt{Upload guinea\_pigs\_densities.png}}
    }
    \caption{Private-belief densities used in the numerical exercises}
    \begin{tablenotes}
    \footnotesize
    \item[] \textit{Notes:} These figures plot representative unconditional densities of private beliefs. The parameter values are $\sigma\in\{0.04,0.08,0.16,0.32\}$ for the truncated normal distribution, $\alpha\in\{0.2,0.5,2,10\}$ for the symmetric Beta distribution, and $\varepsilon\in\{0,0.1,0.3,0.5\}$ for the two mixture distributions.
    \end{tablenotes}
    \label{fig:density-shapes}\label{fig:guinea-pigs-density}
\end{figure}

To see when each of the responsiveness and overturning effects dominates, we numerically compare all networks in $\{\mathcal N^M:1\leq M\leq N-1\}$.
Figure \ref{fig:guinea-pigs-optimal-M} plots the number of guinea pigs that maximizes agent $N$'s expected payoff for $N=3,5,7,9$.
The comparison reveals distinct patterns across these families.
In the truncated normal family, the star network is optimal at low dispersion and $M=N-2$ is optimal at higher dispersion on the reported grid.
Increasing $\sigma$ does not make the complete network optimal for all four horizons: in the uniform-density limit, the maximizing sets are $\{1,2\}$, $\{3,4\}$, $\{5,6\}$, and $\{7,8\}$, respectively.
In the symmetric Beta family, the optimum moves from the star network through intermediate networks to the complete network.
In the mixture families with W-shaped distributions, the optimum likewise shifts from the star network through intermediate networks to the complete network as the extreme-belief component becomes more frequent.

\begin{figure}[ht]
    \centering
    \captionsetup{skip=4pt}
    \IfFileExists{guinea_pigs_optimal_M.png}{
        \includegraphics[width=\linewidth,trim=0 17bp 0 0,clip]{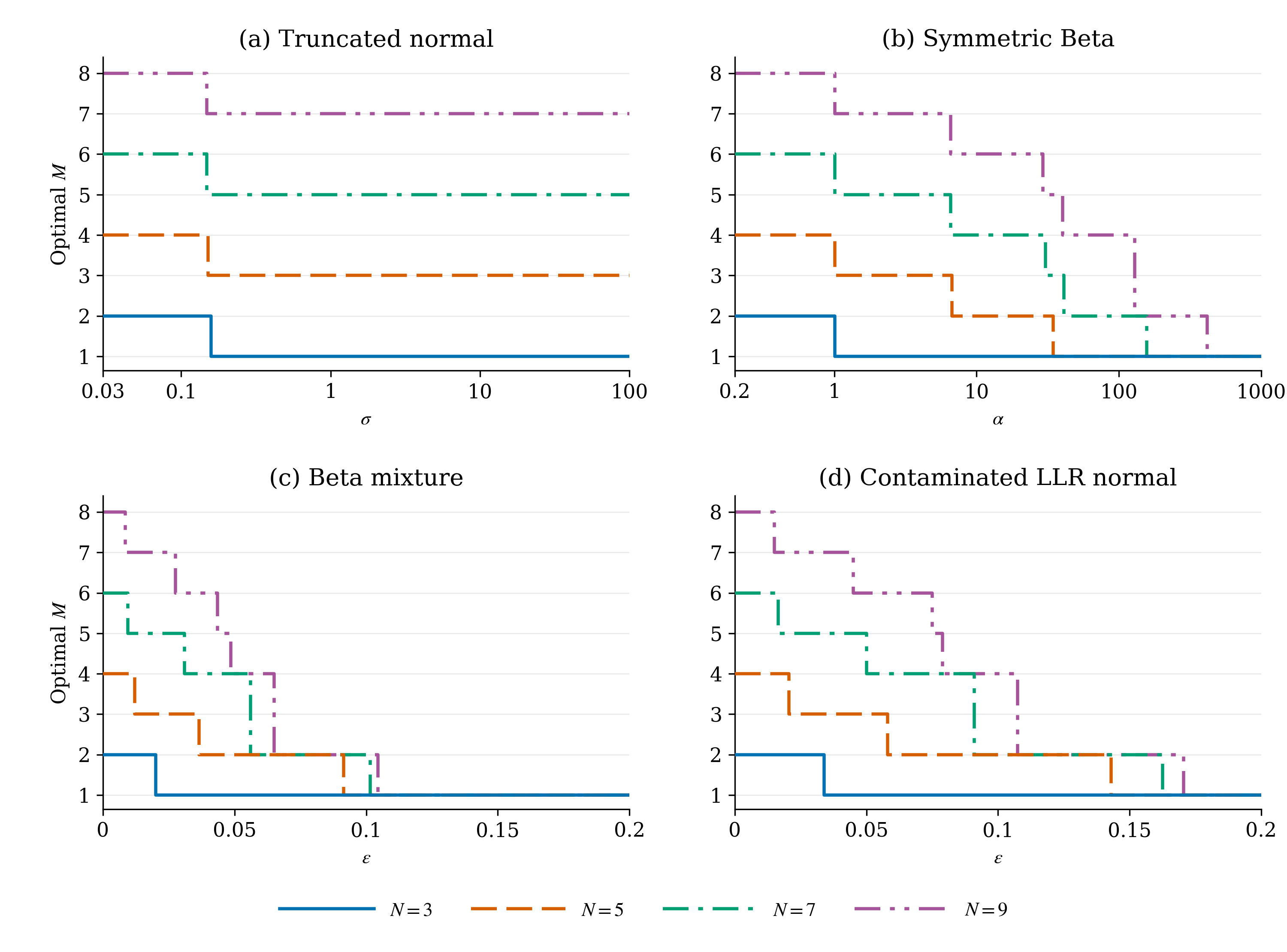}
    }{
        \fbox{\texttt{Upload guinea\_pigs\_optimal\_M.png}}
    }
    \caption{The optimal number of guinea pigs across private-belief distributions}
    \label{fig:guinea-pigs-optimal-M}\label{fig:comparison}
    \begin{tablenotes}
        \footnotesize
        \item[] \textit{Notes:} Each line plots the maximizing $M\in\{1,\ldots,N-1\}$ for one of $N=3,5,7,9$.
        There are 401 parameter values per family: logarithmically spaced on $\sigma\in[0.03,100]$ and $\alpha\in[0.2,1000]$, and linearly spaced on $\varepsilon\in[0,0.2]$.
        The first two horizontal axes are logarithmic; the mixture axes are linear.
        Colors and line patterns distinguish horizons: solid, long-dashed, dash-dot, and dash-dot-dot for $N=3,5,7,9$, respectively.
    \end{tablenotes}
\end{figure}

This pattern can be interpreted as follows.
Under the W-shaped distributions, endpoint realizations are highly informative and occur with non-negligible probability.
Thus, the benefit of the overturning effect is large.
In addition, an agent who observes no history pools private beliefs near the prior with those near an endpoint, and hence her action cannot distinguish between them.
Therefore, the benefit of the responsiveness effect is small.

Overall, the numerical exercise highlights the same trade-off as the theoretical analysis.
When information is primarily driven by moderately informative signals, preserving responsiveness is more important, and many guinea pigs are optimal.
By contrast, when highly informative signals occur with large probability, as in the W-shaped mixture distributions and the symmetric Beta distribution with high parameters, the ability to overturn incorrect public inferences becomes valuable, and the complete network becomes optimal on the reported grids.
The extension to intermediate networks qualifies the endpoint comparison: an intermediate network can strictly outperform both endpoints, and odd values $M>1$ can be optimal.
These exercises describe the reported finite horizons and grids; they do not establish a general monotonicity or asymptotic result for the smooth families.

\section{Conclusion} \label{sec: conclusion} 
Information inferred from others’ behavior is a fundamental source of information in society.
In social learning environments, the network structure, that is, who observes whose past actions, plays a central role in determining how dispersed private information is aggregated.
This paper studies this role by comparing networks in terms of the expected payoff of a given agent, or equivalently, at a given period.

Theorem \ref{thm:characterization} highlights the relationship between informational environments and network structure: a network can be uniquely optimal for some informational environment if and only if the focal agent observes every predecessor.
Theorems \ref{thm: star} and \ref{thm: complete} revisit two important implications of this characterization through transparent constructions for the star and complete networks.

The key mechanism underlying these results is a trade-off between the {\it responsiveness effect} and the {\it overturning effect}.
Sparse networks generate less informational interdependence among agents’ actions and thereby preserve the responsiveness of actions to private signals.
This responsiveness effect mitigates information cascades and improves aggregation especially when moderately informative signals are the main source of information.
By contrast, dense networks allow agents to observe richer histories and make it possible for highly informative signals to be identified and disseminated to subsequent agents through overturnings in the history.
This overturning effect helps correct incorrect inferences induced by early histories, leading to effective information aggregation especially when sufficiently informative signals occur with high probability and ordinary signals are less informative.

Our analysis also has implications for the network design in organizations and other social environments, where institutions, communication protocols, or platform rules may shape who observes whose actions. 
Although we treat the network structure as exogenously given, our results suggest that the value of an observation link cannot be evaluated independently of the underlying information structure. 
Thus, designing an observation network to improve the efficiency of information aggregation requires understanding what kinds of private information those actions are likely to transmit.

These findings naturally raise a broader design question: given an informational environment, which observation network best facilitates information aggregation?
A general characterization of the optimal network remains an important direction for future research.

\titleformat{\section}
		{\Large\bfseries}     
         {Appendix \thesection:}
        {0.5em}
        {}
        []

\renewcommand{\thetheorem}{A.\arabic{theorem}}
\setcounter{theorem}{0}

 \appendix 


\section{Proofs of Theorem \ref{thm:characterization} and Proposition \ref{prop:welfare-characterization}}

\subsection{Proof of Theorem \ref{thm:characterization}}
The sufficiency proceeds in two stages.
First, holding fixed that agent $N$ observes every predecessor, we construct priors and information structures inductively for agents $i=2,\ldots,N-1$ so that the target network strictly outperforms every network that first differs from it in agent $i$'s neighborhood.
Second, rare signals make observing every predecessor strictly better for agent $N$ than observing any proper subset of $\{1,\ldots,N-1\}$.
For any network $\mathcal N$ that differs from the target network before agent $N$, define its first deviation as the smallest $i<N$ such that $\mathcal N_i\neq\widehat{\mathcal N}_i$.
Throughout this subsection, all information structures have finite signal spaces.

For a network $\mathcal N$ and a profile of information structures $\Pi$, write $\mathcal N_{\leq k}:=(\mathcal N_1,\ldots,\mathcal N_k)$ and $\Pi_{\leq k}:=(\Pi_1,\ldots,\Pi_k)$.
Write $\alpha_{\leq k}^\theta(\cdot\mid\mathcal N_{\leq k},\Pi_{\leq k})$ for the equilibrium action distribution, suppressing the fixed prior profile and equilibrium strategy, and set
\[
 \boldsymbol\alpha_{\leq k}(\cdot\mid\mathcal N_{\leq k},\Pi_{\leq k}):=\bigl(\alpha_{\leq k}^L(\cdot\mid\mathcal N_{\leq k},\Pi_{\leq k}),\alpha_{\leq k}^H(\cdot\mid\mathcal N_{\leq k},\Pi_{\leq k})\bigr).
\]
For $\mathcal M\subseteq\{1,\ldots,k\}$, write $A^{\mathcal M}$ for the set of action vectors indexed by the agents in $\mathcal M$.
Given a pair $\boldsymbol\alpha_{\leq k}:=(\alpha_{\leq k}^L,\alpha_{\leq k}^H)$ of action distributions with full support on $A^k$, define the likelihood ratio associated with $\mathcal M\subseteq\{1,\ldots,k\}$ and $h\in A^{\mathcal M}$ by
\[
 \Lambda_k(\mathcal M,h\mid\boldsymbol\alpha_{\leq k})
 :=\frac{\mathbb P(a_{\mathcal M}=h\mid H)}{\mathbb P(a_{\mathcal M}=h\mid L)}
 =\frac{\sum_{\bm a\in A^k:\bm a_{\mathcal M}=h}\alpha_{\leq k}^H(\bm a)}{\sum_{\bm a\in A^k:\bm a_{\mathcal M}=h}\alpha_{\leq k}^L(\bm a)}.
\]
When $\mathcal M=\emptyset$, $\Lambda_k(\emptyset,\emptyset\mid\boldsymbol\alpha_{\leq k})=1$.
\begin{definition}
A pair $\boldsymbol\alpha_{\leq k}:=(\alpha_{\leq k}^L,\alpha_{\leq k}^H)$ of action distributions on $A^k$ is \emph{observation separating} if both distributions have full support and $\Lambda_k(\mathcal M,h\mid\boldsymbol\alpha_{\leq k})\neq\Lambda_k(\mathcal M',h'\mid\boldsymbol\alpha_{\leq k})$ for all $\mathcal M,\mathcal M'\subseteq\{1,\ldots,k\}$, $h\in A^{\mathcal M}$, and $h'\in A^{\mathcal M'}$ such that $(\mathcal M,h)\neq(\mathcal M',h')$.
\end{definition}

The following implementation result provides an open set of action probabilities that can avoid the conditions under which two likelihood ratios are equal.
For any $(\mu_{i},\Pi_i)$, define the induced probability of action $1$ under state $\theta$ after an observation with likelihood ratio $q>0$ by
\[
 \varphi_i^\theta(q;\mu_i,\Pi_i):=\pi_i\bigl(\{s_i\in S_i:\mu_i/(1-\mu_i)qLR_i(s_i)>1\}\mid\theta\bigr).
\]
\begin{lemma}
\label{lem:threshold-open-set}
Fix a prior $\mu_i\in(0,1)$ and likelihood ratios $0<q_1<\cdots<q_m$.
The image of the map $\Pi_i\longmapsto\bigl(\varphi_i^\theta(q_j;\mu_i,\Pi_i)\bigr)_{\theta\in\Theta,\,j\in\{1,\ldots,m\}}$ has nonempty interior in $\mathbb R^{2m}$.
\end{lemma}

\begin{proof}[Proof of Lemma~\ref{lem:threshold-open-set}]
Partition private likelihood ratios into the $m+1$ open cutoff intervals
\begin{align*}
 I_0=\left(\frac{1-\mu_i}{\mu_i q_1},\infty\right), \qquad
 I_k=\left(\frac{1-\mu_i}{\mu_i q_{k+1}},
             \frac{1-\mu_i}{\mu_i q_k}\right)
        \;\; (k=1,\ldots,m-1), \qquad
 I_m=\left(0,\frac{1-\mu_i}{\mu_i q_m}\right).
\end{align*}
A signal in $I_k$ induces action $1$ exactly at $q_{k+1},\ldots,q_m$.
Let $x_k^L$ and $x_k^H$ be the total probabilities of signals in category $I_k$ under states $L$ and $H$, respectively.
Consider the set
\begin{align*}
X=\left\{
\left((x_k^L)_{k=1}^m,(x_k^H)_{k=1}^m\right)\in\mathbb{R}^{2m}
\;\middle|\;
\begin{aligned}
&x_k^L>0,\quad x_k^H>0,\quad
  \frac{x_k^H}{x_k^L}\in I_k \quad (k=1,\dots,m),\\
&\sum_{k=1}^m x_k^L<1,\quad
  \sum_{k=1}^m x_k^H<1,\quad
  \frac{1-\sum_{k=1}^m x_k^H}{1-\sum_{k=1}^m x_k^L}\in I_0
\end{aligned}
\right\}.
\end{align*}
This is a nonempty open subset of $\mathbb R^{2m}$.
To see nonemptiness, choose $r_k\in I_k$ for every $k$, with $r_0>1>r_m$.
A positive mixture of $r_0$ and $r_m$ has mean one; assigning sufficiently small positive probabilities to the remaining $r_k$'s and adjusting the two endpoint probabilities preserves this equality and strict positivity.

Define $F:X\to\mathbb R^{2m}$ by $F((x_k^L)_{k=1}^m,(x_k^H)_{k=1}^m)=((1-\sum_{k=j}^m x_k^L)_{j=1}^m, (1-\sum_{k=j}^m x_k^H)_{j=1}^m)$.
For every point in $X$, there is a corresponding information structure such that $\varphi_i^L(q_j;\mu_i,\Pi_i)=1-\sum_{k\geq j}x_k^L$ and $\varphi_i^H(q_j;\mu_i,\Pi_i)=1-\sum_{k\geq j}x_k^H$ for every $j$.
Thus, $F(X)$ is contained in the image of the map in the statement.
The Jacobian of $F$ is nonsingular, with determinant $(-1)^{2m}=1$.
Hence, $F$ is an open map, and $F(X)$ is a nonempty open subset of $\mathbb R^{2m}$.
\end{proof}

Observation separation allows us to choose agent $i$'s prior so that she is indifferent after a target observation, making two nearby probe signals induce different actions there and the same action after every other observation.
The next lemma combines this local probe construction with a generic perturbation that preserves observation separation.
\begin{lemma}
\label{lem:generic-extension}
Take any $\mathcal N$ and $(\mu_j,\Pi_j)_{j<i}$ such that $\boldsymbol\alpha_{\leq i-1}(\cdot\mid\mathcal N_{\leq i-1},\Pi_{\leq i-1})$ is observation separating.
For any $h^*\in A^{\mathcal N_i}$ and $\eta>0$, there exists $(\mu_i,\Pi_i)$ such that (i) the signal space $S_i$ contains two distinct signals $p^-$ and $p^+$ such that $p^-$ induces action $0$ and $p^+$ induces action $1$ after $h^*$, while $p^-$ and $p^+$ induce the same action after every other observation, (ii) $(1/2)\sum_{s\in S_i}|\pi_i(s\mid L)-\pi_i(s\mid H)|<\eta$, and (iii) $\boldsymbol\alpha_{\leq i}(\cdot\mid\mathcal N_{\leq i},\Pi_{\leq i})$ is observation separating.
\end{lemma}

\begin{proof}[Proof of Lemma~\ref{lem:generic-extension}]
Given $\mathcal M\subseteq\{1,\ldots,i-1\}$ and $h\in A^{\mathcal M}$, let $q(\mathcal M,h):=\Lambda_{i-1}(\mathcal M,h\mid\boldsymbol\alpha_{\leq i-1}(\cdot\mid\mathcal N_{\leq i-1},\Pi_{\leq i-1}))$. Define $q^*:=q(\mathcal N_i,h^*)$.
Set $\mu_i=1/(1+q^*)$, $LR_i(p^-)=1-\varepsilon$, and $LR_i(p^+)=1+\varepsilon$.
Since $q^*\neq q(\mathcal M,h)$ for all $(\mathcal M,h)\neq(\mathcal N_i,h^*)$, for sufficiently small $\varepsilon\in(0,\eta)$,
\[
 \frac{\mu_i}{1-\mu_i}q^*LR_i(p^-)<1<
 \frac{\mu_i}{1-\mu_i}q^*LR_i(p^+),
\]
while the two posterior odds lie strictly on the same side of one after every other observation.
Fix such $\varepsilon$ and let $\pi_i^p(p^-\mid L)=\pi_i^p(p^+\mid L)=1/2$, $\pi_i^p(p^-\mid H)=(1-\varepsilon)/2$, and $\pi_i^p(p^+\mid H)=(1+\varepsilon)/2$. Define $\Pi_i^p=(\{p^-,p^+\},\pi_i^p)$.

Let $q_1<\cdots<q_m$ be the values of $q(\mathcal M,h)$ over all observations $(\mathcal M,h)$ through agent $i-1$.
By observation separation, these values are all distinct, so $q^*=q_{j^*}$ for some $j^*\in\{1,\ldots,m\}$.
By Lemma~\ref{lem:threshold-open-set}, there exist an information structure $\bar\Pi_i$, a vector $\bar{\boldsymbol\varphi}:=\boldsymbol\varphi(\bar\Pi_i)\in\mathbb R^{2m}$, and $r>0$ such that $B_r(\bar{\boldsymbol\varphi})\subseteq\operatorname{Im}\boldsymbol\varphi\cap(0,1)^{2m}$, where $\boldsymbol\varphi(\Pi_i):=(\varphi_i^\theta(q_j;\mu_i,\Pi_i))_{\theta\in\{L,H\},j=1,\ldots,m}$.

Fix $\delta>0$ sufficiently small that $\varepsilon/2+\delta<\eta$.
Consider the information structure $\Pi_i^\delta$ obtained by selecting $\Pi_i^p$ with probability $1-\delta$ and $\Pi_i$ with probability $\delta$, independently of the state, where the two components use disjoint signal spaces. 
Then, $\boldsymbol\varphi(\Pi_i^\delta)=(1-\delta)\boldsymbol\varphi(\Pi_i^p)+\delta\boldsymbol\varphi(\Pi_i)$.
Hence, as $\Pi_i$ varies over information structures, the set of attainable $\boldsymbol\varphi$-vectors contains the open ball $B_{\delta r}((1-\delta)\boldsymbol\varphi(\Pi_i^p)+\delta\bar{\boldsymbol\varphi})$.

The total variation distance between the two state-contingent signal distributions satisfies
\begin{align*}
\frac12\sum_{s\in S_i^\delta}
|\pi_i^\delta(s\mid L)-\pi_i^\delta(s\mid H)|
\leq\frac{\varepsilon}{2}+\delta<\eta.
\end{align*}

It remains to show that generic conditional action probabilities in this open set give observation separation on $A^i$.
Fix two distinct observations $(\mathcal N_{i+1},h)$ and $(\mathcal N'_{i+1},h')$ through agent $i$.
Equality of their likelihood ratios is equivalent to
\[
 \mathbb P(a_{\mathcal N_{i+1}}=h\mid H)\mathbb P(a_{\mathcal N'_{i+1}}=h'\mid L)
 -\mathbb P(a_{\mathcal N_{i+1}}=h\mid L)\mathbb P(a_{\mathcal N'_{i+1}}=h'\mid H)=0.
\]
For fixed probabilities of action histories through agent $i-1$, the left-hand side is a polynomial in the conditional action probabilities of agent $i$.
Let $(\mathcal N_{i+1}^-,h^-)=(\mathcal N_{i+1}\setminus\{i\},h_{\mathcal N_{i+1}\setminus\{i\}})$ and define $(\mathcal N_{i+1}^{\prime-},h^{\prime-})$ analogously.
If $(\mathcal N_{i+1}^-,h^-)\neq(\mathcal N_{i+1}^{\prime-},h^{\prime-})$, make $a_i$ state independent with probability $1/2$ of each action; the two likelihood ratios then reduce to the distinct $\Lambda_{i-1}(\mathcal N_{i+1}^-,h^-\mid\boldsymbol\alpha_{\leq i-1}(\cdot\mid\mathcal N_{\leq i-1},\Pi_{\leq i-1}))$ and $\Lambda_{i-1}(\mathcal N_{i+1}^{\prime-},h^{\prime-}\mid\boldsymbol\alpha_{\leq i-1}(\cdot\mid\mathcal N_{\leq i-1},\Pi_{\leq i-1}))$.
If $(\mathcal N_{i+1}^-,h^-)=(\mathcal N_{i+1}^{\prime-},h^{\prime-})$, evaluate the polynomial at the formal assignment $\mathbb P(a_i=1\mid L)=1/3$ and $\mathbb P(a_i=1\mid H)=2/3$ for every observation.
Omitting $a_i$, observing $a_i=1$, and observing $a_i=0$ then multiply the likelihood ratio by $1$, $2$, and $1/2$, respectively, so the two observations again have different likelihood ratios.
Thus, the polynomial in the above equation is not identically zero.

There are only finitely many pairs of observations.
A finite union of zero sets of nonzero polynomials has empty interior.
Hence, the implementable open set contains a point with full support at which every observation has a distinct likelihood ratio.
\end{proof}

We use the following uniform bound when agent $i$ is assigned an arbitrarily weak private signal.
Define
\[
 V_N^{\leq i}(\mathcal N_{\leq i},\Pi_N\mid\Pi_{\leq i})
 :=\sum_{\bm a^i\in A^i}\sum_{s_N\in S_N}
 \max\left\{
 \begin{aligned}
 &(1-\mu_N)\alpha_{\leq i}^L(\bm a^i\mid\mathcal N_{\leq i},\Pi_{\leq i})\pi_N(s_N\mid L),\\
 &\mu_N\alpha_{\leq i}^H(\bm a^i\mid\mathcal N_{\leq i},\Pi_{\leq i})\pi_N(s_N\mid H)
 \end{aligned}
 \right\}.
\]
Thus, $V_N^{\leq i}$ is agent $N$'s expected payoff when she observes the complete history through agent $i$ and her current private signal.

The next two lemmas bound changes in this value directly from $\ell_1$ differences.
\begin{lemma}
\label{lem:weak-appendage}
Fix $i\in\{2,\ldots,N-1\}$, $\mathcal N_{\leq i}$, and $\Pi_{\leq i}$.
If $(1/2)\sum_{s_i\in S_i}|\pi_i(s_i\mid L)-\pi_i(s_i\mid H)|\leq\eta$, then, for every $\Pi_N$, $V_N^{\leq i}(\mathcal N_{\leq i},\Pi_N\mid\Pi_{\leq i})-V_N^{\leq i-1}(\mathcal N_{\leq i-1},\Pi_N\mid\Pi_{\leq i-1})\leq2\eta$.
\end{lemma}

\begin{proof}[Proof of Lemma~\ref{lem:weak-appendage}]
For each $\bm a^{i-1}$, let $K_\theta(\cdot\mid\bm a^{i-1})$ be the distribution of $a_i$ induced by agent $i$'s decision rule under state $\theta$.
Replace $K_H$ with $K_L$.
Under the resulting experiment, $a_i$ is conditionally independent of the state given $\bm a^{i-1}$, so its value equals $V_N^{\leq i-1}(\mathcal N_{\leq i-1},\Pi_N\mid\Pi_{\leq i-1})$.
Using $|\max\{x,y\}-\max\{x,y'\}|\leq|y-y'|$, the resulting change in value is at most
\[
\mu_N\sum_{\bm a^{i-1}}
\alpha_{\leq i-1}^H(\bm a^{i-1}\mid\mathcal N_{\leq i-1},\Pi_{\leq i-1})
\sum_{a_i}|K_H(a_i\mid\bm a^{i-1})-K_L(a_i\mid\bm a^{i-1})|.
\]
Since the conditional action is a garbling of agent $i$'s private signal, data processing gives
\[
\sum_{a_i}|K_H(a_i\mid\bm a^{i-1})-K_L(a_i\mid\bm a^{i-1})|
\leq\sum_{s_i\in S_i}|\pi_i(s_i\mid H)-\pi_i(s_i\mid L)|\leq2\eta.
\]
The result follows because $\mu_N\leq1$ and $\sum_{\bm a^{i-1}}\alpha_{\leq i-1}^H(\bm a^{i-1}\mid\mathcal N_{\leq i-1},\Pi_{\leq i-1})=1$.
\end{proof}
We also bound rare perturbations of $\Pi_i$ in comparisons through agent $i$.
\begin{lemma}
\label{lem:rare-mixture}
Fix $\Pi_{\leq i}$, and let $\Pi_i^\delta$ be obtained by selecting $\Pi_i$ with probability $1-\delta$ and an arbitrary information structure on a disjoint signal set with probability $\delta$, independently of the state.
Set $\Pi_{\leq i}^\delta:=(\Pi_{\leq i-1},\Pi_i^\delta)$.
Then, for every $\mathcal N_{\leq i}$, $\left|V_N^{\leq i}(\mathcal N_{\leq i},\Pi_N\mid\Pi_{\leq i}^\delta)-V_N^{\leq i}(\mathcal N_{\leq i},\Pi_N\mid\Pi_{\leq i})\right|\leq2\delta$.
\end{lemma}
\begin{proof}[Proof of Lemma~\ref{lem:rare-mixture}]
For $\theta\in\{L,H\}$, let $\alpha_\theta^\delta$ and $\alpha_\theta$ denote the distributions of $\bm a^i$ under $\Pi_{\leq i}^\delta$ and $\Pi_{\leq i}$, respectively. Since $\Pi_i^\delta$ selects the original component with probability $1-\delta$, we have $\sum_{\bm a^i} |\alpha_\theta^\delta(\bm a^i)-\alpha_\theta(\bm a^i)|\leq 2\delta$.
Then,
\begin{align*}
& \left|V_N^{\leq i}(\mathcal N_{\leq i},\Pi_N\mid\Pi_{\leq i}^\delta)- V_N^{\leq i}(\mathcal N_{\leq i},\Pi_N\mid\Pi_{\leq i})\right|\\
& \leq (1-\mu_N)\sum_{\bm a^i} |\alpha_L^\delta(\bm a^i)-\alpha_L(\bm a^i)| + \mu_N\sum_{\bm a^i} |\alpha_H^\delta(\bm a^i)-\alpha_H(\bm a^i)|\\
& \leq 2(1-\mu_N)\delta+2\mu_N\delta  = 2\delta.
\end{align*}
\end{proof}

The following lemma extends the strict advantage of $\widehat{\mathcal N}_{\leq i-1}$ to $\widehat{\mathcal N}_{\leq i}$ while preserving observation separation.
\begin{lemma}
\label{lem:one-step}
Fix $i\in\{2,\ldots,N-1\}$.
Suppose that, under $(\mu_j,\Pi_j)_{j<i}$, there exists an information structure $\Pi_N^0$ containing conclusive signals for both states such that (i) $\boldsymbol\alpha_{\leq i-1}(\cdot\mid\widehat{\mathcal N}_{\leq i-1},\Pi_{\leq i-1})$ is observation separating and (ii), for every $\mathcal N$ whose first deviation from $\widehat{\mathcal N}$ is at most $i-1$, $V_N^{\leq i-1}(\widehat{\mathcal N}_{\leq i-1},\Pi_N^0\mid\Pi_{\leq i-1})>V_N^{\leq i-1}(\mathcal N_{\leq i-1},\Pi_N^0\mid\Pi_{\leq i-1})$. 
Then, there exist $(\mu_i,\Pi_i)$ and a $\Pi_N$, which is a modification of $\Pi_N^0$, retaining conclusive signals for both states, such that (i) $\boldsymbol\alpha_{\leq i}(\cdot\mid\widehat{\mathcal N}_{\leq i},\Pi_{\leq i})$ is observation separating and (ii), for every $\mathcal N$ whose first deviation from $\widehat{\mathcal N}$ is at most $i$, $V_N^{\leq i}(\widehat{\mathcal N}_{\leq i},\Pi_N\mid\Pi_{\leq i})
 >V_N^{\leq i}(\mathcal N_{\leq i},\Pi_N\mid\Pi_{\leq i})$.
\end{lemma}
\begin{proof}[Proof of Lemma~\ref{lem:one-step}]
If there exists a network whose first deviation from $\widehat{\mathcal N}$ is at most $i-1$, let $3\eta>0$ be the minimum of
\begin{align*}
V_N^{\leq i-1}(\widehat{\mathcal N}_{\leq i-1},\Pi_N^0\mid\Pi_{\leq i-1})
-
V_N^{\leq i-1}(\mathcal N_{\leq i-1},\Pi_N^0\mid\Pi_{\leq i-1})
\end{align*}
over all such $\mathcal N_{\leq i-1}$; otherwise, fix any $\eta>0$.
In the former case, $\eta>0$ by the assumption.

Fix any $h^*\in A^{\widehat{\mathcal N}_i}$. 
Use the probe construction in the proof of Lemma~\ref{lem:generic-extension} to choose $\mu_i$ and a probe component $\Pi_i^p$ with signals $p^-$ and $p^+$ such that, for some sufficiently small $\varepsilon>0$, $LR_i(p^-)=1-\varepsilon$, $LR_i(p^+)=1+\varepsilon$, and $(1/2)\sum_{s_i\in S_i}|\pi_i^p(s_i|L)-\pi_i^p(s_i|H)|<\eta/2.$

Write $\bm a^{i-1}=(a_1,\ldots,a_{i-1})$. 
Since the action distribution under $\widehat{\mathcal N}_{\leq i-1}$ is full support, there exists $\bm a^*\in A^{i-1}$ such that $a^*_{\widehat{\mathcal N}_i}=h^*.$
Let
\begin{align*}
\kappa^*=\Lambda_{i-1}\left(\{1,\ldots,i-1\},\bm a^*\mid\boldsymbol\alpha_{\leq i-1}(\cdot\mid\widehat{\mathcal N}_{\leq i-1},\Pi_{\leq i-1})\right).
\end{align*}
We add a test signal $t_i$ to agent $N$ satisfying $LR_N(t_i)=(1-\mu_N)/(\mu_N\kappa^*)$.

Consider now a network $\mathcal N$ whose first deviation from $\widehat{\mathcal N}$ is exactly $i$. 
The distribution of $\bm a^{i-1}$ is then the same under $\mathcal N_{\leq i}$ and $\widehat{\mathcal N}_{\leq i}$. 
Conditional on $(\bm a^*,t_i)$, the two probe signals $p^-$ and $p^+$ therefore induce different actions under $\widehat{\mathcal N}_i$, while they induce the same action under $\mathcal N_i$. Under $\mathcal N_i$, the conditional expected payoff from the two probe signals is $1/2$, whereas under $\widehat{\mathcal N}_i$, each signal induces the action that is more likely to be correct, and hence the conditional expected payoff is strictly greater than $1/2$. Since under $\mathcal N_i$ the action $a_i$ is a state-independent garbling of $\bm a^{i-1}$, this gain is not offset elsewhere. Since $(\bm a^*,t_i)$ occurs with positive probability, this yields a strictly positive gain for $\widehat{\mathcal N}_i$ relative to $\mathcal N_i$.

We now give $t_i$ a small positive probability, taking the corresponding probability from the conclusive signals of $\Pi_N^0$. 
Denote the resulting information structure by $\Pi_N$. 
This preserves conclusive signals for both states. 
Since there are only finitely many networks whose first deviation from $\widehat{\mathcal N}$ is at most $i-1$, the probability of $t_i$ can be chosen sufficiently small so that
\begin{align*}
V_N^{\leq i-1}(\widehat{\mathcal N}_{\leq i-1},\Pi_N \mid\Pi_{\leq i-1})
-
V_N^{\leq i-1}(\mathcal N_{\leq i-1},\Pi_N \mid\Pi_{\leq i-1})
> 2\eta
\end{align*}
for every such $\mathcal N$.

We next add a generic perturbation to the probe component. 
Let $\Pi_i=(1-\delta)\Pi_i^p+\delta\Pi_i^g$, where $\Pi_i^g$ is supported on a signal space disjoint from that of $\Pi_i^p$. 
By the open-set argument in the proof of Lemma~\ref{lem:generic-extension} and Lemma~\ref{lem:rare-mixture}, $\Pi_i^g$ and $\delta>0$ can be chosen so that $\boldsymbol\alpha_{\leq i}(\cdot\mid\widehat{\mathcal N}_{\leq i},\Pi_{\leq i})$ is observation separating, $(1/2)\sum_{s_i\in S_i}|\pi_i(s_i|L)-\pi_i(s_i|H)|<\eta,$ and 
\begin{align*}
V_N^{\leq i}(\widehat{\mathcal N}_{\leq i},\Pi_N \mid\Pi_{\leq i})
-
V_N^{\leq i}(\mathcal N_{\leq i},\Pi_N \mid\Pi_{\leq i})
>0,
\end{align*}
for every $\mathcal N_{\leq i}$ whose first deviation from $\widehat{\mathcal N}$ is exactly $i$. 
Take $\Pi_N$ and $\Pi_i$ as above.

It remains to consider networks whose first deviation from $\widehat{\mathcal N}$ is at most $i-1$. 
In this case, it follows from Lemma~\ref{lem:weak-appendage} that
\begin{align*}
& V_N^{\leq i}(\mathcal N_{\leq i},\Pi_N\mid\Pi_{\leq i})- V_N^{\leq i}(\widehat{\mathcal N}_{\leq i},\Pi_N\mid\Pi_{\leq i}) \\
& \leq  V_N^{\leq i}(\mathcal N_{\leq i},\Pi_N\mid\Pi_{\leq i}) - V_N^{\leq i-1}(\mathcal N_{\leq i-1},\Pi_N\mid\Pi_{\leq i-1}) \\
& \qquad \qquad \qquad \qquad + V_N^{\leq i-1}(\mathcal N_{\leq i-1},\Pi_N\mid\Pi_{\leq i-1})-V_N^{\leq i-1}(\widehat{\mathcal N}_{\leq i-1},\Pi_N\mid\Pi_{\leq i-1}) \\
& \leq \sum_{s_i\in S_i}|\pi_i(s_i|L)-\pi_i(s_i|H)|+ V_N^{\leq i-1}(\mathcal N_{\leq i-1},\Pi_N\mid\Pi_{\leq i-1})-V_N^{\leq i-1}(\widehat{\mathcal N}_{\leq i-1},\Pi_N\mid\Pi_{\leq i-1}) \\
& < \sum_{s_i\in S_i}|\pi_i(s_i|L)-\pi_i(s_i|H)|-2\eta\\
& < 0.
\end{align*}
Thus, the desired strict inequality holds for every $\mathcal N_{\leq i}$ whose first deviation from $\widehat{\mathcal N}$ is at most $i$.
\end{proof}

The following lemma uses observation separation to make full observation by agent $N$ strictly better than every proper subset of predecessors while preserving the previously established comparisons.
\begin{lemma}
\label{lem:full-observation}
Let $\widehat{\mathcal N}$ be a network with $\widehat{\mathcal N}_N=\{1,2,\dots,N-1\}$. 
For given $(\mu_j,\Pi_j)_{j<N}$, $\mu_N$, and $\Pi_N^0$ which contains a conclusive signal for both states, suppose that (i) $\boldsymbol\alpha_{\leq N-1}(\cdot\mid\widehat{\mathcal N}_{\leq N-1},\Pi_{\leq N-1})$ is observation separating and (ii) for every $\mathcal N$ whose first deviation from $\widehat{\mathcal N}$ is at most $N-1$, $V_N^{\leq N-1}(\widehat{\mathcal N}_{\leq N-1},\Pi_N^0\mid\Pi_{\leq N-1})>V_N^{\leq N-1}(\mathcal N_{\leq N-1},\Pi_N^0\mid\Pi_{\leq N-1})$. 
Then, there exists $\Pi_N$ such that for every $\mathcal N \neq \widehat{\mathcal N}$, $V_N(\widehat{\mathcal N},\mu,\Pi)>V_N({\mathcal N},\mu,\Pi)$, where $(\mu,\Pi)=(\mu_1,\dots,\mu_N,\Pi_1,\dots,\Pi_N)$.
\end{lemma}

\begin{proof}[Proof of Lemma~\ref{lem:full-observation}]
Fix a network $\mathcal{N}$ with $\mathcal{N}_i=\widehat{\mathcal N}_i$ for all $i<N$ and $\mathcal N_N\subsetneq\widehat{\mathcal N}_N$. 
Note that, conditional on each signal received by agent $N$, agent $N$'s expected payoff is weakly higher under $\widehat{\mathcal N}$ than under $\mathcal N$, because agent $N$ can (partially) ignore the information received under $\widehat{\mathcal N}$.

Take distinct full histories $\bm a,\bm a'\in A^{N-1}$ with $\bm a_{\mathcal N_N}=(\bm a')_{\mathcal N_N}$.
Full support and observation separation give
\[
 \Lambda_{N-1}\left(\widehat{\mathcal N}_N,\bm a \mid
 \boldsymbol\alpha_{\leq N-1}(\cdot\mid\widehat{\mathcal N}_{\leq N-1},\Pi_{\leq N-1})\right)
 \neq
 \Lambda_{N-1}\left(\widehat{\mathcal N}_N,\bm a'\mid
 \boldsymbol\alpha_{\leq N-1}(\cdot\mid\widehat{\mathcal N}_{\leq N-1},\Pi_{\leq N-1})\right).
\]
Relabel the histories so that, denoting their likelihood ratios by $\kappa$ and $\kappa'$ respectively, we have $\kappa<\kappa'$, and introduce a new signal $s_{\mathcal N_N}\in S_N$ so that $(1-\mu_N)/(\mu_N\kappa')<LR_N(s_{\mathcal N_N})<(1-\mu_N)/(\mu_N\kappa)$.
Conditional on $s_{\mathcal N_N}$, full observation network $\widehat{\mathcal N}_N$ induces action $0$ at $\bm a$ and action $1$ at $\bm a'$, whereas observation of $\mathcal N_N$ pools these histories.
Thus, $\widehat{\mathcal N}_N$ gives a strictly higher expected payoff conditional on $s_{\mathcal N_N}$.

Add one such signal to $\Pi_N^0$ for each $\mathcal{N}$ with $\mathcal{N}_i=\widehat{\mathcal N}_i$ for all $i<N$ and  $\mathcal N_N\subsetneq\widehat{\mathcal N}_N$, assigning it sufficiently small probabilities under both states and subtracting these probabilities from the corresponding conclusive signals.
Denote the information structure obtained by modifying $\Pi_N^0$ in this way by $\Pi_N$. 
By choosing the probabilities of the added signals sufficiently small, we can ensure that, for every $\mathcal N$ whose first deviation from $\widehat{\mathcal N}$ is at most $N-1$,
\begin{align*}
    V_N^{\leq N-1}(\widehat{\mathcal N}_{\leq N-1},\Pi_N\mid\Pi_{\leq N-1})
    >
    V_N^{\leq N-1}(\mathcal N_{\leq N-1},\Pi_N\mid\Pi_{\leq N-1}).
\end{align*}

Finally, show that $V_N(\widehat{\mathcal N},\mu,\Pi)>V_N({\mathcal N},\mu,\Pi)$ for all $\mathcal N \neq \widehat{\mathcal N}$. If the first deviation from $\widehat{\mathcal N}$ is $N$, this strict inequality holds by the above construction of $\Pi_N$. If the first deviation from $\widehat{\mathcal N}$ is less than $N$, define $\mathcal N^+:=(\mathcal N_{\leq N-1},\widehat{\mathcal N}_N)$. Then,
\begin{align*}
    V_N(\mathcal N,\mu,\Pi)
    & \leq V_N(\mathcal N^+,\mu,\Pi) \\
    & = V_N^{\leq N-1}(\mathcal N_{\leq N-1},\Pi_N\mid\Pi_{\leq N-1}) \\
    & < V_N^{\leq N-1}(\widehat{\mathcal N}_{\leq N-1},\Pi_N\mid\Pi_{\leq N-1}) = V_N(\widehat{\mathcal N},\mu,\Pi).
\end{align*}
This completes the proof.
\end{proof}

\begin{proof}[Proof of Theorem~\ref{thm:characterization}]
Suppose first that $\widehat{\mathcal N}_N=\{1,\ldots,N-1\}$.
Fix any $\mu_N\in(0,1)$.
Initialize $\Pi_N$ with one signal conclusive for each state, each having probability one in that state.
Set $\mu_1=1/2$ and give agent $1$ an arbitrarily weak binary information structure with likelihood ratios $1-\varepsilon_1$ and $1+\varepsilon_1$.
Then, it follows that $\Lambda_1(\emptyset,\emptyset\mid\boldsymbol\alpha_{\leq 1})=1$, $\Lambda_1(\{1\},0\mid\boldsymbol\alpha_{\leq 1})=1-\varepsilon_1$, and $\Lambda_1(\{1\},1\mid\boldsymbol\alpha_{\leq 1})=1+\varepsilon_1$. Hence, $\boldsymbol\alpha_{\leq1}(\cdot\mid\widehat{\mathcal N}_{\leq1},\Pi_{\leq1})$ is observation separating.

Apply Lemma~\ref{lem:one-step} for $i=2,\ldots,N-1$.
The test signals can be assigned sufficiently small probabilities that the conclusive signals retain positive probabilities and $\Pi_N$ remains finite.
At $i=N-1$, $V_N^{\leq N-1}$ equals agent $N$'s actual payoff whenever she observes every predecessor.

Then, we can apply Lemma~\ref{lem:full-observation} while preserving these strict comparisons. We have  $V_N(\widehat{\mathcal N},\mu,\Pi)>V_N({\mathcal N},\mu,\Pi)$ for every $\mathcal N \neq \widehat{\mathcal N}$.

Conversely, suppose that $j\notin\widehat{\mathcal N}_N$ for some $j<N$, and define
$\mathcal N^+:=(\widehat{\mathcal N}_{\leq N-1},\widehat{\mathcal N}_N\cup\{j\})$.
For every informational environment $(\mu,\Pi)$, $V_N(\widehat{\mathcal N},\mu,\Pi)\leq V_N(\mathcal N^+,\mu,\Pi)$, because agent $N$ can ignore the additional action.
Hence, $\widehat{\mathcal N}$ cannot be a unique maximizer.
This proves necessity and sufficiency.
\end{proof}

\subsection{Details for Example~\ref{ex:four-agent-shield}}
\label{app:detailed-four-agent}

\begin{proof}[Proof of Example~\ref{ex:four-agent-shield}]
The target network and priors are
\[
 (\mathcal N_1,\mathcal N_2,\mathcal N_3,\mathcal N_4)
 =\bigl(\emptyset,\{1\},\{1\},\{1,2,3\}\bigr),
 \qquad
 (\mu_1,\mu_2,\mu_3,\mu_4)
 =\left(1/2,1/3,1/3,1/2\right).
\]
The heterogeneous priors place agents $2$ and $3$ at selected action cutoffs, so their probe signals induce different actions under the target but the same action under the alternatives; the test signals $t_2$ and $t_3$ make these distinctions valuable to agent $4$.
The auxiliary components $\pi_2^T$ and $\pi_3^T$ create full support and distinct likelihood ratios without overturning these comparisons, after which the signals $u_C$ make full observation strictly better than every proper subset.

Give agent $1$ the binary information structure $\Pi_1=(S_1,\pi_1)$, where $S_1=\{l_1,h_1\}$, $\pi_1(l_1\mid L)=\pi_1(h_1\mid H)=2/3$, and $\pi_1(h_1\mid L)=\pi_1(l_1\mid H)=1/3$.
Agent $1$ chooses action $0$ after $l_1$ and action $1$ after $h_1$.
Thus, the likelihood ratios associated with observing $a_1=0$, no predecessor, and $a_1=1$ are $1/2$, $1$, and $2$, respectively.

Set $\delta_2=1/10$.
Give agent $2$ the information structure $\Pi_2=(S_2,\pi_2)$ with $S_2=\{p_2^-,p_2^+,d_2,m_2,c_2\}$ where the disclosure rule is the mixture $\pi_2=(1-\delta_2)\pi_2^P+\delta_2\pi_2^T$ of the following component disclosure rules:
\[
\begin{array}{c|cc}
\pi_2^P&p_2^-&p_2^+\\ \hline
\pi_2^P(\cdot\mid L)&1/2&1/2\\
\pi_2^P(\cdot\mid H)&1/4&3/4\\
\end{array}
\qquad
\begin{array}{c|ccc}
\pi_2^T&d_2&m_2&c_2\\ \hline
\pi_2^T(\cdot\mid L)&23/36&14/45&1/20\\
\pi_2^T(\cdot\mid H)&23/180&28/45&1/4\\
\end{array}
\]
Since agent $2$'s prior odds are $1/2$, her action table is
\begin{center}
\begin{tabular}{c@{\qquad}ccccc}
\toprule
 & $d_2$ & $p_2^-$ & $p_2^+$ & $m_2$ & $c_2$ \\
\midrule
$\mathcal N_2=\{1\},\ a_1=0$ & $0$ & $0$ & $0$ & $0$ & $1$ \\
$\mathcal N_2=\{1\},\ a_1=1$ & $0$ & $0$ & $1$ & $1$ & $1$ \\
$\mathcal N_2=\emptyset$     & $0$ & $0$ & $0$ & $0$ & $1$ \\
\bottomrule
\end{tabular}
\end{center}
The action induced by a given private signal differs across the two networks only after $a_1=1$.

Give agent $4$ a test signal $t_2$ satisfying $\pi_4(t_2\mid L)=w_2=1/10$ and $LR_4(t_2)=1/2$.
The action table implies that the difference between agent $4$'s expected payoff under the target neighborhood $\mathcal N_2=\{1\}$ and the alternative $\mathcal N_2=\emptyset$ is
\[
G_2:=\frac{w_2(45+11\delta_2)}{1080}=\frac{461}{108000}>0.
\]
This expression accounts both for the benefit of making agent $2$'s action distinguish the two probe signals under the target and for the possible informational advantage generated by the strong signal $c_2$ under the alternative.

Under the target network, the state-contingent distributions of $(a_1,a_2)$ are
\begin{center}
\begin{tabular}{c@{\qquad}cc}
\toprule
$(a_1,a_2)$ & $\mathbb P(a_1,a_2\mid L)$ & $\mathbb P(a_1,a_2\mid H)$ \\
\midrule
$(0,0)$ & $199/300$ & $13/40$ \\
$(0,1)$ & $1/300$ & $1/120$ \\
$(1,0)$ & $37/216$ & $107/675$ \\
$(1,1)$ & $35/216$ & $343/675$ \\
\bottomrule
\end{tabular}
\end{center}
The likelihood ratio under the target neighborhood $\{1\}$ equals $2$ after $a_1=1$ and $1/2$ after $a_1=0$.
Using the displayed probabilities, the likelihood ratio after every possible observation under an alternative neighborhood $\emptyset$, $\{2\}$, or $\{1,2\}$ is at most one or at least $5/2$.

Let $0<\varepsilon_3<1/5$ and $\delta_3\in(0,1)$.
Give agent $3$ the information structure $\Pi_3=(S_3,\pi_3)$, where $S_3=\{p_3^-,p_3^+,d_3,m_3,c_3\}$ and the disclosure rule is the mixture $\pi_3=(1-\delta_3)\pi_3^P+\delta_3\pi_3^T$ of the following component disclosure rules with disjoint supports:
\[
\begin{array}{c|cc}
\pi_3^P&p_3^-&p_3^+\\ \hline
\pi_3^P(\cdot\mid L)&1/2&1/2\\
\pi_3^P(\cdot\mid H)&(1-\varepsilon_3)/2&(1+\varepsilon_3)/2\\
\end{array}
\qquad
\begin{array}{c|ccc}
\pi_3^T&d_3&m_3&c_3\\ \hline
\pi_3^T(\cdot\mid L)&2/3&17/60&1/20\\
\pi_3^T(\cdot\mid H)&2/15&17/30&3/10\\
\end{array}
\]
After the target observation $a_1=1$, the two probe signals induce different actions by agent $3$.
Under each alternative neighborhood $\emptyset$, $\{2\}$, or $\{1,2\}$, they induce the same action after every possible observed action vector.
This follows from the likelihood ratios above and the restriction $\varepsilon_3<1/5$.

Consider the history $(a_1,a_2)=(1,0)$ under the target network.
Its conditional probability under state $L$ is $\mathbb P((a_1,a_2)=(1,0)\mid L)=37/216$, and its likelihood ratio is $856/925$.
Give agent $4$ a test signal $t_3$ satisfying $\pi_4(t_3\mid L)=w_3$ and $LR_4(t_3)=r_3:=925/856$.
Before adding the auxiliary component $\pi_3^T$, the history $(a_1,a_2)=(1,0)$ and signal $t_3$ generate a payoff advantage of $(37/216)\varepsilon_3w_3/4=(37/864)\varepsilon_3w_3$ for the target relative to each of agent $3$'s three alternative neighborhoods.
Indeed, with $\delta_3=0$, every alternative pools the two probe signals, so its information is a garbling of the target's.

The component $\pi_3^T$ is selected with probability $\delta_3$ under either state.
The information structures before and after this perturbation can be coupled to coincide whenever $\pi_3^P$ is selected, so a fixed network's payoff changes by at most $2\delta_3$ and the difference in expected payoff between the target and an alternative changes by at most $4\delta_3$.
Hence, every comparison with $\mathcal N_2=\{1\}$ and $\mathcal N_3\neq\{1\}$ yields a payoff difference in favor of the target of at least $\Delta_3:=(37/864)\varepsilon_3w_3-4\delta_3$.

We must also preserve the payoff comparison between the two neighborhoods of agent $2$.
Agent $3$'s signal distributions satisfy
\[
 \frac12\sum_{s_3\in S_3}|\pi_3(s_3\mid L)-\pi_3(s_3\mid H)|
 \leq\frac{\varepsilon_3}{2}+\delta_3.
\]
The additional observation of $a_3$ can increase agent $4$'s expected payoff by at most twice the left-hand side under either neighborhood of agent $2$, so the target's payoff advantage can fall by at most $\varepsilon_3+2\delta_3$.
Moreover, the entire payoff contribution of $t_3$ is at most $w_3(1+r_3)/2$.
Thus, every network with $\mathcal N_2=\emptyset$ yields a payoff difference in favor of the target of at least $\Delta_2:=G_2-\varepsilon_3-2\delta_3-(1+r_3)w_3/2$.
It is therefore enough to choose the parameters so that
\[
 \varepsilon_3+2\delta_3+\frac{1+r_3}{2}w_3<G_2,
 \qquad
 4\delta_3<\frac{37}{864}\varepsilon_3w_3.
\]
These inequalities require the probe signals to be sufficiently weak to preserve $G_2$, the test signal $t_3$ to be sufficiently rare, and the mixing probability $\delta_3$ to be small relative to the payoff gain generated by $t_3$.

For concreteness, take $\varepsilon_3=1/1000$, $w_3=1/500$, and $\delta_3=10^{-9}$.
Indeed, $G_2-\varepsilon_3-(1+r_3)w_3/2=5491/4622400>0$ and $2\delta_3<5491/4622400$, while $\delta_3<37/1728000000$, so the above inequalities hold.
Direct substitution also gives distinct likelihood ratios for all 27 possible observations through agent $3$.

Let $\Delta:=\min\{\Delta_2,\Delta_3\}$.
For the displayed values, $\Delta=\Delta_3=4409/54000000000>0$.
Among the $2\cdot4=8$ networks in which agent $4$ observes all three predecessors, four have $\mathcal N_2=\emptyset$ and yield an advantage of at least $\Delta_2$, while three of the four with $\mathcal N_2=\{1\}$ have $\mathcal N_3\neq\{1\}$ and yield an advantage of at least $\Delta_3$; the remaining network is the target.
Thus, the target is strictly best among these networks.

It remains to compare full observation with each of the seven strict subsets $C\subsetneq\{1,2,3\}$ that agent $4$ may observe.
Write $\lambda_x$ for the likelihood ratio of target history $x$.
For each $C$, choose two action histories and a signal likelihood ratio as follows:
\[
\begin{array}{c|ccc|c}
C&x_C&y_C&LR_4(u_C)&
LR_4(u_C)\lambda_{x_C}<1<LR_4(u_C)\lambda_{y_C}\\ \hline
\emptyset&(0,0,0)&(0,1,1)&1&\lambda_{000}<1<\lambda_{011}=15\\
\{1\}&(0,0,0)&(0,1,1)&1&\lambda_{000}<1<\lambda_{011}=15\\
\{2\}&(0,0,0)&(0,0,1)&1&\lambda_{000}<1<\lambda_{001}=585/199\\
\{3\}&(1,0,1)&(0,1,1)&1&\lambda_{101}<1<\lambda_{011}=15\\
\{1,2\}&(0,0,0)&(0,0,1)&1&\lambda_{000}<1<\lambda_{001}=585/199\\
\{1,3\}&(0,0,0)&(0,1,0)&1&\lambda_{000}<1<\lambda_{010}\\
\{2,3\}&(1,1,1)&(0,1,1)&1/5&\lambda_{111}/5<1<\lambda_{011}/5=3
\end{array}
\]
The two histories in each row induce the same observation when agent $4$ observes only $C$, and the inequalities in the last column follow from the action distributions and parameter values above.
Conditional on signal $u_C$, the posterior odds at $x_C$ and $y_C$ lie on opposite sides of one.
Agent $4$ therefore chooses different actions at these histories under full observation, whereas she cannot distinguish them when she observes only $C$.
Because both histories occur with positive probability under both states, full observation yields a strictly higher expected payoff conditional on $u_C$.

Give agent $4$ the information structure $\Pi_4=(S_4,\pi_4)$, where $S_4=\{t_2,t_3,t^-,t^+\}\cup \{u_C:C\subsetneq\{1,2,3\}\}$ and the disclosure rule $\pi_4$ is
\begingroup
\setlength{\arraycolsep}{2.5pt}
\[
\begin{array}{c|ccccc}
\pi_4&t_2&t_3&u_C&t^-&t^+\\ \hline
\pi_4(\cdot\mid L)
&w_2&w_3&\eta&1-w_2-w_3-7\eta&0\\
\pi_4(\cdot\mid H)
&w_2/2&r_3w_3&\eta LR_4(u_C)&0&1-w_2/2-r_3w_3-(31/5)\eta
\end{array}
\]
\endgroup
For each strict subset $C$, the $u_C$ column specifies the probability of the corresponding signal $u_C$.
Since six of these signals have likelihood ratio one and one has likelihood ratio $1/5$, the probabilities in the table sum to one under each state.
The maximum reduction in any payoff difference established above due to the seven signals $u_C$ is $(33/5)\eta$.
Hence any $0<\eta<(5/33)\Delta$ preserves all of these strict comparisons; we take $\eta=10^{-9}$, which is below $(5/33)\Delta=4409/356400000000$.
The remaining probabilities assigned to $t^-$ and $t^+$ are positive, and these signals are conclusive for states $L$ and $H$, respectively.
If a competing network differs from the target network in the neighborhood of some agent $i<4$, giving agent $4$ all predecessor actions weakly increases her expected payoff, and the resulting full-observation network is already strictly worse than the target network by the preceding comparison.
Together with the signals $u_C$ and the restriction $\eta<(5/33)\Delta$, this shows that the target strictly dominates every competitor in which agent $4$ observes a strict subset of her predecessors.

We have therefore constructed one informational environment under which the target network strictly outperforms every other network.
\end{proof}

\subsection{Proof of Proposition~\ref{prop:welfare-characterization}}
Fix a network $\widehat{\mathcal N}$ satisfying $\widehat{\mathcal N}_N=\{1,\ldots,N-1\}$. 
We follow the two-stage proof of Theorem~\ref{thm:characterization}, except that the first stage must control weighted welfare. 
The construction behind Theorem~\ref{thm:characterization} gives agent $N$ a strict gain under $\widehat{\mathcal N}$ relative to any alternative but does not control the other agents' payoffs: relative to an alternative neighborhood for agent $i$, $\widehat{\mathcal N}$ may impose a payoff loss on agent $i$ that may offset agent $N$'s gain. 
We therefore need a local comparison in which agent $N$'s gain is of larger order than agent $i$'s loss. 
The following lemma provides this comparison and is the welfare counterpart of Lemma~\ref{lem:one-step}. 
Throughout this subsection, all information structures have finite signal spaces.

Set $\mu_N=1/2$ and, for every $j<N$, let $\omega_j:=2^{j-1}$; for $\mathcal J\subseteq\{1,\ldots,N-1\}$, define $\omega(\mathcal J):=\sum_{j\in\mathcal J}\omega_j$ and $\Omega_i:=\sum_{j<i}\omega_j=\omega_i-1$.
The construction is indexed by $t\in(0,1)$, and we consider the limit as
$t\searrow0$.
For fixed $i$ and $k\in\{i,N\}$, let $V_k(\mathcal N_i;t)$ denote agent $k$'s expected payoff under $\widehat{\mathcal N}$ when only agent $i$'s neighborhood is replaced by $\mathcal N_i$, with the information environments held fixed.
Similarly, let $\bar V_k(\mathcal N_i;t)$ denote agent $k$'s expected payoff when agent $i$'s neighborhood is $\mathcal N_i$ and agent $N$'s neighborhood is $\{1,\ldots,i\}$, with all other neighborhoods and information environments held fixed at those under $\widehat{\mathcal N}$.
For any $i\in\{2,\ldots,N-1\}$ and $\bm a^{i-1}\in A^{i-1}$, define $\mathcal J(\bm a^{i-1}):=\{j<i:a_j=1\}$ and $w(\bm a^{i-1}):=\omega(\mathcal J(\bm a^{i-1}))$; when the history is clear, write simply $\mathcal J$ and $w$. 
Also, define $b_i:=\omega(\widehat{\mathcal N}_i)$ and $\bm a^{i-1,*}=(a_1^*,\ldots,a_{i-1}^*)$ by $a_j^*=1$ if $j\in\widehat{\mathcal N}_i$ and $a_j^*=0$ otherwise. 
Thus, $b_i=w(\bm a^{i-1,*})$. 
For any fixed $M\geq2$, define $\rho_i(t):=t^{\Omega_i+1/2+Mb_i}$.

\begin{lemma}
\label{lem:local-neighborhood-comparison}
Fix $i\in\{2,\ldots,N-1\}$ and $M\geq2$.
For each $t\in(0,1)$, let $(a_j)_{j<i}$ be conditionally
independent random variables taking values in $\{0,1\}$, with $\mathbb P(a_j=1\mid L)=t^{(M+1)\omega_j}$ and $\mathbb P(a_j=1\mid H)=t^{M\omega_j}$.
Then, there exists $\bar t\in(0,1)$ such that, for every
$t\in(0,\bar t)$, if the actions of agents $1,\ldots,i-1$
have the joint distribution $(a_j)_{j<i}$ specified above under $(\mu_j,\Pi_j)_{j<i}$, there exist
$\mu_i(t)$, $\Pi_i(t)$, and $\Pi_N(t)$ such that
\[
 \lambda_i\bigl[\bar{V}_i(\widehat{\mathcal N}_i;t)-\bar{V}_i(\mathcal N_i;t)\bigr]
 +\lambda_N\bigl[\bar{V}_N(\widehat{\mathcal N}_i;t)-\bar{V}_N(\mathcal N_i;t)\bigr]>0
\]
for every $\mathcal N_i\neq\widehat{\mathcal N}_i$.\footnote{Note that $\bar{V}_i$ and $\bar{V}_N$ do not depend on $\mu_{i+1},\dots,\mu_{N-1}$ and $\Pi_{i+1},\dots,\Pi_{N-1}.$}
\end{lemma}

\begin{proof}[Proof of Lemma~\ref{lem:local-neighborhood-comparison}]
The likelihood ratios associated with agent $j$'s actions are
\[
 \frac{\mathbb P(a_j=1\mid H)}{\mathbb P(a_j=1\mid L)}=t^{-\omega_j},
 \qquad
 l_j^0(t):=\frac{\mathbb P(a_j=0\mid H)}{\mathbb P(a_j=0\mid L)}=1-t^{M\omega_j}+o(t^{M\omega_j}).
\]
Moreover, $\omega(\mathcal J)=\omega(\mathcal J')$ if and only if $\mathcal J=\mathcal J'$.

Choose $\mu_i(t)\in(0,1)$ so that $(1-\mu_i(t))/\mu_i(t)=t^{-(\Omega_i+1/2)}$.
If $\widehat{\mathcal N}_i\neq\{1,\dots,i-1\}$, let $\bar\omega_i:=\max_{j\notin\widehat{\mathcal N}_i,\,j<i}\omega_j$ and $\varepsilon_i(t):=t^{M\bar\omega_i}/4$; otherwise, let $\varepsilon_i(t):=t^M/4$.
Fix $c\in(0,1/4)$ and define $\Pi_i(t)=( S_i,\pi_i(t))$, where $S_i=\{p^+,p^-,p^0\}$, and
\[
\begin{alignedat}{3}
&\pi_i(p^+\mid H;t)=c,\quad
&&\pi_i(p^-\mid H;t)=\frac{ct^{1/2}}{2},\quad
&&\pi_i(p^0\mid H;t)=1-c-\frac{ct^{1/2}}{2},\\
&\pi_i(p^+\mid L;t)=\frac{ct^{\Omega_i+1/2-b_i}}{1+\varepsilon_i(t)},\quad
&&\pi_i(p^-\mid L;t)=c,\quad
&&\pi_i(p^0\mid L;t)=1-c-\frac{ct^{\Omega_i+1/2-b_i}}{1+\varepsilon_i(t)}.
\end{alignedat}
\]
Note that $\pi_i(p^0\mid H;t)>0$ and $\pi_i(p^0\mid L;t)>0$ since $t\in (0,1)$.
Then, $LR_i(p^+)=t^{b_i-\Omega_i-1/2}(1+\varepsilon_i(t))$, $LR_i(p^-)=t^{1/2}/2$, and $LR_i(p^0)=1+o(1)$. 
Since the likelihood ratio of any observed history is at most $t^{-\Omega_i}$, the posterior odds for agent $i$ when she observes $p^-$ are at most $t/2$, and those when she observes $p^0$ are at most $t^{1/2}(1+o(1))$.
Since there are only finitely many possible neighborhoods, both signals induce action $0$ for every neighborhood for all sufficiently small $t$.

Given a history $\bm{a}^{i-1}$, write $\mathcal{J}(\bm{a}^{i-1})$ as $\mathcal{J}$.  Then, agent $i$'s posterior odds when she receives signal $p^+$ are
\begin{align*}
& \frac{\mu_i(t)}{1-\mu_i(t)}\left(\prod_{j\in \mathcal{N}_i\cap \mathcal{J}}t^{-\omega_j}\prod_{j\in \mathcal{N}_i \backslash \mathcal{J}}l_j^0(t)\right)t^{b_i-\Omega_i-1/2}(1+\varepsilon_i(t))\\
& =t^{b_i-\omega( \mathcal{N}_i\cap \mathcal{J})}(1+\varepsilon_i(t))\prod_{j\in  \mathcal{N}_i \backslash \mathcal{J}}l_j^0(t).
\end{align*}
We now characterize the action induced by $p^+$.
If $\mathcal N_i=\widehat{\mathcal N}_i$ and $\widehat{\mathcal N}_i\not\subseteq\mathcal J$, then $b_i-\omega(\widehat{\mathcal N}_i\cap\mathcal J)\geq 1$.
Hence, the posterior odds are bounded above by $t(1+\varepsilon_i(t))\to 0$ as $t\searrow0$.
Thus, for all sufficiently small $t$, $p^+$ induces action $0$.
If $\mathcal N_i=\widehat{\mathcal N}_i$ and $\widehat{\mathcal N}_i\subseteq\mathcal J$, then $b_i-\omega(\widehat{\mathcal N}_i\cap\mathcal J)=0$ and $\mathcal{N}_i \backslash \mathcal{J}=\emptyset$.
Hence, the posterior odds are $1+\varepsilon_i(t)>1$. 
Thus, $p^+$ induces action $1$.

Now suppose that $\mathcal N_i\neq\widehat{\mathcal N}_i$.
If $\omega(\mathcal N_i\cap\mathcal J)>b_i$, the posterior odds tend to infinity as $t\searrow0$. 
Since there are only finitely many possible neighborhoods of agent $i$, for all sufficiently small $t$, $p^+$ induces action $1$ for all such $\mathcal N_i$. 
If $\omega(\mathcal N_i\cap\mathcal J)<b_i$, the posterior odds tend to zero, and hence $p^+$ induces action $0$.

It remains to consider the case $\mathcal N_i\neq\widehat{\mathcal N}_i$ and $\omega(\mathcal N_i\cap\mathcal J)=b_i=\omega(\widehat{\mathcal N}_i)$.
By uniqueness of subset sums, $\mathcal N_i\cap\mathcal J=\widehat{\mathcal N}_i$.
Thus, $\widehat{\mathcal N}_i\subseteq\mathcal N_i$ and, since $\mathcal N_i\neq\widehat{\mathcal N}_i$, the set $\mathcal D:=\mathcal N_i\setminus\widehat{\mathcal N}_i$ is nonempty. 
Moreover, $\mathcal D\cap\mathcal J=\emptyset$.
Therefore, the posterior odds reduce to $(1+\varepsilon_i(t))\prod_{j\in\mathcal D}l_j^0(t)$.
Since $l_j^0(t)=1-t^{M\omega_j}+o(t^{M\omega_j})$ and $t^{M\omega_j}\geq 4\varepsilon_i(t)$ for every $j\in\mathcal D$, we have, for all sufficiently small $t$, $l_j^0(t)<1-2\varepsilon_i(t)$ for all such $\mathcal{N}_i$ and $j$.
Hence,
\[
(1+\varepsilon_i(t))\prod_{j\in\mathcal D}l_j^0(t) \leq (1+\varepsilon_i(t))(1-2\varepsilon_i(t)) <1.
\]
Thus, $p^+$ also induces action $0$ in this case.

Consider a history $\bm{a}^{i-1}$ such that $\widehat{\mathcal N}_i\neq\mathcal J(\bm{a}^{i-1})$.
Since $p^-$ and $p^0$ always induce action $0$, a difference in agent $i$'s action between $\widehat{\mathcal N}_i$ and an alternative neighborhood $\mathcal N_i\neq\widehat{\mathcal N}_i$ can then occur only after $p^+$. 
We show that $w=\omega(\mathcal{J})>b_i$ is necessary for a difference to occur.
If $\widehat{\mathcal N}_i\subsetneq\mathcal J(\bm{a}^{i-1})$, we have $b_i=\omega(\widehat{\mathcal N}_i)<w$.
If $\widehat{\mathcal N}_i\not\subseteq\mathcal J(\bm{a}^{i-1})$, agent $i$ chooses action $0$ under $\widehat{\mathcal N}_i$. 
Thus, for a difference to occur, agent $i$ must choose action $1$ under $\mathcal N_i$, which requires $w\geq\omega(\mathcal N_i\cap\mathcal J)>b_i$.
For every $w>b_i$, the joint probability of state $L$, a history with $\omega(\mathcal J)=w$, and signal $p^+$ is $O(t^{\Omega_i+1/2-b_i+(M+1)w})$, since the history occurs with probability $O(t^{(M+1)w})$ under $L$, while $p^+$ occurs with probability $O(t^{\Omega_i+1/2-b_i})$, and $\Pr(L)=O(1)$.
Under $H$, the corresponding probability is $O(t^{\Omega_i+1/2+Mw})$, since $\Pr(H)=O(t^{\Omega_i+1/2})$, the history occurs with probability $O(t^{Mw})$, and $\Pr(p^+\mid H)=O(1)$.
Since $w>b_i$,
\[
\frac{t^{\Omega_i+1/2-b_i+(M+1)w}}{\rho_i(t)} =t^{(M+1)(w-b_i)}\to 0,
\qquad
\frac{t^{\Omega_i+1/2+Mw}}{\rho_i(t)} =t^{M(w-b_i)}\to 0.
\]
Thus, both probabilities are $o(\rho_i(t))$.

Finally, consider the remaining history $\bm{a}^{i-1}$ with $\widehat{\mathcal N}_i=\mathcal J(\bm{a}^{i-1})$. Even in this case, a difference in agent $i$'s action can occur only after $p^+$. 
The joint probability of this history and $p^+$ is $O(\rho_i(t))$, by the same calculation.
Moreover, the posterior odds after observing the full history and $p^+$ are $(1+\varepsilon_i(t))\prod_{j\notin\widehat{\mathcal N}_i,\,j<i}l_j^0(t)=1+o(1)$.
Thus, the difference in the conditional expected payoff between the two actions is $o(1)$, and hence the contribution of this history and $p^+$ to the payoff difference is $O(\rho_i(t))\cdot o(1)=o(\rho_i(t))$.
Therefore, uniformly over $\mathcal N_i\neq\widehat{\mathcal N}_i$, we have $\bar{V}_i(\widehat{\mathcal N}_i;t)-\bar{V}_i(\mathcal N_i;t)\geq-o(\rho_i(t))$.

At $\bm a^{i-1,*}$, $\widehat{\mathcal N}_i$ separates $p^+$ from $p^-$ and $p^0$, whereas every $\mathcal N_i\neq\widehat{\mathcal N}_i$ pools them; give agent $N$ a test signal $s_N$ with likelihood ratio $2t^{\Omega_i+1/2}$.
Under $\widehat{\mathcal N}_i$, the event $(\bm a^{i-1,*},a_i=1)$ has likelihood ratio $t^{-(\Omega_i+1/2)}(1+o(1))$, whereas under every $\mathcal N_i\neq\widehat{\mathcal N}_i$, the event $(\bm a^{i-1,*},a_i=0)$ has likelihood ratio $O(t^{-b_i})$.
Thus, after receiving $s_N$, agent $N$ chooses action $1$ in the former event and action $0$ in the latter.
Assign $s_N$ a fixed sufficiently small probability under state $L$ and set $\pi_N(s_N\mid H)=2t^{\Omega_i+1/2}\pi_N(s_N\mid L)$.
The joint probabilities of states $L$ and $H$ with $(\bm a^{i-1,*},p^+,s_N)$ are both $\Theta(\rho_i(t))$ and $\Pr(H \mid \bm a^{i-1,*},p^+,s_N)/\Pr(L\mid \bm a^{i-1,*},p^+,s_N) \to 2$, so this event gives agent $N$ a gain of order $\rho_i(t)$.

Decompose agent $N$'s payoff difference according to agent $i$'s private signal. 
When agent $i$ receives $p^-$ or $p^0$, the posterior odds for agent $N$, after receiving $s_N$, are at most $t$ and $2t^{1/2}(1+o(1))$, respectively. 
Hence, agent $N$ chooses action $0$ under both networks in either case, so these two cases contribute no payoff difference. 
It remains to consider the case that agent $i$ receives  $p^+$.

Among the events in which agent $i$ receives $p^+$ and agent $N$ receives $s_N$, any event other than $(\bm a^{i-1,*},p^+,s_N)$ on which agent $N$ takes different actions under the two networks must have a set of agents choosing action $1$ with weight $w>b_i$. Including the probability of $s_N$, the total ex ante probability of all such events is $o(\rho_i(t))$, uniformly over the alternatives. 
Hence, there are constants $g_{\mathcal N_i}>0$ such that
\[
 \bar V_N(\widehat{\mathcal N}_i;t) - \bar V_N(\mathcal N_i;t) \geq g_{\mathcal N_i}\rho_i(t)-o(\rho_i(t)).
\]
Assign the remaining probabilities under the two states to conclusive signals, completing the construction of $\Pi_N(t)$. 
Combining the bounds for agents $i$ and $N$ and using $\lambda_i,\lambda_N>0$ proves the claim for all sufficiently small $t$.
\end{proof}

We now combine these local constructions across agents, using the same logic for rare mixtures and preservation of payoffs as in Lemmas~\ref{lem:weak-appendage} and~\ref{lem:rare-mixture}, to obtain the informational environment used in the first stage of the sufficiency proof.
\begin{lemma}
\label{lem:payoff-safe-full-observation}
There exist $(\mu^0,\Pi^0)$ such that (i) $W_\lambda(\widehat{\mathcal N},\mu^0,\Pi^0)>W_\lambda(\mathcal N,\mu^0,\Pi^0)$ for every $\mathcal N\neq\widehat{\mathcal N}$ satisfying $\mathcal N_N=\widehat{\mathcal N}_N$, (ii) $\boldsymbol\alpha_{\leq N-1}(\cdot\mid\widehat{\mathcal N}_{\leq N-1},\Pi^0_{\leq N-1})$ is observation separating, and (iii) $\Pi_N^0$ contains a conclusive signal for each state.
\end{lemma}

\begin{proof}[Proof of Lemma~\ref{lem:payoff-safe-full-observation}]
Choose $M>\max\{2,\Omega_{N-1}\}$.
Take $t\in(0,1)$ small enough.
For every $j<N$, let $\Pi_j^{\mathrm b}(t)=(\{0,1\},\pi_j^{\mathrm b}(t))$, where $\pi_j^{\mathrm b}(1\mid L;t)=t^{(M+1)\omega_j}$ and $\pi_j^{\mathrm b}(1\mid H;t)=t^{M\omega_j}$.
The likelihood ratios are $t^{-\omega_j}$ and $l_j^0(t):=\pi_j^{\mathrm b}(0\mid H;t)/\pi_j^{\mathrm b}(0\mid L;t)=1-t^{M\omega_j}+o(t^{M\omega_j})$.
Choose $\mu_j(t)$ such that $(1-\mu_j(t))/\mu_j(t)=t^{-(\Omega_j+1/2)}$.

Agent $1$ chooses action $1$ if and only if she receives signal $1$, since $t^{-\omega_1}=t^{-1}>t^{-1/2}$ and $l_1^0(t)<1<t^{-1/2}$.
Suppose that agent $k=1,\dots,j-1$ follow their private signal $s_k^b$, that is, they choose action $1$ if and only if $s_k^b=1$.
Then, the likelihood ratio of any observation available to agent $j$ lies between $\prod_{k<j}l_k^0(t)$ and $\prod_{k<j}t^{-\omega_k}=t^{-\Omega_j}$, and 
\[
 \frac{t^{-\Omega_j}l_j^0(t)}{t^{-(\Omega_j+1/2)}}=O(t^{1/2})<1,
 \qquad
 \frac{\bigl(\prod_{k<j}l_k^0(t)\bigr)t^{-\omega_j}}{t^{-(\Omega_j+1/2)}}=\Theta(t^{-1/2})>1
\]
for sufficiently small $t$.
Hence, induction gives $a_j=s_j^{\mathrm b}$ for every $j<N$, irrespective of the observed history.

For each $i=2,\ldots,N-1$, apply Lemma~\ref{lem:local-neighborhood-comparison} to the baseline action distributions of agents $1,\ldots,i-1$ and $\widehat{\mathcal N}_i$. 
Accordingly, let $\widetilde\Pi_i(t)$ and $\widetilde\mu_i(t)$ denote the information structure and prior of agent $i$, respectively, obtained following the procedure in the proof of Lemma~\ref{lem:local-neighborhood-comparison}. 
By the proof, $\widetilde\mu_i(t)=\mu_i(t)$. Let $u_i$ denote the test signal for agent $N$ obtained in the same construction. 
Thus, $LR_N(u_i)=2t^{\Omega_i+1/2}$. 
Construct a single information structure $\Pi_N(t)$ for agent $N$ that contains all the test signals $u_2,\dots,u_{N-1}$, assigning each $u_i$ sufficiently small probabilities satisfying $\pi_N(u_i\mid H)=LR_N(u_i)\pi_N(u_i\mid L)$, and assign the remaining probabilities to two conclusive signals.

We next show that, although the network of agent $N$
constructed above differs from that in
Lemma~\ref{lem:local-neighborhood-comparison}, the payoff comparison
established in that lemma remains valid in the present environment.

Fix $i$ and, for $\delta\in[0,1)$, define
$\Pi_i(t,\delta):=(1-\delta)\Pi_i^{\mathrm b}(t)
+\delta\widetilde\Pi_i(t)$, with disjoint signal sets for the two
components. Assign $\Pi_i(t,\delta)$ to agent $i$ and
$\Pi_j^{\mathrm b}(t)$ to every other agent $j<N$.
For sufficiently small $\delta$, all agents $j\neq i,N$
follow their baseline binary signals regardless of the network, by the strict cutoff
inequalities established above. Hence, as $\delta\to0$, agent $N$'s
posterior odds after observing $(\bm a^{N-1},u_k)$ converge to
\[
2t^{\Omega_k+1/2-\sum_{j<N}\omega_j a_j}
\prod_{j<N:a_j=0}l_j^0(t).
\]
Since $\Omega_k=\omega_k-1$ and $l_j^0(t)=1+o(1)$, for sufficiently
small $t$, agent $N$ chooses action $1$ if and only if some
$j\geq k$ chooses action $1$.
We take $\delta\in(0,\bar{\delta}_i]$ so that these conditions hold for all $k$.
Then, for $\delta\in(0,\bar{\delta}_i]$, agent $i$'s payoff difference is $ V_i(\widehat{\mathcal N}_i;t)- V_i(\mathcal N_i;t)\geq-\delta o(\rho_i(t))$ for sufficient small $t$.

Let $E_i^0:=\{a_{i+1}=\cdots=a_{N-1}=0\}$ and $E_i^1=\{\bm{a}^{i-1}=\bm{a}^{i-1,*}\}$, where $\bm{a}^{i-1,*}$ is defined in Lemma~\ref{lem:local-neighborhood-comparison}.
We partition agent $N$'s ex ante payoff difference into four contributions.
The baseline contribution cancels, and agent $i$'s action can differ only after $p^+$.

First, consider $E_i^1\cap E_i^0$. 
On this event, after receiving $u_i$, agent $N$ follows $a_i$. 
By Lemma~\ref{lem:local-neighborhood-comparison}, the actions under $\widehat{\mathcal N}_i$ and $\mathcal N_i\neq\widehat{\mathcal N}_i$ can differ only when agent $i$ receives $p^+$; in particular, on $(\bm a^{i-1,*},p^+,E_i^0,u_i)$, agent $i$ chooses $1$ under $\widehat{\mathcal N}_i$ and $0$ under $\mathcal N_i$. 
Hence, agent $N$ also chooses $1$ and $0$, respectively, and this event contributes
\begin{align*}
    & \frac{1}{2} t^{Mb_i}\delta c\,\pi_N(u_i\mid H)(1+o(1))
    - \frac{1}{2} t^{(M+1)b_i}\frac{\delta c t^{\Omega_i+1/2-b_i}}{1+\varepsilon_i(t)}\pi_N(u_i\mid L)(1+o(1)) \\
    & = \frac{1}{2} t^{Mb_i}\delta c\cdot 2t^{\Omega_i+1/2}\pi_N(u_i\mid L)(1+o(1))
    - \frac{1}{2} t^{(M+1)b_i}\frac{\delta c t^{\Omega_i+1/2-b_i}}{1+\varepsilon_i(t)}\pi_N(u_i\mid L)(1+o(1)) \\
    & = \delta\frac{c\pi_N(u_i\mid L)}{2}\rho_i(t)(1+o(1)),
\end{align*}
where $\pi_N(u_i\mid L)$ is chosen to be a positive constant independent of $t$.

Second, on $E_i^1\cap E_i^0$, if agent $N$ receives $u_k$ with $k>i$, she chooses action $0$ under both neighborhoods. If $N$ receives $u_k$ with $k<i$, she either chooses action $1$ under both neighborhoods or follows $a_i$.
In the latter case, actions differ only after $p^+$, and the ratio of the joint probabilities of $H$ and $L$ with $(\bm a^{i-1,*},p^+,E_i^0,u_k)$ is $2t^{\Omega_k-\Omega_i}(1+o(1))>1$.
Thus, choosing action $1$ under the target rather than action $0$ under the alternative gives a positive contribution.
Conclusive signals contribute zero, so this case generates no loss.

Third, on $(E_i^0)^c$, agent $N$ chooses action $1$ under both neighborhoods after $u_k$ with $k\leq i$. 
For $k>i$, her action depends only on $a_{i+1},\dots,a_{N-1}$, which are unchanged when agent $i$'s neighborhood changes. 
Hence, the contribution is zero.

Finally, on $(E_i^1)^c\cap E_i^0$, only $u_k$ with $k\leq i$ can contribute.
By the proof of Lemma~\ref{lem:local-neighborhood-comparison}, a difference requires $p^+$ and $w>b_i$. 
Since $\pi_i(p^+\mid H)=\delta\cdot O(1)$, $\pi_N(u_k\mid H)=O(t^{\Omega_k+1/2})$, $\pi_i(p^+\mid L)=\delta\cdot O(t^{\Omega_i+1/2-b_i})$, and $\pi_N(u_k\mid L)=O(1)$, for each such history the absolute contribution is bounded by $\delta[O(t^{\Omega_k+1/2+Mw})+O(t^{\Omega_i+1/2-b_i+(M+1)w})]$.
Both terms are $\delta o(\rho_i(t))$, since $M(w-b_i)-(\Omega_i-\Omega_k)>0$ and $(M+1)(w-b_i)>0$ by $w-b_i\geq1$ and $M>\Omega_{N-1}$.
There are finitely many histories and signals, so the total loss is $\delta o(\rho_i(t))$.
Consequently, the positive contribution of order $\delta\rho_i(t)$ in the first case dominates all possible losses.

Fix $t>0$ sufficiently small. 
For each $i$, the preceding argument implies that there exists $d_i>0$ such that
\[
\lambda_i\bigl[V_i(\widehat{\mathcal N}_i;t)-V_i(\mathcal N_i;t)\bigr]
+\lambda_N\bigl[V_N(\widehat{\mathcal N}_i;t)-V_N(\mathcal N_i;t)\bigr]
>2\delta d_i
\]
for all $\delta\in(0,\bar{\delta}_i]$ and $\mathcal N_i\neq\widehat{\mathcal N}_i$.
Since all cutoff inequalities are strict, the induced actions are locally constant in the predecessor action distributions. 
Within this neighborhood, the weighted payoff difference is proportional to $\delta$, so the weighted payoff difference divided by $\delta$ extends continuously to $\delta=0$. 
Since there are finitely many $\mathcal N_i\neq\widehat{\mathcal N}_i$, this continuity is uniform over $\mathcal N_i\neq\widehat{\mathcal N}_i$ and $\delta\in[0,\bar{\delta}_i]$. 
Hence, there exists $\eta_i>0$ such that the weighted payoff difference remains at least $\delta d_i$ whenever the predecessor action distributions are within total variation distance $\eta_i$ of the baseline distributions.

Choose positive $\delta_2,\ldots,\delta_{N-1}$ backward such that $4(\sum_{j=i+1}^{N-1}\delta_j)\sum_{k=i+1}^N\lambda_k<\delta_i d_i/2$ for every $i$.
By homogeneity of this inequality, rescale them by a common constant such that $\delta_i\leq\bar\delta_i$ for every $i$ and $\sum_{j=2}^{N-1}\delta_j<\min_i\eta_i$.
Then, every baseline signal $s_j^{\mathrm b}$ induces $a_j=s_j^{\mathrm b}$ under every network.
Assign $\Pi_i(t,\delta_i)$ to every $i\in\{2,\ldots,N-1\}$ and denote the resulting profiles by $(\mu,\Pi)$.

Fix $\mathcal N\neq\widehat{\mathcal N}$ satisfying $\mathcal N_N=\widehat{\mathcal N}_N$, and let $i\in\{2,\ldots,N-1\}$ be its first deviation from $\widehat{\mathcal N}$. 
Let $\Pi^{[i]}$ be the information environment obtained by replacing $\Pi_j(t,\delta_j)$ with $\Pi_j^{\mathrm b}(t)$ for every $j\in\{i+1,\ldots,N-1\}$.
For each state, the distribution of $\bm a^{i-1}$ under $\Pi^{[i]}$ is within total variation distance $\sum_{j=2}^{i-1}\delta_j<\eta_i$ of its baseline distribution. 
Indeed, by coupling each $\Pi_j(t,\delta_j)$ with $\Pi_j^{\mathrm b}(t)$, the two action histories can differ only if some $\widetilde\Pi_j(t)$ with $2\leq j<i$ is selected, which occurs with probability at most $\sum_{j=2}^{i-1}\delta_j$.
Hence, by the preceding uniform comparison, the weighted payoff difference between $\widehat{\mathcal N}_i$ and $\mathcal N_i$ for agents $i$ and $N$ is at least $\delta_i d_i$.
For every $k<i$, the two networks have the same neighborhood, so agent $k$ has the same payoff under both networks. 
For every $i<k<N$, we have $a_k=s_k^{\mathrm b}$ under both networks, and hence their payoffs also coincide. Therefore, the welfare difference between $\widehat{\mathcal N}$ and $\mathcal N$ under $\Pi^{[i]}$ is at least $\delta_i d_i>0$.

Under either state, the continuation distributions under $\Pi^{[i]}$ and $\Pi$ have total variation distance at most $\sum_{j=i+1}^{N-1}\delta_j$, since they coincide unless some $\widetilde\Pi_j(t)$ with $i<j<N$ is selected. 
Since each agent's payoff takes values in $\{0,1\}$, the difference in expected payoff for any agent $k>i$ is at most twice this total variation distance. 
Applying this bound to both $\widehat{\mathcal N}$ and $\mathcal N$ and weighting by $\lambda_k$, replacing $\Pi^{[i]}$ with $\Pi$ reduces the welfare difference by at most
\[
4\left(\sum_{j=i+1}^{N-1}\delta_j\right) \sum_{k=i+1}^N\lambda_k <\frac{\delta_i d_i}{2}.
\]
Hence, the welfare difference remains positive.

Finally, successively for $i=1,\ldots,N-1$, perturb $\Pi_i$ by an arbitrarily small mixture supported on a disjoint signal set, without changing $\mu_i$.
By strictness and continuity, sufficiently small perturbations preserve all the strict inequalities established above.
Moreover, Lemma~\ref{lem:threshold-open-set} and the polynomial argument in the proof of Lemma~\ref{lem:generic-extension} allow us to choose these perturbations to make $\boldsymbol\alpha_{\leq N-1}(\cdot\mid \widehat{\mathcal N}_{\leq N-1},\Pi_{\leq N-1})$ observation separating.
Since the perturbations concern only $i<N$, $\Pi_N$ remains unchanged.
Denote the resulting profiles by $(\mu^0,\Pi^0)$.
\end{proof}

\begin{proof}[Proof of Proposition~\ref{prop:welfare-characterization}]
For sufficiency, take the network $\widehat{\mathcal N}$ fixed above and apply Lemma~\ref{lem:payoff-safe-full-observation}.
The minimum of $W_\lambda(\widehat{\mathcal N},\mu^0,\Pi^0)-W_\lambda(\mathcal N,\mu^0,\Pi^0)$ over networks satisfying $\mathcal N_N=\widehat{\mathcal N}_N$ and $\mathcal N\neq\widehat{\mathcal N}$ is strictly positive.

The properties of observation separation and conclusive signals in Lemma~\ref{lem:payoff-safe-full-observation} allow us to apply the same argument as in Lemma~\ref{lem:full-observation} to agent $N$.
Choose the added signal probabilities small enough to preserve the strict comparisons of weighted welfare established in Lemma~\ref{lem:payoff-safe-full-observation}.
Writing $(\mu,\Pi)$ for the resulting informational environment, we have $W_\lambda(\widehat{\mathcal N},\mu,\Pi)>W_\lambda(\mathcal N,\mu,\Pi)$ for every $\mathcal N\neq\widehat{\mathcal N}$ with $\mathcal N_N=\widehat{\mathcal N}_N$ and $W_\lambda(\widehat{\mathcal N},\mu,\Pi)>W_\lambda(\mathcal N,\mu,\Pi)$ for every $\mathcal N$ with $\mathcal N_{\leq N-1}=\widehat{\mathcal N}_{\leq N-1}$ and $\mathcal N_N\subsetneq\widehat{\mathcal N}_N$.
Now, take any $\mathcal N$ such that
$\mathcal N_{\leq N-1}\neq\widehat{\mathcal N}_{\leq N-1}$ and
$\mathcal N_N\subsetneq\widehat{\mathcal N}_N$, and define
$\mathcal N^+:=(\mathcal N_{\leq N-1},\widehat{\mathcal N}_N)$.
Adding these observations changes no predecessor's payoff and weakly increases agent $N$'s payoff. Therefore,
\[
 W_\lambda(\mathcal N,\mu,\Pi)
 \leq W_\lambda(\mathcal N^+,\mu,\Pi)
 <W_\lambda(\widehat{\mathcal N},\mu,\Pi).
\]
Hence, $\widehat{\mathcal N}$ is socially rationalizable.

Conversely, suppose that $j\notin\widehat{\mathcal N}_N$ for some $j<N$, and define
$\mathcal N^+:=(\widehat{\mathcal N}_{\leq N-1},\widehat{\mathcal N}_N\cup\{j\})$.
For every informational environment $(\mu,\Pi)$, $W_\lambda(\widehat{\mathcal N},\mu,\Pi)\leq W_\lambda(\mathcal N^+,\mu,\Pi)$, because the additional observation changes no payoff of agents $i<N$ and agent $N$ can ignore the additional action.
Hence, $\widehat{\mathcal N}$ cannot be a unique maximizer of weighted welfare.
\end{proof}

\section{Common Notation and Auxiliary Results}

Throughout this section and the subsequent proofs of the star- and complete-network results, all agents have prior $1/2$ and share a common information structure $\Pi=(S,\pi)$.

\subsection{Notation}
For two information structures $\Pi=(S,\pi)$ and $\Pi'=(S',\pi')$, define their product $\Pi\otimes\Pi'=(S\times S',\pi\otimes\pi')$ by $(\pi\otimes\pi')((s,s')\mid\theta)=\pi(s\mid\theta)\pi'(s'\mid\theta)$ for all $s\in S$, $s'\in S'$, and $\theta\in\Theta$.
We denote by $\Pi^{\otimes i}=(S^i,\pi^{\otimes i})$ the information structure generated by $i$ conditionally independent observations from $\Pi$.

Given a strategy profile $(\sigma_1,\dots,\sigma_i)$, $\sigma_i^{-1}(\bm{a}\mid \mathcal{N}, \Pi)$ denotes the set of signal profiles that are consistent with agent $i$'s history $\bm{a} = (a_j)_{j\in\mathcal{N}_i}$ under a network $\mathcal{N}$ and an information structure $\Pi$. Formally,
\begin{align*}
    \sigma_i^{-1}(\bm{a}\mid \mathcal{N}, \Pi) := \left\{ (s_1, \dots, s_{i-1}) \in S^{i-1} \;\middle|\; \sigma_j(s_1, \dots, s_j \mid \mathcal{N}, \Pi) = a_j \text{ for all } j \in \mathcal{N}_i \right\},
\end{align*}
In particular, $\sigma_j^{\ast}(s_1,\dots,s_j \mid \mathcal{N}, \Pi)$ denotes the equilibrium action of agent $j$ under signal profile $(s_1,\dots,s_j)$ and $\sigma_i^{\ast-1}(\bm{a}\mid \mathcal{N}, \Pi)$ denotes the set of signal profiles that is consistent with agent $i$'s history $\bm{a}$ in the equilibrium.

$\sigma_i^\ast(\bm{S}, s_i)$ denotes the optimal action under a set of signal profiles $\bm{S}\subset S^{i-1}$ and signal $s_i\in S$. Formally,
\begin{align*}
    \sigma_i^\ast(\bm{S}, s_i) :=
    \begin{cases}
        1 & \text{if } \mathbb{E}_\theta\left[u(1,\theta)\mid \bm{S}, s_i\right] > \mathbb{E}_\theta\left[u(0,\theta)\mid \bm{S}, s_i\right], \\
        0 & \text{otherwise,}
    \end{cases}
\end{align*}
where $\mathbb{E}_\theta\left[u(a,\theta)\mid \bm{S}, s_i\right] := \frac{\sum_{\theta\in\{L,H\}}\pi^{\otimes(i-1)}(\bm{S}\mid\theta)\,\pi(s_i\mid\theta)\,u(a,\theta)}{\sum_{\theta\in\{L,H\}}\pi^{\otimes(i-1)}(\bm{S}\mid\theta)\,\pi(s_i\mid\theta)}$.
As in the tie-breaking rule introduced for equilibrium, agent $i$ chooses action $0$ when indifferent between the two actions.
Note that $\sigma_i^\ast(\bm{S},s_i)=1$ is equivalent to $LR(\bm{S})LR(s_i)>1$, where $LR(\bm{S}) := \pi^{\otimes(i-1)}(\bm{S}\mid H) / \pi^{\otimes(i-1)}(\bm{S}\mid L)$.
This is because, by the definition of the utility function and the prior distribution, the condition $\mathbb{E}_\theta[u(1,\theta)\mid \bm{S},s_i]>\mathbb{E}_\theta[u(0,\theta)\mid \bm{S},s_i]$ reduces to $\pi^{\otimes(i-1)}(\bm{S}\mid H)\pi(s_i\mid H)>\pi^{\otimes(i-1)}(\bm{S}\mid L)\pi(s_i\mid L)$, which is equivalent to $LR(\bm{S})LR(s_i)>1$.

The definition of optimal action is consistent with the equilibrium strategy. Indeed, it holds that $\sigma_i^\ast(\bm{a}, s_i\mid \mathcal{N}, \Pi) = \sigma_i^\ast(\sigma_i^{\ast-1}(\bm{a}\mid \mathcal{N}, \Pi), s_i)$ for any $\bm{a}$, $s_i$, and $i$, since conditioning on agent $i$'s history $\bm{a}$ is equivalent to conditioning on the set of signal profiles $\sigma_i^{\ast-1}(\bm{a}\mid\mathcal{N},\Pi)$.

\subsection{Lemma}
Here, we define $V_i(\mathcal{N}^S,\Pi^B\mid s_1,\dots,s_{i-1})$ as the expected payoff for agent $i$ at equilibrium when agent $j$ receives private signal $s_j$ for $j=1,\dots,i-1$.

\begin{lemma}\label{lem: ignoring principle}
    Consider networks $\mathcal{N}$ and $\mathcal{N}^{\prime}$ with $\mathcal{N}_{i}=\mathcal{N}_{i}^{\prime}$ for all $i<N$ and $\mathcal{N}_{N}\subset\mathcal{N}_{N}^{\prime}$.
    For any information structure $\Pi$ and any signal profile $(s_1,\dots,s_{N-1})$, $V_{N}\left(\mathcal{N},\Pi \mid s_1,\dots,s_{N-1} \right)\leq V_{N}\left(\mathcal{N}^{\prime},\Pi \mid s_1,\dots,s_{N-1} \right)$.
    Therefore, $V_{N}\left(\mathcal{N},\Pi\right)\leq V_{N}\left(\mathcal{N}^{\prime},\Pi\right)$.
\end{lemma}
\begin{proof}[Proof of Lemma \ref{lem: ignoring principle}]
    Let $\bm{\sigma}^{\ast}$ and $\bm{\sigma}^{\prime\ast}$ be the equilibria under networks $\mathcal{N}$ and $\mathcal{N}^{\prime}$, respectively.
    By using $\sigma_{N}^{\ast}$, we define agent $N$'s strategy, $\tilde{\sigma}_{N}$, under network $\mathcal{N}^{\prime}$ by  $\tilde{\sigma}_{N}(\left(a_{i}\right)_{i\in\mathcal{N}_{N}^{\prime}},s_{N})=\sigma_{N}^{\ast}\left(\left( \tilde{a} _{i}\right)_{i\in\mathcal{N}_{N}},s_{N}\right)$ if $a_i = \tilde{a}_i$ for all $i\in \mathcal{N}_N$.
    We can define such a strategy since $\mathcal{N}_{N}\subset\mathcal{N}_{N}^{\prime}$.
    Then, $\tilde{\sigma}_{N}\left(s_{1},\dots,s_{N}\right)=\sigma_{N}^{\ast}\left(s_{1},\dots,s_{N}\right)$ for all $s_{1},\dots,s_{N}$, since $\sigma_{i}^{\prime\ast}\left(s_{1},\dots,s_{i}\right)=\sigma_{i}^{\ast}\left(s_{1},\dots,s_{i}\right)$ for all $i<N$ and $s_{1},\dots,s_{i}$. The following calculation completes the proof: for all $s_{1},\dots,s_{N-1}$,
    \begin{align*}
         & V_{N}\left(\mathcal{N},\Pi \mid s_1,\dots,s_{N-1} \right)\\
         & =  \sum_{s_N, \theta}\frac{1}{2}u\left(\sigma_{N}^{\ast}\left(s_{1},\dots,s_{N}\right),\theta\right)\pi^{\otimes N}\left(s_{1},\dots,s_{N}\mid\theta\right) / \pi^{\otimes N-1}\left(s_{1},\dots,s_{N-1} \right) \\
         & = \sum_{s_N, \theta}\frac{1}{2}u\left(\tilde{\sigma}_{N}\left(s_{1},\dots,s_{N}\right),\theta\right)\pi^{\otimes N}\left(s_{1},\dots,s_{N}\mid\theta\right) / \pi^{\otimes N-1}\left(s_{1},\dots,s_{N-1} \right) \\
         & \leq   \sum_{s_N, \theta}\frac{1}{2}u\left(\sigma_{N}^{\prime\ast}\left(s_{1},\dots,s_{N}\right),\theta\right)\pi^{\otimes N}\left(s_{1},\dots,s_{N}\mid\theta\right) / \pi^{\otimes N-1}\left(s_{1},\dots,s_{N-1} \right) \\
         & =  V_{N}\left(\mathcal{N}^\prime,\Pi \mid s_1,\dots,s_{N-1} \right),
    \end{align*}
    where $\pi^{\otimes N-1}\left(s_{1},\dots,s_{N-1} \right):=\sum_{\theta} \frac{1}{2}\pi^{\otimes N-1}\left(s_{1},\dots,s_{N-1} \mid \theta \right)$ , and the inequality follows since $\sigma_{N}^{\prime\ast}$ is agent $N$'s equilibrium strategy in network $\mathcal{N}^{\prime}$.
\end{proof}

\section{Proof of Theorem \ref{thm: star}}
First, we show that $V_N(\mathcal{N}^S,\Pi^B) \geq V_N(\mathcal{N},\Pi^B)$ holds for all $\mathcal{N}$ and $\Pi^B$ by using the following lemma. 
Here, we define $\overline{V_i}(\Pi^{\otimes i})$ as 
\[
    \overline{V_i}(\Pi^{\otimes i})=
    \max_{\sigma_{i}:S^i\to A} \mathbb{E}_\theta\left[\sum_{a\in A}\sum_{\bm{s}\in S^i}\sigma_{i}(a\mid \bm{s}) \pi^{\otimes i}(\bm{s}\mid \theta)u(a,\theta)\right].
\]
In words, this is the maximum expected payoff if an agent observes $i$ signals from $\Pi$.
Also, we define $\overline{V_i}(\Pi^{\otimes i} \mid s_1,\dots, s_{i-1})$ as the maximum expected payoff under the same situation conditional on the signal profile $(s_1,\dots, s_{i-1})$.

\begin{lemma} \label{lem: upper bound}
Take any $\mathcal{N}$ and $\Pi$. 
Then, we have
    \begin{align*}
        V_i(\mathcal{N},\Pi \mid s_1,\dots,s_{i-1}) \leq \overline{V_i}(\Pi^{\otimes i} \mid s_1,\dots,s_{i-1})
    \end{align*}
    for all $i$ and $(s_1,\dots,s_{i-1})\in \supp{\pi^{\otimes (i-1)}}$.\footnote{$\supp{\pi^{\otimes (i-1)}}$ is defined as the set of signal profiles $(s_1\dots,s_{i-1})$ satisfying $\pi^{\otimes (i-1)} (s_1\dots,s_{i-1} \mid L) >0$ or $\pi^{\otimes (i-1)} (s_1\dots,s_{i-1} \mid H) >0$.}
    In particular, $V_i(\mathcal{N},\Pi) \leq \overline{V_i}(\Pi^{\otimes i})$ for all $i$.
    \end{lemma}
    \begin{proof}[Proof of Lemma \ref{lem: upper bound}]
    Let $\bm{\sigma}^*$ be the equilibrium strategy under $(\mathcal{N},\Pi)$. 
    Fix $(s_1,\dots,s_{i-1})\in S^{i-1}$ arbitrarily.
    It follows that
    \begin{align*}
        V_i(\mathcal{N},\Pi \mid s_1,\dots,s_{i-1})=\frac{\sum_{s_i\in S_0}\pi^{\otimes i}(s_1,\dots,s_i\mid L) + \sum_{s_i\in S_1}\pi^{\otimes i}(s_1,\dots,s_i\mid H)}{\pi^{\otimes (i-1)}(s_1,\dots,s_{i-1}\mid H)+\pi^{\otimes (i-1)}(s_1,\dots,s_{i-1}\mid L)},
        \end{align*}
        where $S_0=\{s_i\in S \mid \sigma^\ast_i(s_1,\dots,s_i)=0\}$ and $S_1=\{s_i\in S \mid  \sigma^\ast_i(s_1,\dots,s_i)=1\}$.
Note that
        \begin{multline*}
            \sum_{s_i\in S_1}\pi^{\otimes i}(s_1,\dots,s_i\mid H)+\sum_{s_i\in S_0}\pi^{\otimes i}(s_1,\dots,s_i\mid L) \\
            \leq \sum_{s_i\in  S}\max\{\pi^{\otimes i}(s_1,\dots,s_i\mid H),\pi^{\otimes i}(s_1,\dots,s_i\mid L)\}.
        \end{multline*}
        Hence, we have
         \begin{align*}
        V_i(\mathcal{N},\Pi\mid s_1,\dots,s_{i-1})&\leq\frac{ \sum_{s_i\in  S}\max\{\pi^{\otimes i}(s_1,\dots,s_i\mid H),\pi^{\otimes i}(s_1,\dots,s_i\mid L)\}}{\pi^{\otimes (i-1)}(s_1,\dots,s_{i-1}\mid H)+\pi^{\otimes (i-1)}(s_1,\dots,s_{i-1}\mid L)}\\
        &=\overline{V_i}(\Pi^{\otimes i}\mid s_1,\dots,s_{i-1}).
        \end{align*}
        Taking expectations over $(s_1,\dots,s_{i-1})$ yields $V_i(\mathcal{N},\Pi)\le \overline{V_i}(\Pi^{\otimes i})$.
    \end{proof}
    
    \begin{proof}[Proof of Theorem \ref{thm: star}]
    Since each agent $i<N$ chooses action $1$ if and only if $s_{i}=h$ under $\mathcal{N}^S$, we have
    \begin{align*}
        V_{N}(\mathcal{N}^{S},\Pi^{B}\mid s_1,\dots,s_{N-1})=\overline{V}_{N}(\Pi^{B \otimes N}\mid s_1,\dots,s_{N-1})
    \end{align*}
    and $V_{N}(\mathcal{N}^{S},\Pi^{B})=\overline{V}_{N}(\Pi^{B \otimes N})$ for all $\Pi^B$ and $(s_1,\dots,s_{N-1})$.
    Combining this with Lemma \ref{lem: upper bound}, we obtain $V_N(\mathcal{N}^S,\Pi^B) \geq V_N(\mathcal{N},\Pi^B)$ for all $\mathcal{N}$ and $\Pi^B$.

    Next, we show the strict inequality.
    We have $V_N(\mathcal{N}^S,\Pi^B\mid s_1,\dots,s_{N-1}) \geq V_N(\mathcal{N},\Pi^B\mid s_1,\dots,s_{N-1})$ for all $\Pi^B$, $\mathcal{N}$, and $(s_1,\dots,s_{N-1})$.
    Thus, it suffices to show that there exists $\Pi^B$ such that for all $\mathcal{N}$, there exists $(s_1,\dots,s_{N-1})\in \supp{\pi^{\otimes (N-1)}}$ such that $ \overline{V}_{N}(\Pi^{B \otimes N}\mid s_1,\dots,s_{N-1}) >V_N(\mathcal{N},\Pi^B\mid s_1,\dots,s_{N-1})$. 
    Take information structure $\Pi^B$ that satisfies $LR(l)LR(h)^{N-2}<1<LR(l)LR(h)^{N-1}$.\footnote{This is feasible, for example, $\pi(l\mid \theta=H)=\frac{2}{2\cdot 3^{N-1}-1}, \pi(l\mid \theta=L)=\frac{4\cdot 3^{N-2}}{2\cdot 3^{N-1}-1}, \pi(h\mid \theta=H)=1-\frac{2}{2\cdot 3^{N-1}-1}$, and $ \pi(h\mid \theta=L)=1-\frac{4\cdot 3^{N-2}}{2\cdot 3^{N-1}-1}$ induce $LR(l)LR(h)^{N-2}=\frac{1}{2}$ and $LR(l)LR(h)^{N-1}=\frac{3}{2}$.}
    Take any $\mathcal{N}\neq \mathcal{N}^S$. We divide the proof into two cases. Case (i): $\mathcal{N}_i=\emptyset$ for all $i=1,2,\dots,N-1$. Case (ii): $\mathcal{N}_i\neq \emptyset$ for some $i\in \{1,2,\dots,N-1\}$.

    We begin with Case (i): $\mathcal{N}_i=\emptyset$ for all $i=1,2,\dots,N-1$. 
    In this case, if $|\mathcal{N}_N|=N-1$, $\mathcal{N}$ coincides with $\mathcal{N}^S$. Thus, $|\mathcal{N}_N|<N-1$.
    Consider $(s_1,\dots,s_{N-1})=(h,h,\dots,h).$ Since agent $N$ observes action $1$ only and agent $i<N$ takes action $1$ if and only if he receives $h$, agent $N$ takes action $1$ if she receives $h$. 
    If agent $N$ receives $l$, her posterior is $LR(l)LR(h)^{|\mathcal{N}_N|}\leq LR(l)LR(h)^{N-2}<1$. 
    Thus, she takes action $0$. 
    Hence,
    \begin{align*}
    \frac{V_N(\mathcal{N},\Pi^B\mid s_1,\dots,s_{N-1})}{\overline{V}_N(\Pi^{B\otimes N}\mid s_1,\dots,s_{N-1})}&=\frac{\sum_{s_N\in S_0}\pi^{\otimes N}(s_1,\dots,s_N\mid L) + \sum_{s_N\in S_1}\pi^{\otimes N}(s_1,\dots,s_N\mid H)}{\sum_{s_N\in  S}\max\{\pi^{\otimes N}(s_1,\dots,s_N\mid H),\pi^{\otimes N}(s_1,\dots,s_N\mid L)\}}\\
    &=\frac{\pi^{\otimes N}(s_1,\dots,s_{N-1},l\mid L) + \pi^{\otimes N}(s_1,\dots,s_{N-1},h\mid H)}{\pi^{\otimes N}(s_1,\dots,s_{N-1},h\mid H)+\pi^{\otimes N}(s_1,\dots,s_{N-1},l\mid H)}\\
    &<1,
    \end{align*}
    where $S_0=\{s_N\in S\mid \sigma^\ast_N(s_1,\dots,s_N)=0\}$ and $S_1=\{s_N\in S\mid  \sigma^\ast_N(s_1,\dots,s_N)=1\}$, and $\bm{\sigma}^\ast$ is the equilibrium strategy under $(\mathcal{N}, \Pi^B)$.
    The second equality and the last inequality come from $\pi^{\otimes N}(s_1,\dots,s_{i-1},h\mid H)>\pi^{\otimes N}(s_1,\dots,s_{i-1},h\mid L)$ and $\pi^{\otimes N}(s_1,\dots,s_{i-1},l\mid H)>\pi^{\otimes N}(s_1,\dots,s_{i-1},l\mid L)$, which is derived from $LR(l)LR(h)^{N-1}>1.$

    It remains to consider Case (ii): $\mathcal{N}_i\neq \emptyset$ for some $i\in \{1,2,\dots,N-1\}$.
    Consider another network $\mathcal{N}'$ which satisfies $\mathcal{N}'_i=\mathcal{N}_i$ for all $i\in \{1,2,\dots,N-1\}$ and $\mathcal{N}'_N=\{1,2,\dots,N-1\}$.
    By Lemma \ref{lem: ignoring principle}, $V_N(\mathcal{N},\Pi^B\mid s_1,\dots,s_{N-1}) \leq V_N(\mathcal{N}',\Pi^B\mid s_1,\dots,s_{N-1}) $ for any $(s_1,\dots,s_{N-1})$. Hence, it suffices to show that there exists $(s_1,\dots,s_{N-1})$ such that $ \overline{V}_{N}(\Pi^{B \otimes N}\mid s_1,\dots,s_{N-1}) >V_N(\mathcal{N}^\prime,\Pi^B\mid s_1,\dots,s_{N-1})$.
    
    By the construction of $\mathcal{N}^\prime$, we can take $i,j\neq N$ such that $j\in \mathcal{N}'_i$. 
    Consider the signal profile $(s_1,\dots,s_{N-1})$, where $s_j=l$ and $s_k=h$ for all $k\neq j$. 
    Then, agent $j$ chooses action $0$. 
    Agent $i$ observes this, and he cascades to choose action $0$, because he knows that some agent $j^\prime\leq j$ receives $l$. 
    This is because, otherwise, all the agents $j^\prime \leq j$ would choose action 1. 
    Agent $N$ chooses action $0$ if she receives $l$ because she knows that there is at least one agent other than herself who received $l$, by the same argument applied to agent $i$.
    In addition, agent $N$ chooses action $0$ even if she receives $h$ because the number of agents whom she knows to have received $h$ (including herself) is at most $N-2$, while the number of agents whom she knows to have received $l$ is at least $1$, and $LR(l)LR(h)^{N-2}<1$.
Hence,
    \begin{align*}
    \frac{V_N(\mathcal{N}',\Pi^B\mid s_1,\dots,s_{N-1})}{\overline{V}_N(\Pi^{B\otimes N}\mid s_1,\dots,s_{N-1})}
    &=\frac{\pi^{\otimes N}(s_1,\dots,s_{N-1},h\mid L)+\pi^{\otimes N}(s_1,\dots,s_{N-1},l\mid L)}{\sum_{s_N\in  S}\max\{\pi^{\otimes N}(s_1,\dots,s_N\mid H),\pi^{\otimes N}(s_1,\dots,s_N\mid L)\}}\\
    &=\frac{\pi^{\otimes N}(s_1,\dots,s_{N-1},h\mid L)+\pi^{\otimes N}(s_1,\dots,s_{N-1},l\mid L)}{\pi^{\otimes N}(s_1,\dots,s_{N-1},h\mid H)+\pi^{\otimes N}(s_1,\dots,s_{N-1},l\mid L)}\\
    &<1.
    \end{align*}
        The second equality and the last inequality come from $\pi^{\otimes N}(s_1,\dots,s_{i-1},h\mid H)>\pi^{\otimes N}(s_1,\dots,s_{i-1},h\mid L)$ and $\pi^{\otimes N}(s_1,\dots,s_{i-1},l\mid H)<\pi^{\otimes N}(s_1,\dots,s_{i-1},l\mid L)$, which follows from $LR(l)LR(h)^{N-2}<1<LR(l)LR(h)^{N-1}.$
\end{proof}

\section{Proofs of Theorem \ref{thm: complete} and Proposition \ref{prop: pareto dominate}}

\subsection{Construction and Properties of \texorpdfstring{$\Pi^\ast$}{}}
We begin by constructing an information structure $\Pi^\ast$, and  then derive properties of $\Pi^\ast$. We define an information structure $\Pi^\ast=(S^\ast, \pi^\ast)$ with $S^\ast = \{l^\ast, l, h, h^\ast \}$ as follows. Take any $\delta\in\left(0,\frac{1}{3N}\right)$.
\begin{alignat*}{4}
&\pi^{\ast}(h\mid H)=1-B(\delta),
&&\quad \pi^{\ast}(l\mid H)=\delta^{N+1},
&&\quad \pi^{\ast}(h^{\ast}\mid H)=\delta^{N+1},
&&\quad \pi^{\ast}(l^{\ast}\mid H)=\delta^{(\alpha+1)N^{3}},\\
&\pi^{\ast}(h\mid L)=1-A(\delta),
&&\quad \pi^{\ast}(l\mid L)=\delta,
&&\quad \pi^{\ast}(h^{\ast}\mid L)=\delta^{\alpha},
&&\quad \pi^{\ast}(l^{\ast}\mid L)=\delta^{\alpha+1}.
\end{alignat*}
where $\alpha:=N^{2}+N+1$, $A\left(\delta\right):=\delta+\delta^{\alpha}+\delta^{\alpha+1}$
and $B\left(\delta\right):=2\delta^{N+1}+\delta^{\left(\alpha+1\right)N^{3}}$.
Then, the likelihood ratio of each signal is $LR\left(h\right)=\frac{1-B\left(\delta\right)}{1-A\left(\delta\right)}$,
$LR\left(l\right)=\delta^{N}$, $LR\left(h^{\ast}\right)=\delta^{-N^{2}}$
and $LR\left(l^{\ast}\right)=\delta^{N^{3}}$, respectively.

\begin{lemma}\label{lem: properties of Pi^ast}
    Under information structure $\Pi^\ast$, the following eight inequalities hold for any $i = 1,\dots, N$,
        \begin{gather}
        LR(h)^{i-1} LR(l) < 1, \label{eq: LR_hl} \\
        LR(l)^{i-1} LR(h^{\ast}) > 1, \label{eq: LR_lh*} \\
        LR(h^{\ast})^{i-1} LR(l^{\ast}) < 1, \label{eq: LR_h*l*} \\
        \frac{1}{LR(h)} < \pi^{\ast}(h\mid H)^{i-1}, \label{eq: prob_h_theta1} \\
        LR(l) < \pi^{\ast}(h\mid L)^{i-1}, \label{eq: prob_h_theta0} \\
        \bigl(1-\pi^{\ast}(h\mid H)^{i-1}\bigr) LR(h)
        < \pi^{\ast}(l\mid L)^{i-1}, \label{eq: prob_l_theta0} \\
        \bigl(1-\pi^{\ast}(h\mid L)^{i-1}\bigr)\frac{1}{LR(h^{\ast})}
        < \pi^{\ast}(l\mid1)^{i-1}, \label{eq: prob_l_theta1} \\
        \bigl(1-\pi^{\ast}(\{l,h\}\mid L)^{i-1}\bigr)\frac{1}{LR(l)}
        < \pi^{\ast}(h^{\ast}\mid H)^{i-1}. \label{eq: prob_h*_theta1}
    \end{gather}
\end{lemma}

\begin{proof}[Proof of Lemma \ref{lem: properties of Pi^ast}]
    To establish these inequalities, it suffices to show each inequality for $i=N$. To show the inequalities above, the followings are useful:
    \begin{gather}
        \delta < A(\delta) < 2\delta, \label{eq: A_delta} \\
        2\delta^{N+1} < B(\delta) < 3\delta^{N+1}, \label{eq: B_delta} \\
        (1-x)^{m} \ge 1 - mx
        \ \text{for any } x \in (0,1) \text{ and } m \in \mathbb{N}.
        \label{eq: useful_ineq}
    \end{gather}
    Note that we have the inequalities in (\ref{eq: A_delta}) and (\ref{eq: B_delta}) because $\delta<\frac{1}{3N}<\frac{1}{2}$.

    First, we show the inequality in (\ref{eq: LR_hl}), i.e., $\left(\frac{1-B\left(\delta\right)}{1-A\left(\delta\right)}\right)^{N-1}\delta^{N}<1$.
    Since $\delta^{-\left(N-1\right)}\delta^{N}<1$, it suffices to show $\frac{1-B\left(\delta\right)}{1-A\left(\delta\right)}<\delta^{-1}$.
    We obtain this result by the following calculation:
    \[
    \frac{1-B\left(\delta\right)}{1-A\left(\delta\right)}<\frac{1}{1-A\left(\delta\right)}<\frac{1}{1-2\delta}<\frac{1}{\delta},
    \]
    where the second inequality follows from the inequality in (\ref{eq: A_delta}) and the second inequality follows from $\delta<\frac{1}{3N}<\frac{1}{3}$.

    Second, the inequalities in (\ref{eq: LR_lh*}) and (\ref{eq: LR_h*l*}) directly follow from the construction of the information structure.

    Third, we show the inequality in (\ref{eq: prob_h_theta1}), i.e.,
    $\frac{1-A\left(\delta\right)}{1-B\left(\delta\right)}<\left(1-B\left(\delta\right)\right)^{N-1}$, which is equivalent to $1-A\left(\delta\right)<\left(1-B\left(\delta\right)\right)^{N}$.
    We obtain this result by the following calculation:
    \[
    1-A\left(\delta\right)<1-\delta<1-3N\delta^{N+1}\leq\left(1-3\delta^{N+1}\right)^{N}<\left(1-B\left(\delta\right)\right)^{N},
    \]
    where the first inequality follows from the inequality in (\ref{eq: A_delta}), the second inequality follows from $\delta<\frac{1}{3N}$ and the third inequality follows from the inequality in (\ref{eq: useful_ineq}) and the last inequality follows from the inequality in (\ref{eq: B_delta}).

    Fourth, we show the inequality in (\ref{eq: prob_h_theta0}), i.e.,
    $\delta^{N}<\left(1-A\left(\delta\right)\right)^{N-1}$. We obtain
    this result by the following calculation:
    \[
    \left(1-A\left(\delta\right)\right)^{N-1}>\left(1-2\delta\right)^{N-1}\geq1-2\left(N-1\right)\delta>\delta^{N},
    \]
    where the first inequality follows from the inequality in (\ref{eq: A_delta}), the second inequality follows from the inequality in (\ref{eq: useful_ineq}) and the last inequality follows from $\delta<\frac{1}{3N}$. 

        Fifth, we show the inequality in (\ref{eq: prob_l_theta0}), i.e.,
    $\left(1-\left(1-B\left(\delta\right)\right)^{N-1}\right)\frac{1-B\left(\delta\right)}{1-A\left(\delta\right)}<\delta^{N-1}.$
    Note that $\frac{1-B\left(\delta\right)}{1-A\left(\delta\right)}<2$ since 
    \[
    \frac{1-B\left(\delta\right)}{1-A\left(\delta\right)}<\frac{1}{1-A\left(\delta\right)}<\frac{1}{1-2\delta}<2,
    \]
    where the second inequality follows from the inequalities in (\ref{eq: A_delta}) and the last inequality follows from $\delta<\frac{1}{3N}<\frac{1}{4}$.
    Thus, it suffices to show $1-\left(1-B\left(\delta\right)\right)^{N-1}<\frac{1}{2}\delta^{N-1}$.
    We obtain this result by the following calculation:
    \[
    1-\left(1-B\left(\delta\right)\right)^{N-1}<1-\left(1-3\delta^{N+1}\right)^{N-1}\leq3\left(N-1\right)\delta^{N+1}<\delta^{N}<\frac{1}{2}\delta^{N-1},
    \]
    where the first inequality follows from the inequality in (\ref{eq: B_delta}), the second inequality follows from the inequality in (\ref{eq: useful_ineq}) and the third and last inequalities follows from $\delta<\frac{1}{3N}$.

    Sixth, we show the inequality in (\ref{eq: prob_l_theta1}), i.e.,
    $\left(1-\left(1-A\left(\delta\right)\right)^{N-1}\right)\delta^{N^{2}}<\delta^{\left(N+1\right)\left(N-1\right)}$.
    It suffices to show $1-\left(1-A\left(\delta\right)\right)^{N-1}<\delta^{-1}$.
    This inequality follows because $1-\left(1-A\left(\delta\right)\right)^{N-1}<1<\delta^{-1}$.

    Seventh, we show the inequality in (\ref{eq: prob_h*_theta1}), i.e.,
    $\left(1-\left(1-\delta^{\alpha}-\delta^{\alpha+1}\right)^{N-1}\right)\delta^{-N}<\delta^{\left(N+1\right)\left(N-1\right)}$.
    It suffices to show $1-\left(1-\delta^{\alpha}-\delta^{\alpha+1}\right)^{N-1}<\delta^{N^{2}+N-1}$.
    We obtain this result by the following calculation:
    \[
    1-\left(1-\delta^{\alpha}-\delta^{\alpha+1}\right)^{N-1}
    <
    1-\left(1-2\delta^{\alpha}\right)^{N-1}
    \leq
    2\left(N-1\right)\delta^{\alpha}
    <
    \delta^{N^{2}+N-1},
    \]
    where the second inequality follows from the inequality in (\ref{eq: useful_ineq}) and the last inequality follows from $\delta<\frac{1}{3N}$.
\end{proof}

For $S\subset S^\ast$, $\sigma_S:S^\ast \to A$ denotes a function defined by $\sigma_S(s)=1$ if $s\in S$ and $\sigma_S(s)=0$ if $s\not\in S$.  In particular, $\sigma_\emptyset(s)=0$ for all $s\in S^\ast$.

\begin{lemma}\label{lem: equilibrium characterization}
    The optimal actions under the information structure $\Pi^\ast$ are always unique and characterized as follows. Take $i\geq 2$ and $\bm{S}\subset S^{\ast i-1}$ arbitrarily. 
    \begin{enumerate}[(i)]
        \item
        Suppose that $\bm{S}\cap \{h\}^{i-1}\neq\emptyset$. Then, $\sigma_{i}^\ast \left(\bm{S}, \cdot \right) = \sigma_{\{ h,h^\ast \}}(\cdot)$
            
        \item
        Suppose that $\bm{S}\cap \{l,h\}^{i-1}\neq\emptyset$ and $\bm{S}\cap \{h\}^{i-1} = \emptyset$. Then, $\sigma_{i}^\ast \left(\bm{S}, \cdot \right) = \sigma_{\{h^\ast \}}(\cdot)$
            
        \item
        Suppose that $\bm{S}\cap \{l,h,h^\ast \}^{i-1}\neq\emptyset$ and $\bm{S}\cap \{l,h\}^{i-1} = \emptyset$. Then, $\sigma_{i}^\ast \left(\bm{S}, \cdot \right) = \sigma_{\{l,h,h^\ast \}}(\cdot)$
            
        \item
        Suppose that $\bm{S}\cap \{l,h, h^\ast\}^{i-1} = \emptyset$. Then, $\sigma_{i}^\ast \left(\bm{S}, \cdot \right) = \sigma_\emptyset(\cdot)$
    \end{enumerate}
\end{lemma}

\begin{proof}[Proof of Lemma \ref{lem: equilibrium characterization}]
    \proofpart{Part (i)}
    Fix $\bm{S}\subset S^{\ast i-1}$ with $\bm{S}\cap\{h\}^{i-1}\neq\emptyset$ arbitrarily. To show $\sigma_{\{ h,h^\ast \}}(s_i)$ is uniquely optimal under $\bm{S}$ for any $s_i\in S^\ast$, it suffices to show that $LR(\bm{S})LR(l)<1<LR(\bm{S})LR(h)$. Let $S_h:=\bm{S}\cap\{h\}^{i-1}$ and $\bar{S}_h:=\bm{S}\setminus\{h\}^{i-1}$.
    
    First, we show $LR(\bm{S})LR(l)<1$, which is equivalent to
    \begin{align*}
        &\ \pi^{\ast\otimes(i-1)}(\bm{S}\mid H)\,LR(l) < \pi^{\ast\otimes(i-1)}(\bm{S}\mid L) \\
        \Leftrightarrow&\ \pi^{\ast\otimes(i-1)}(\bar{S}_h\mid H)\,LR(l) - \pi^{\ast\otimes(i-1)}(\bar{S}_h\mid L) < \pi^{\ast\otimes(i-1)}(S_h\mid L) - \pi^{\ast\otimes(i-1)}(S_h\mid H)\,LR(l).
    \end{align*}
    The following calculation establishes this inequality:
    \begin{align*}
        \pi^{\ast\otimes(i-1)}(\bar{S}_h\mid H)\,LR(l) - \pi^{\ast\otimes(i-1)}(\bar{S}_h\mid L)
        &\leq \pi^{\ast\otimes(i-1)}(\bar{S}_h\mid H)\,LR(l) \\
        &\leq (1-\pi^{\ast\otimes(i-1)}(S_h\mid H))\,LR(l) \\
        &< \pi^{\ast\otimes(i-1)}(S_h\mid L) - \pi^{\ast\otimes(i-1)}(S_h\mid H)\,LR(l),
    \end{align*}
    where the third inequality follows from (\ref{eq: prob_h_theta0}).
    
    Next, we show $LR(\bm{S})LR(h)>1$, which is equivalent to
    \begin{align*}
        &\ \pi^{\ast\otimes(i-1)}(\bm{S}\mid L) < \pi^{\ast\otimes(i-1)}(\bm{S}\mid H)\,LR(h) \\
        \Leftrightarrow&\ \pi^{\ast\otimes(i-1)}(\bar{S}_h\mid L) - \pi^{\ast\otimes(i-1)}(\bar{S}_h\mid H)\,LR(h) < \pi^{\ast\otimes(i-1)}(S_h\mid H)\,LR(h) - \pi^{\ast\otimes(i-1)}(S_h\mid L).
    \end{align*}
    The following calculation establishes this inequality:
    \begin{align*}
        \pi^{\ast\otimes(i-1)}(\bar{S}_h\mid L) - \pi^{\ast\otimes(i-1)}(\bar{S}_h\mid H)\,LR(h)
        &\leq \pi^{\ast\otimes(i-1)}(\bar{S}_h\mid L) \\
        &\leq 1-\pi^{\ast\otimes(i-1)}(S_h\mid L) \\
        &< \pi^{\ast\otimes(i-1)}(S_h\mid H)\,LR(h) - \pi^{\ast\otimes(i-1)}(S_h\mid L),
    \end{align*}
    where the third inequality follows from (\ref{eq: prob_h_theta1}).

    \proofpart{Part (ii)}
    Fix $\bm{S}\subset S^{\ast i-1}$ with $\bm{S}\cap\{l,h\}^{i-1}\neq\emptyset$ and $\bm{S}\cap\{h\}^{i-1}=\emptyset$ arbitrarily. To show $\sigma_{\{ h^\ast \}}(s_i)$ is uniquely optimal under $\bm{S}$ for any $s_i\in S^\ast$, it suffices to show that $LR(\bm{S})LR(h)<1<LR(\bm{S})LR(h^\ast)$. Let $S_{lh}:=\bm{S}\cap\{l,h\}^{i-1}$ and $\bar{S}_{lh}:=\bm{S}\setminus\{l,h\}^{i-1}$.
    
    First, we show $LR(\bm{S})LR(h)<1$, which is equivalent to
    \begin{align*}
        &\ \pi^{\ast\otimes(i-1)}(\bm{S}\mid H)\,LR(h) < \pi^{\ast\otimes(i-1)}(\bm{S}\mid L) \\
        \Leftrightarrow&\ \pi^{\ast\otimes(i-1)}(\bar{S}_{lh}\mid H)\,LR(h) - \pi^{\ast\otimes(i-1)}(\bar{S}_{lh}\mid L) < \pi^{\ast\otimes(i-1)}(S_{lh}\mid L) - \pi^{\ast\otimes(i-1)}(S_{lh}\mid H)\,LR(h).
    \end{align*}
    The following calculation establishes this inequality:
    \begin{align*}
        \pi^{\ast\otimes(i-1)}(\bar{S}_{lh}\mid H)\,LR(h) - \pi^{\ast\otimes(i-1)}(\bar{S}_{lh}\mid L)
        &\leq \pi^{\ast\otimes(i-1)}(\bar{S}_{lh}\mid H)\,LR(h) \\
        &\leq (1-\pi^{\ast\otimes(i-1)}(\{h\}^{i-1}\mid H)-\pi^{\ast\otimes(i-1)}(S_{lh}\mid H))\,LR(h) \\
        &< \pi^{\ast\otimes(i-1)}(\{l\}^{i-1}\mid L) - \pi^{\ast\otimes(i-1)}(S_{lh}\mid H)\,LR(h) \\
        &\leq \pi^{\ast\otimes(i-1)}(S_{lh}\mid L) - \pi^{\ast\otimes(i-1)}(S_{lh}\mid H)\,LR(h),
    \end{align*}
    where the second inequality follows because $\{h\}^{i-1}$, $S_{lh}$, and $\bar{S}_{lh}$ are disjoint by $\bm{S}\cap\{h\}^{i-1}=\emptyset$, the third inequality follows from (\ref{eq: prob_l_theta0}), and the fourth inequality follows because $S_{lh}\neq\emptyset$ and $\pi^{\ast\otimes(i-1)}(\{l\}^{i-1}\mid L)\leq \pi^{\ast\otimes(i-1)}(s_1,\dots,s_{i-1}\mid L)$ for any $(s_1,\dots,s_{i-1})\in S_{lh}$.
    
    Next, we show $LR(\bm{S})LR(h^\ast)>1$, which is equivalent to
    \begin{align*}
        &\ \pi^{\ast\otimes(i-1)}(\bm{S}\mid L) < \pi^{\ast\otimes(i-1)}(\bm{S}\mid H)\,LR(h^\ast) \\
        \Leftrightarrow&\ \pi^{\ast\otimes(i-1)}(\bar{S}_{lh}\mid L) - \pi^{\ast\otimes(i-1)}(\bar{S}_{lh}\mid H)\,LR(h^\ast) < \pi^{\ast\otimes(i-1)}(S_{lh}\mid H)\,LR(h^\ast) - \pi^{\ast\otimes(i-1)}(S_{lh}\mid L).
    \end{align*}
    The following calculation establishes this inequality:
    \begin{align*}
        \pi^{\ast\otimes(i-1)}(\bar{S}_{lh}\mid L) - \pi^{\ast\otimes(i-1)}(\bar{S}_{lh}\mid H)\,LR(h^\ast)
        &\leq \pi^{\ast\otimes(i-1)}(\bar{S}_{lh}\mid L) \\
        &\leq 1-\pi^{\ast\otimes(i-1)}(\{h\}^{i-1}\mid L)-\pi^{\ast\otimes(i-1)}(S_{lh}\mid L) \\
        &< \pi^{\ast\otimes(i-1)}(\{l\}^{i-1}\mid H)\,LR(h^\ast) - \pi^{\ast\otimes(i-1)}(S_{lh}\mid L) \\
        &\leq \pi^{\ast\otimes(i-1)}(S_{lh}\mid H)\,LR(h^\ast) - \pi^{\ast\otimes(i-1)}(S_{lh}\mid L),
    \end{align*}
    where the second inequality follows by the same 
    argument as above, the third inequality follows from (\ref{eq: prob_l_theta1}), and the fourth inequality follows by the same argument as in the previous calculation.

    \proofpart{Part (iii)}
    Fix $\bm{S}\subset S^{\ast i-1}$ with $\bm{S}\cap\{l,h,h^\ast\}^{i-1}\neq\emptyset$ and $\bm{S}\cap\{l,h\}^{i-1}=\emptyset$ arbitrarily. To show $\sigma_{\{l,h, h^\ast \}}(s_i)$ is uniquely optimal under $\bm{S}$ for any $s_i\in S^\ast$, it suffices to show that $LR(\bm{S})LR(l^\ast)<1<LR(\bm{S})LR(l)$. Let $S_{lhh^\ast}:=\bm{S}\cap\{l,h,h^\ast\}^{i-1}$ and $\bar{S}_{lhh^\ast}:=\bm{S}\setminus\{l,h,h^\ast\}^{i-1}$.
    
    First, we show $LR(\bm{S})LR(l^\ast)<1$. The inequality in (\ref{eq: LR_h*l*}) implies that $LR(s_1,\dots,s_{i-1})LR(l^\ast)<1$ for all $(s_1,\dots,s_{i-1})\in S^{\ast i-1}$. Hence $LR(\bm{S})LR(l^\ast)<1$ holds.
    
    Next, we show $LR(\bm{S})LR(l)>1$, which is equivalent to
    \begin{align*}
        &\ \pi^{\ast\otimes(i-1)}(\bm{S}\mid L) < \pi^{\ast\otimes(i-1)}(\bm{S}\mid H)\,LR(l) \\
        \Leftrightarrow&\ \pi^{\ast\otimes(i-1)}(\bar{S}_{lhh^\ast}\mid L) - \pi^{\ast\otimes(i-1)}(\bar{S}_{lhh^\ast}\mid H)\,LR(l) < \pi^{\ast\otimes(i-1)}(S_{lhh^\ast}\mid H)\,LR(l) - \pi^{\ast\otimes(i-1)}(S_{lhh^\ast}\mid L).
    \end{align*}
    The following calculation establishes this inequality:
    \begin{align*}
        \pi^{\ast\otimes(i-1)}(\bar{S}_{lhh^\ast}\mid L) - \pi^{\ast\otimes(i-1)}(\bar{S}_{lhh^\ast}\mid H)\,LR(l)
        &\leq \pi^{\ast\otimes(i-1)}(\bar{S}_{lhh^\ast}\mid L) \\
        &\leq 1-\pi^{\ast\otimes(i-1)}(\{l,h\}^{i-1}\mid L)-\pi^{\ast\otimes(i-1)}(S_{lhh^\ast}\mid L) \\
        &< \pi^{\ast\otimes(i-1)}(\{h^\ast\}^{i-1}\mid H)\,LR(l) - \pi^{\ast\otimes(i-1)}(S_{lhh^\ast}\mid L) \\
        &\leq \pi^{\ast\otimes(i-1)}(S_{lhh^\ast}\mid H)\,LR(l) - \pi^{\ast\otimes(i-1)}(S_{lhh^\ast}\mid L),
    \end{align*}
    where the second inequality follows because $\{l,h\}^{i-1}$, $S_{lhh^\ast}$, and $\bar{S}_{lhh^\ast}$ are disjoint by $\bm{S}\cap\{l,h\}^{i-1}=\emptyset$, the third inequality follows from (\ref{eq: prob_h*_theta1}), and the fourth inequality follows because $S_{lhh^\ast}\neq\emptyset$ and $\pi^{\ast\otimes(i-1)}(\{h^\ast\}^{i-1}\mid L)\leq\pi^{\ast\otimes(i-1)}(s_1,\dots,s_{i-1}\mid L)$ for any $(s_1,\dots,s_{i-1})\in S_{lhh^\ast}$.

    \proofpart{Part (iv)}
    Fix $\bm{S}\subset S^{\ast i-1}$ with $\bm{S}\cap\{l,h,h^\ast\}^{i-1}=\emptyset$ arbitrarily. To show $\sigma_\emptyset(s_i)$ is uniquely optimal under $\bm{S}$ for any $s_i\in S^\ast$, it suffices to show $LR(\bm{S})LR(h^\ast)<1$. Take $(s_1,\dots,s_{i-1})\in\bm{S}$ arbitrarily. Since $\bm{S}\cap\{l,h,h^\ast\}^{i-1}=\emptyset$, there exists $j$ such that $s_j=l^\ast$. Thus, $LR(s_1,\dots,s_{i-1})\leq LR(l^\ast)LR(h^\ast)^{i-2}$ for all $(s_1,\dots,s_{i-1})\in\bm{S}$. Therefore,
    \begin{align*}
        LR(\bm{S})LR(h^\ast)\leq LR(l^\ast)LR(h^\ast)^{i-1}<1,
    \end{align*}
    where the second inequality follows from (\ref{eq: LR_h*l*}).
\end{proof}

For each signal profile $(s_1,\dots,s_{j})\in S^{\ast\,j}$, define
\[
\underline{i}_1 := \min\bigl\{i \in \{1,\dots,j\} \mid s_i \in \{l,l^\ast\}\bigr\},
\qquad
\underline{i}_2 := \min\bigl\{i \in \{ \underline{i}_1+1,\dots, j \} \mid s_i = h^\ast\bigr\},
\]
whenever they exist.
Note that $\underline{i}_1$ and $\underline{i}_2$ depend on the signal profile $(s_1,\dots,s_{j})$,
but we omit this dependence for notational simplicity.
We introduce a partition of $S^{\ast\,j}$ as follows:
\[
\begin{aligned}
\bm{S}_1^{j} &:= \bigl\{(s_1,\dots,s_{j})\in S^{\ast\,j}
  \mid s_i\in\{h,h^\ast\}\ \text{for all } i=1,\dots,j \bigr\},\\[4pt]
\bm{S}_2^{j} &:= \bigl\{(s_1,\dots,s_{j})\in S^{\ast\,j}
  \mid \underline{i}_1 \text{ exists and }
  s_i\neq h^\ast\ \text{for all } i>\underline{i}_1 \bigr\},\\[4pt]
\bm{S}_3^{j} &:= \bigl\{(s_1,\dots,s_{j})\in S^{\ast\,j}
  \mid \underline{i}_2 \text{ exists and }
  s_i\neq l^\ast\ \text{for all } i>\underline{i}_2 \bigr\},\\[4pt]
\bm{S}_4^{j} &:= \bigl\{(s_1,\dots,s_{j})\in S^{\ast\,j}
  \mid \underline{i}_2 \text{ exists and }
  s_i = l^\ast\ \text{for some } i>\underline{i}_2 \bigr\}.
\end{aligned}
\]
\begin{lemma}\label{lem: equilibrium outcomes under particular signals}
    Under the information structure $\Pi^\ast$, the equilibrium strategy under each signal profile is derived as follows.
    \begin{enumerate}[(i)]
        \item Suppose $j=2,\dots,N$. For any signal profile $(s_1,\dots,s_{j-1}) \in \bm{S}_1^{j-1}$, $\sigma^\ast_j(s_1,\dots,s_{j-1},\cdot\mid\mathcal{N},\Pi^\ast) = \sigma_{\{h,h^\ast\}}(\cdot)$ for any $\mathcal{N}$.
        \item 
        Suppose $j=2,\dots,N$ and $\mathcal{N}_j=\{1,\dots,j-1 \}$. 
        For any signal profile $(s_1,\dots,s_{j-1}) \in \bm{S}_2^{j-1}$,
        $\sigma^\ast_j(s_1,\dots,s_{j-1},\cdot\mid\mathcal{N},\Pi^\ast) = \sigma_{\{h^\ast\}}(\cdot)$.
        
        \item
        Suppose $j=3,\dots,N$, $\mathcal{N}_j=\{1,\dots,j-1 \}$, and $j^\prime\in \mathcal{N}_{j^\prime+1}$ for all $j^\prime<j$. 
        For any signal profile $(s_1,\dots,s_{j-1}) \in \bm{S}_3^{j-1}$,
        $\sigma^\ast_j(s_1,\dots,s_{j-1},\cdot\mid\mathcal{N},\Pi^\ast) = \sigma_{\{l,h,h^\ast\}}(\cdot)$.
        
        \item
        Suppose $j=3,\dots,N$ and $\mathcal{N}_{j^\prime}=\{1,\dots,j^\prime-1\}$ for all $j^\prime\in\{2,\dots,j \}$.
        For any signal profile $(s_1,\dots,s_{j-1}) \in \bm{S}_4^{j-1}$,
        $\sigma^\ast_j(s_1,\dots,s_{j-1},\cdot\mid\mathcal{N},\Pi^\ast) = \sigma_{\emptyset}(\cdot)$.   
    \end{enumerate}
\end{lemma}

\begin{proof}[Proof of Lemma \ref{lem: equilibrium outcomes under particular signals}]
    \proofpart{Part (i)}
    Take  $j=2,\dots,N$, $\mathcal{N}$, and $(\tilde{s}_1,\dots,\tilde{s}_{j-1}) \in \bm{S}_1^{j-1}$ arbitrarily. Let $\bm{a}=(a_i)_{i\in \mathcal{N}_j}  :=(\sigma^\ast_i(\tilde{s}_1,\dots,\tilde{s}_{i}) \mid \mathcal{N}, \Pi^\ast)_{i\in \mathcal{N}_j}$ and $\bm{S}:=\sigma_j^{\ast-1}(\bm{a}\mid\mathcal{N},\Pi^\ast)$. 
    By Lemma \ref{lem: equilibrium characterization} (i), it suffices to show that $\bm{S}\cap\{h\}^{j-1}\neq\emptyset$, that is, $\sigma_i^\ast(h,\dots,h)=a_i$ for all $i\in\mathcal{N}_j$. 

    To this end, we show that $\sigma_i^\ast(s_1,\dots, s_i \mid \mathcal{N}, \Pi^\ast) = 1$ for all $(s_1,\dots, s_i) \in \{h,h^\ast \}^{i}$ and $i<j$ by induction on $i$. For $i=1$, since $LR(s_1)>1$ for any $s_1\in\{h,h^\ast\}$, we have $\sigma_1^\ast(s_1)=1$. 
    
    Next, fix $2\leq i<j$ and assume $\sigma_k^\ast(s_1,\dots,s_k\mid\mathcal{N},\Pi^\ast)=1$ for all $(s_1,\dots,s_k)\in\{h,h^\ast\}^k$ and all $k<i$. Fix any $(s_1,\dots,s_i)\in\{h,h^\ast\}^i$. Let $\bm{S}_i:=\sigma_i^{\ast-1}((a_k)_{k\in\mathcal{N}_i}\mid\mathcal{N},\Pi^\ast)$, where $a_k:=\sigma_k^\ast(s_1,\dots,s_k\mid\mathcal{N},\Pi^\ast)$ and $a_k = 1$ for all $k\in\mathcal{N}_i$ by the assumption.
    Since $(h,\dots,h)\in\{h,h^\ast\}^k$ for each $k\in\mathcal{N}_i$, the assumption also gives $\sigma_k^\ast(h,\dots,h)=1=a_k$, so $(h,\dots,h)\in\bm{S}_i$. Therefore, Lemma \ref{lem: equilibrium characterization} (i) implies $\sigma_i^\ast(\bm{S}_i,s_i)=1$, which completes the induction.

    \proofpart{Part (ii)}
    Take  $j=2,\dots,N$, $\mathcal{N}$ with $\mathcal{N}_j = \{1,\dots,j-1 \}$, and $(\tilde{s}_1,\dots,\tilde{s}_{j-1}) \in \bm{S}_2^{j-1}$ arbitrarily. 
    Let $\bm{a}=(a_i)^{j-1}_{i=1}  :=(\sigma^\ast_i(\tilde{s}_1,\dots,\tilde{s}_{i}) \mid \mathcal{N}, \Pi^\ast)^{j-1}_{i=1}$ and $\bm{S}:=\sigma_j^{\ast-1}(\bm{a}\mid\mathcal{N},\Pi^\ast)$. 
    By Lemma \ref{lem: equilibrium characterization} (ii), it suffices to show that $\bm{S}\cap\{h\}^{j-1}=\emptyset$ and $\bm{S}\cap\{l,h\}^{j-1}\neq\emptyset$. Since $\sigma_j^\ast(h,\dots,h)=1$ for all $i=1,\dots,j-1$ by Part (i), we have $\bm{S}\cap\{h\}^{j-1}=\emptyset$.

    We proceed to show $\bm{S}\cap\{l,h\}^{j-1}\neq\emptyset$. 
    Define $\tilde{\bm{s}}^\prime\in\{l,h\}^{j-1}$ by $\tilde{s}_i^\prime=l$ if $a_i=0$ and $\tilde{s}_i^\prime=h$ if $a_i=1$. We establish $\tilde{\bm{s}}^\prime \in \bm{S}$ by showing $\sigma_k^\ast(\tilde{s}_1^\prime,\dots,\tilde{s}_k^\prime)=a_k$ for all $k\in\{1,\dots,j-1\}$ by induction on $k$. First, $\sigma_k^\ast(\tilde{s}_k^\prime)=a_k$ if $\mathcal{N}_k=\emptyset$, since $\sigma_k^\ast(l)=0$ and $\sigma_k^\ast(h)=1$ and the definition of $\tilde{s}_k^\prime$.
    Next, assume $\mathcal{N}_k\neq\emptyset$ and $\sigma_m^\ast(\tilde{s}_1^\prime,\dots,\tilde{s}_m^\prime)=a_m$ for all $m<k$.
    Let $\bm{S}_k:=\sigma_k^{\ast-1}((a_m)_{m\in\mathcal{N}_k}\mid\mathcal{N},\Pi^\ast)$.
    By assumption, $(\tilde{s}_1^\prime,\dots,\tilde{s}_{k-1}^\prime)\in\bm{S}_k$, so $\bm{S}_k\cap\{l,h\}^{k-1}\neq\emptyset$.
    We begin with a case where $a_m=1$ for all $m\in\mathcal{N}_k$.
    By Part (i), $\sigma_m^\ast(h,\dots,h)=1=a_m$ for all $m\in\mathcal{N}_k$, so $(h,\dots,h)\in\bm{S}_k$, which implies $\sigma_k^\ast(\bm{S}_k,\tilde{s}_k^\prime)=1$ if $\tilde{s}_k^\prime=h$ and $\sigma_k^\ast(\bm{S}_k,\tilde{s}_k^\prime)=0$ if $\tilde{s}_k^\prime=l$ by Lemma \ref{lem: equilibrium characterization} (i), 
    By the definition of $\tilde{\bm{s}}^\prime$, we have $\sigma_k^\ast(\tilde{s}_1^\prime,\dots,\tilde{s}_k^\prime) = \sigma_k^\ast(\bm{S}_k,\tilde{s}_k^\prime) = a_k$.
    We proceed to a case where $a_m=0$ for some $m\in\mathcal{N}_k$.
    Since $\sigma_m^\ast(h,\dots,h)=1\neq 0=a_m$, we have $(h,\dots,h)\notin\bm{S}_k$ and thus $\bm{S}_k\cap\{h\}^{k-1}=\emptyset$.
    Also, since $(\tilde{s}_1,\dots,\tilde{s}_{j-1})\in\bm{S}_2^{j-1}$, we have $\tilde{s}_k\neq h^\ast$.
    Thus, Lemma \ref{lem: equilibrium characterization} (ii) implies $a_k=\sigma_k^\ast(\bm{S}_k,\tilde{s}_k^\prime)=0$, so $\tilde{s}_k^\prime=l$.
    Therefore, we have $\sigma_k^\ast(\tilde{s}_1^\prime,\dots,\tilde{s}_k^\prime)=\sigma_k^\ast(\bm{S}_k,l)=0=a_k$.

    \proofpart{Part (iii)}
    Take $j=3,\dots,N$, $\mathcal{N}$ with $\mathcal{N}_j=\{1,\dots,j-1\}$ and $j'\in\mathcal{N}_{j'+1}$ for all $j'<j$, and $(\tilde{s}_1,\dots,\tilde{s}_{j-1})\in\bm{S}_3^{j-1}$ arbitrarily.
    Let $\bm{a}=(a_i)_{i=1}^{j-1}:=(\sigma_i^\ast(\tilde{s}_1,\dots,\tilde{s}_i\mid\mathcal{N},\Pi^\ast))_{i=1}^{j-1}$ and $\bm{S}:=\sigma_j^{\ast-1}(\bm{a}\mid\mathcal{N},\Pi^\ast)$.
    By Lemma~\ref{lem: equilibrium characterization}~(iii), it suffices to show that $\bm{S}\cap\{l,h,h^\ast\}^{j-1}\neq\emptyset$ and $\bm{S}\cap\{l,h\}^{j-1}=\emptyset$.

    We first show $\bm{S}\cap\{l,h\}^{j-1}=\emptyset$.
    Since $(\tilde{s}_1,\dots,\tilde{s}_{j-1})\in\bm{S}_3^{j-1}$, the index $\underline{i}_1$ and $\underline{i}_2$ exists. Hereafter, we fix $\underline{i}_1$ and $\underline{i}_2$ as the indexes for the signal profile $(\tilde{s}_1,\dots,\tilde{s}_{j-1})$. 
    We begin by establishing $a_k=0$ for all $k=\underline{i}_1,\dots,\underline{i}_2-1$ and $a_{\underline{i}_2}=1$ by induction on $k$. 
    First, we have $a_{\underline{i}_1}=0$. This is because, since $\tilde{s}_k\in\{h,h^\ast\}$ for all $k<\underline{i}_1$ and $\tilde{s}_{\underline{i}_1} = l$, Part (i) implies $a_{\underline{i}_1}=0$. 
    Next, suppose $a_k=0$ for all $k=\underline{i}_1,\dots,m-1$ where $\underline{i}_1<m\leq \underline{i}_2$. Since $a_{m-1}=0$, $\sigma^{\ast-1}_{m}((a_i)_{i\in \mathcal{N}_m}\mid \mathcal{N}, \Pi^\ast) \cap \{h \}^{m-1} = \emptyset$. 
    On the other hand, signal profile $(s_1,\dots,s_{\underline{i}_1-1},s_{\underline{i}_1},s_{\underline{i}_1+1}, \dots,s_{m-1}) = (h,\dots,h,l,h,\dots,h)$ belongs to $\sigma^{\ast-1}_{m}((a_i)_{i\in \mathcal{N}_m}\mid \mathcal{N}, \Pi^\ast)$ by the assumption and Lemma \ref{lem: equilibrium characterization} (ii).
    Therefore, we have $\sigma^{\ast-1}_{m}((a_i)_{i\in \mathcal{N}_m}\mid \mathcal{N}, \Pi^\ast) \cap \{l, h \}^{m-1} \neq \emptyset$. 
    These two results imply $\sigma^{\ast}_{m}((a_i)_{i\in \mathcal{N}_m}, \tilde{s}_m \mid \mathcal{N}, \Pi^\ast)= \sigma_{\{ h^\ast\}}(\tilde{s}_m)$ by Lemma \ref{lem: equilibrium characterization} (ii).
    Thus, $a_{m} = \sigma_{\{ h^\ast\}}(\tilde{s}_m) =1$ by $\tilde{s}_m \neq h^\ast$. By the same argument, we have $a_{\underline{i}_2} = \sigma_{\{ h^\ast\}}(h^\ast) = 1$.

    Next, suppose for contradiction that there exists a signal profile $(s^\prime_1,\dots,s^\prime_j) \in \{l,h \}^j \cap  \bm{S}$. Since $a_{\underline{i}_1}=0$, $(s^\prime_1,\dots,s^\prime_{\underline{i}_2-1})\not\in\bm{S}^{\underline{i}_2-1}_1$. 
    Also, since $s_i^\prime  \neq h^\ast$ for all $i=1,\dots,\underline{i}_2-1$, $(s^\prime_1,\dots,s^\prime_{\underline{i}_2-1})$ does not belong to $\bm{S}^{\underline{i}_2-1}_3$ or $\bm{S}^{\underline{i}_2-1}_4$. 
    Therefore, $(s^\prime_1,\dots,s^\prime_{\underline{i}_2-1}) \in \bm{S}^{\underline{i}_2-1}_2$. By Lemma \ref{lem: equilibrium characterization} (ii), $\sigma_{\underline{i}_2}^\ast (s^\prime_1,\dots,s^\prime_{\underline{i}_2-1},s^\prime_{\underline{i}_2})=\sigma_{\{h^\ast \}}(s^\prime_{\underline{i}_2})=0$ as $s^\prime_{\underline{i}_2} \in \{l,h \}$. This contradicts to $a_{\underline{i}_2} = 1$.
    Now, we complete to establish $\bm{S}\cap\{l,h\}^{j-1}=\emptyset$.
    
    We next show $\bm{S}\cap\{l,h,h^\ast\}^{j-1}\neq\emptyset$ by establishing $\tilde{\bm{s}}^\prime = (\tilde{s}_1^\prime, \dots, \tilde{s}_{j-1}^\prime) \in \bm{S}\cap\{l,h,h^\ast\}^{j-1}$, where $\tilde{\bm{s}}^\prime$ is defined by $\tilde{s}_k^\prime=l$ if $\tilde{s}_k=l^\ast$, and $\tilde{s}_k^\prime=\tilde{s}_k$ otherwise.
    We show $\tilde{\bm{s}}^\prime\in\bm{S}$ by verifying $\sigma_k^\ast(\tilde{s}_1^\prime,\dots,\tilde{s}_k^\prime)=a_k$ for all $k=1,\dots,j-1$.
    For $k<\underline{i}_1$, since $\tilde{s}_k\in\{h,h^\ast\}$,  $\tilde{s}_k^\prime=\tilde{s}_k$, which implies $\sigma_k^\ast(\tilde{s}_1^\prime,\dots,\tilde{s}_k^\prime)=a_k=1$.
    The inductive argument above shows that $\sigma_k^\ast(\tilde{s}_1^\prime,\dots,\tilde{s}_k^\prime)=0=a_k$ for $k=\underline{i}_1,\dots,\underline{i}_2-1$, 
    and $\sigma_{\underline{i}_2}^\ast(\tilde{s}_1^\prime,\dots,\tilde{s}_{\underline{i}_2}^\prime)=1=a_{\underline{i}_2}$, as $\tilde{s}_{\underline{i}_2}^\prime = \tilde{s}_{\underline{i}_2} = h^\ast$.
    For $k>\underline{i}_2$, $\tilde{s}_k^\prime=\tilde{s}_k$ and the induced action profile up to agent $\underline{i}_2$ is the same for $\tilde{\bm{s}}^\prime$ as for $\tilde{\bm{s}}$, so $\sigma_k^\ast(\tilde{s}_1^\prime,\dots,\tilde{s}_k^\prime)=a_k$.
    Hence $\tilde{\bm{s}}^\prime\in\bm{S}\cap\{l,h,h^\ast\}^{j-1}$, which implies that $\bm{S}\cap\{l,h,h^\ast\}^{j-1}\neq\emptyset$.
    
    \proofpart{Part (iv)}
    Take $j=3,\dots,N$, $\mathcal{N}$ with $\mathcal{N}_{j'}=\{1,\dots,j'-1\}$ for all $j'\in\{2,\dots,j\}$, and $(\tilde{s}_1,\dots,\tilde{s}_{j-1})\in\bm{S}_4^{j-1}$ arbitrarily.
    Let $\bm{a}=(a_i)_{i=1}^{j-1}:=(\sigma_i^\ast(\tilde{s}_1,\dots,\tilde{s}_i\mid\mathcal{N},\Pi^\ast))_{i=1}^{j-1}$ and $\bm{S}:=\sigma_j^{\ast-1}(\bm{a}\mid\mathcal{N},\Pi^\ast)$.
    By Lemma~\ref{lem: equilibrium characterization}~(iv), it suffices to show $\bm{S}\cap\{l,h,h^\ast\}^{j-1}=\emptyset$.
    Since $(\tilde{s}_1,\dots,\tilde{s}_{j-1})\in\bm{S}_4^{j-1}$, there exists $i>\underline{i}_2$ such that $\tilde{s}_{i}=l^\ast$. Let $i^{\ast}:=\min\{i>\underline{i}_2\mid \tilde{s}_i=l^\ast\}$.
    Since $\tilde{s}_k\neq l^\ast$ for all $\underline{i}_2 < k < i^\ast$, we have $(\tilde{s}_1,\dots,\tilde{s}_{i^\ast-1})\in\bm{S}_3^{i^\ast-1}$.
    Part (iii) implies that $a_{i^\ast}=\sigma_{\{l,h,h^\ast\}}(l^\ast)=0$. Moreover, the argument in the proof of Part~(iii) gives $a_{\underline{i}_1}=0$ and $a_{\underline{i}_2}=1$.
    
    Now suppose for contradiction that there exists $\bm{s}^\prime=(s^\prime_1,\dots,s^\prime_{j-1})\in\{l,h,h^\ast\}^{j-1}\cap\bm{S}$.
    Since $a_{\underline{i}_1}=0$, $s^\prime_i  = l$ for some $i \leq \underline{i}_1$, which implies that $(s^\prime_1,\dots,s^\prime_{i^\ast-1}) \not\in \bm{S}_1^{i^\ast-1}$. 
    By the argument in the proof of Part (iii), $a_{\underline{i}_1}=0$ and $a_{\underline{i}_2}=1$ imply that $s^\prime_i = h^\ast$ for some $i=\underline{i}_1+1,\dots,\underline{i}_2$. Thus, $(s^\prime_1,\dots,s^\prime_{i^\ast-1}) \not\in \bm{S}_2^{i^\ast-1}$. 
    Also, since $(s^\prime_1,\dots,s^\prime_{i^\ast-1}) \in \{l,h,h^\ast \}^{i^\ast-1}$, $(s^\prime_1,\dots,s^\prime_{i^\ast-1}) \not\in \bm{S}_4^{i^\ast-1}$. Therefore, $(s^\prime_1,\dots,s^\prime_{i^\ast-1}) \in \bm{S}_3^{i^\ast-1}$. 
    Part (iii) implies $\sigma_{i^\ast}^\ast(s^\prime_1,\dots,s^\prime_{i^\ast-1}, s^\prime_{i^\ast} \mid\mathcal{N},\Pi^\ast)=\sigma_{\{l,h,h^\ast\}}(s^\prime_{i^\ast}) = 1$ for any $s^\prime_{i^\ast}\neq l^\ast$, which contradicts $a_{i^\ast}=0$.
\end{proof}

Lemma \ref{lem: equilibrium outcomes under particular signals} characterizes the equilibrium strategy under the complete network and information structure $\Pi^\ast$.
\begin{corollary}\label{cor: equilibrium strategy under N^C}
    $\sigma^\ast_N(s_1\dots,s_N \mid \mathcal{N}^C, \Pi^\ast)$ is characterized as follows.
    \[
    \sigma^\ast_N(s_1\dots,s_N \mid \mathcal{N}^C, \Pi^\ast) = 
    \begin{cases}
            \sigma_{\{h,h^\ast\}}(s_N) & \text{if } (s_1,\dots,s_{N-1})\in \bm{S}_1^{N-1}, \\
            \sigma_{\{h^\ast\}}(s_N) & \text{if } (s_1,\dots,s_{N-1})\in \bm{S}_2^{N-1}, \\
            \sigma_{\{l,h,h^\ast\}}(s_N) & \text{if } (s_1,\dots,s_{N-1})\in \bm{S}_3^{N-1},\\
            \sigma_{\emptyset}(s_N) & \text{if } (s_1,\dots,s_{N-1})\in \bm{S}_4^{N-1}.
        \end{cases}
    \]
\end{corollary}

\begin{lemma}\label{lem: N^C detects rare signals the most}
    Under the information structure $\Pi^\ast$, the following results hold for any network $\mathcal{N}$.
    \begin{enumerate}[(i)]
        \item Suppose $(s_1,\dots,s_{i-1}) \in \bm{S}^{i-1}_1$ and $a_j=\sigma^\ast_j(s_1,\dots,s_j\mid \mathcal{N}, \Pi^\ast)$ for each $j=1,\dots,i-1$. Then, $\sigma_i^{\ast -1} ((a_j)_{j\in\mathcal{N}_i} \mid \Pi^\ast, \mathcal{N}) \cap \{h \}^{i-1} \neq \emptyset$ 
        
        \item Suppose $(s_1,\dots,s_{i-1}) \in \bm{S}^{i-1}_2$ and $a_j=\sigma^\ast_j(s_1,\dots,s_j\mid \mathcal{N}, \Pi^\ast)$ for each $j=1,\dots,i-1$. Then, $\sigma_i^{\ast -1} ((a_j)_{j\in\mathcal{N}_i} \mid \Pi^\ast, \mathcal{N}) \cap \{l, h \}^{i-1} \neq \emptyset$ 
    
        \item Suppose $(s_1,\dots,s_{i-1}) \in \bm{S}^{i-1}_3$ and $a_j=\sigma^\ast_j(s_1,\dots,s_j\mid \mathcal{N}, \Pi^\ast)$ for each $j=1,\dots,i-1$. Then, $\sigma_i^{\ast -1} ((a_j)_{j\in\mathcal{N}_i} \mid \Pi^\ast, \mathcal{N}) \cap \{l,h,h^\ast \}^{i-1} \neq \emptyset$  
    \end{enumerate}
\end{lemma}

\begin{proof}[Proof of Lemma \ref{lem: N^C detects rare signals the most}]
    \proofpart{Part (i)}
    Take any $(\tilde{s}_1,\dots,\tilde{s}_{i-1}) \in \bm{S}^{i-1}_1$. We show $a_j:=\sigma_j^\ast(\tilde{s}_1,\dots,\tilde{s}_j\mid\mathcal{N},\Pi^\ast)=\sigma_j^\ast(h,\dots,h\mid\mathcal{N},\Pi^\ast)$ for any $j<i$. 
    
    Since $LR(h^\ast) > LR(h) > 1$, we have $\sigma_1^\ast(\tilde{s}_1\mid\mathcal{N},\Pi^\ast)=\sigma_1^\ast(h\mid\mathcal{N},\Pi^\ast) =1$.
    
    For induction, fix $j\in\{2,\dots,i-1\}$ and suppose $\sigma_k^\ast(\tilde{s}_1,\dots,\tilde{s}_k\mid\mathcal{N},\Pi^\ast)=\sigma_k^\ast(h,\dots,h\mid\mathcal{N},\Pi^\ast)$ for all $k<j$. Thus, $(h,\dots, h)\in\sigma_j^{\ast-1}((a_k)_{k\in\mathcal{N}_j}\mid\Pi^\ast,\mathcal{N})$.
    Lemma~\ref{lem: equilibrium characterization}~(i) implies that $\sigma_j^\ast(\tilde{s}_1,\dots,\tilde{s}_j\mid\mathcal{N},\Pi^\ast) = \sigma_{\{h, h^\ast\}}(\tilde{s}_j) =1$, because $\tilde{s}_j\in \{h,h^\ast \}$ by $\tilde{\bm{s}} \in \bm{S}^{i-1}_1$.
    Since $\sigma_j^\ast(h, \dots, h\mid\mathcal{N},\Pi^\ast) =1$ by Lemma~\ref{lem: equilibrium characterization}~(i), we have $\sigma_k^\ast(\tilde{s}_1,\dots,\tilde{s}_k\mid\mathcal{N},\Pi^\ast)=1=\sigma_k^\ast(h,\dots,h\mid\mathcal{N},\Pi^\ast)$,
    which completes the proof.
    
    \proofpart{Part (ii)}
    Take any $(\tilde{s}_1,\dots,\tilde{s}_{i-1}) \in \bm{S}^{i-1}_2$.
    Let $\tilde{\bm{s}}^\prime$ be the signal profile obtained from $\tilde{\bm{s}}$ by replacing each $h^\ast$ with $h$ and each $l^\ast$ with $l$, so that $\tilde{\bm{s}}^\prime\in\{l,h\}^{i-1}$.
    We show $a_j:=\sigma_j^\ast(\tilde{s}_1,\dots,\tilde{s}_j\mid\mathcal{N},\Pi^\ast)=\sigma_j^\ast(\tilde{s}_1^\prime,\dots,\tilde{s}_j^\prime\mid\mathcal{N},\Pi^\ast)$ for any $j<i$. 
    Since $\sigma_1^\ast(\tilde{s}_1\mid\mathcal{N},\Pi^\ast)=1$ if $\tilde{s}_1\in\{h,h^\ast \}$ and $\sigma_1^\ast(\tilde{s}_1\mid\mathcal{N},\Pi^\ast)=0$ if $\tilde{s}_1\in\{l,l^\ast \}$, we have $\sigma_1^\ast(\tilde{s}_1\mid\mathcal{N},\Pi^\ast)=\sigma_1^\ast(\tilde{s}_1^\prime \mid\mathcal{N},\Pi^\ast)$.
    
    For induction, fix $j\in\{2,\dots,i-1\}$ and suppose $\sigma_k^\ast(\tilde{s}_1,\dots,\tilde{s}_k\mid\mathcal{N},\Pi^\ast)=\sigma_k^\ast(\tilde{s}_1^\prime,\dots,\tilde{s}_k^\prime\mid\mathcal{N},\Pi^\ast)$ for all $k<j$. 
    First, consider a case in which $j\leq\underline{i}_1$. Then, $(\tilde{s}_1,\dots,\tilde{s}_{j-1}),(\tilde{s}_1^\prime,\dots,\tilde{s}_{j-1}^\prime)\in\bm{S}_1^{j-1}$. 
    By Part (i) and Lemma \ref{lem: equilibrium characterization}, we have $\sigma_j^\ast(\tilde{s}_1,\dots,\tilde{s}_j\mid\mathcal{N},\Pi^\ast) =\sigma_{\{h,h^\ast \}}(\tilde{s}_j)$ and $\sigma_j^\ast(\tilde{s}_1^\prime,\dots,\tilde{s}_j^\prime\mid\mathcal{N},\Pi^\ast) = \sigma_{\{h,h^\ast \}}(\tilde{s}_j^\prime)$. 
    Thus, $\sigma_j^\ast(\tilde{s}_1,\dots,\tilde{s}_j\mid\mathcal{N},\Pi^\ast) = \sigma_j^\ast(\tilde{s}_1^\prime,\dots,\tilde{s}_j^\prime\mid\mathcal{N},\Pi^\ast)$.
    
    Next, consider a case in which $\underline{i}_1 < j < i$. 
    Since $\sigma_k^\ast(\tilde{s}_1,\dots,\tilde{s}_k\mid\mathcal{N},\Pi^\ast)=\sigma_k^\ast(\tilde{s}_1^\prime,\dots,\tilde{s}_k^\prime\mid\mathcal{N},\Pi^\ast)$ for all $k<j$, $\sigma_j^{\ast -1}( (a_m)_{m\in\mathcal{N}_j} \mid\mathcal{N},\Pi^\ast) \cap \{l,h \}^{j-1} \neq \emptyset$. Thus, Lemma~\ref{lem: equilibrium characterization} places agent $j$ in case (i) or (ii). 
    In case (i), we have $\sigma_j^\ast(\tilde{s}_1,\dots,\tilde{s}_j\mid\mathcal{N},\Pi^\ast) = \sigma_{\{h,h^\ast\}}(\tilde{s}_j)=\sigma_{\{h,h^\ast\}}(\tilde{s}_j^\prime) = \sigma_j^\ast(\tilde{s}_1^\prime,\dots,\tilde{s}_j^\prime\mid\mathcal{N},\Pi^\ast)$.
    In case (ii), agent $j$'s equilibrium strategy is $\sigma_{h^\ast}$ under both signal profiles. Since $\tilde{\bm{s}}\in\bm{S}_2^{i-1}$ and $j>\underline{i}_1$ implies $\tilde{s}_j \neq h^\ast$, we have $\sigma_j^\ast(\tilde{s}_1,\dots,\tilde{s}_j\mid\mathcal{N},\Pi^\ast) = 0 = \sigma_j^\ast(\tilde{s}_1^\prime,\dots,\tilde{s}_j^\prime\mid\mathcal{N},\Pi^\ast)$.
    
    Therefore, we have established $\sigma_j^\ast(\tilde{s}_1,\dots,\tilde{s}_j\mid\mathcal{N},\Pi^\ast)=\sigma_j^\ast(\tilde{s}_1^\prime,\dots,\tilde{s}_j^\prime\mid\mathcal{N},\Pi^\ast)$ for any $j<i$, which completes the proof.
    
    \proofpart{Part (iii)}
    Take any $(\tilde{s}_1,\dots,\tilde{s}_{i-1}) \in \bm{S}^{i-1}_3$.
    Let $\tilde{\bm{s}}^\prime$ be the signal profile obtained from $\tilde{\bm{s}}$ by replacing each $l^\ast$ with $l$, so that $\tilde{\bm{s}}^\prime\in\{l,h,h^\ast\}^{i-1}$.
    We show $a_j:=\sigma_j^\ast(\tilde{s}_1,\dots,\tilde{s}_j\mid\mathcal{N},\Pi^\ast)=\sigma_j^\ast(\tilde{s}_1^\prime,\dots,\tilde{s}_j^\prime\mid\mathcal{N},\Pi^\ast)$ for any $j<i$. 
    
    Since $\sigma_1^\ast(\tilde{s}_1\mid\mathcal{N},\Pi^\ast)=1$ if $\tilde{s}_1\in\{h,h^\ast \}$ and $\sigma_1^\ast(\tilde{s}_1\mid\mathcal{N},\Pi^\ast)=0$ if $\tilde{s}_1\in\{l,l^\ast \}$, we have $\sigma_1^\ast(\tilde{s}_1\mid\mathcal{N},\Pi^\ast)=\sigma_1^\ast(\tilde{s}_1^\prime \mid\mathcal{N},\Pi^\ast)$.
    
    For induction, fix $j\in\{2,\dots,i-1\}$ and suppose $\sigma_k^\ast(\tilde{s}_1,\dots,\tilde{s}_k\mid\mathcal{N},\Pi^\ast)=\sigma_k^\ast(\tilde{s}_1^\prime,\dots,\tilde{s}_k^\prime\mid\mathcal{N},\Pi^\ast)$ for all $k<j$. 
    First, consider a case in which $j\leq\underline{i}_1$. 
    In this case, we can show $\sigma_k^\ast(\tilde{s}_1,\dots,\tilde{s}_k\mid\mathcal{N},\Pi^\ast)=\sigma_k^\ast(\tilde{s}_1^\prime,\dots,\tilde{s}_k^\prime\mid\mathcal{N},\Pi^\ast)$ by the same argument as that in Part (ii). 
    
    Next, consider a case in which $\underline{i}_1<j\leq\underline{i}_2$.
    Since $\sigma_k^\ast(\tilde{s}_1,\dots,\tilde{s}_k\mid\mathcal{N},\Pi^\ast)=\sigma_k^\ast(\tilde{s}_1^\prime,\dots,\tilde{s}_k^\prime\mid\mathcal{N},\Pi^\ast)$ for all $k<j$
    and $j\leq\underline{i}_2$,
    $\sigma_j^{\ast -1}( (a_m)_{m\in\mathcal{N}_j} \mid\mathcal{N},\Pi^\ast) \cap \{l,h \}^{j-1} \neq \emptyset$. 
    Thus, Lemma~\ref{lem: equilibrium characterization} places agent $j$ in case (i) or (ii).     
    Then, by the same argument as that in Part (ii), we have $\sigma_j^\ast(\tilde{s}_1,\dots,\tilde{s}_j\mid\mathcal{N},\Pi^\ast)=\sigma_j^\ast(\tilde{s}_1^\prime,\dots,\tilde{s}_j^\prime\mid\mathcal{N},\Pi^\ast)$.
    
    Finally, consider a case in which $\underline{i}_2 < j < i$.
    Since $\tilde{\bm{s}}\in\bm{S}_3^{i-1}$ and $j > \underline{i}_2$ implies $\tilde{s}_j\neq l^\ast$,
    we have $\tilde{s}_j = \tilde{s}_j^\prime$.
    Thus, $\sigma_j^\ast(\tilde{s}_1,\dots,\tilde{s}_j\mid\mathcal{N},\Pi^\ast)=\sigma_j^\ast(\tilde{s}_1^\prime,\dots,\tilde{s}_j^\prime\mid\mathcal{N},\Pi^\ast)$.
    
    Therefore, we have established $\sigma_j^\ast(\tilde{s}_1,\dots,\tilde{s}_j\mid\mathcal{N},\Pi^\ast)=\sigma_j^\ast(\tilde{s}_1^\prime,\dots,\tilde{s}_j^\prime\mid\mathcal{N},\Pi^\ast)$ for any $j<i$, which completes the proof.
\end{proof}

\subsection{Lemma}

For simplicity, we abbreviate $\sigma_i^\ast(s_1,\dots,s_i\mid\mathcal{N},\Pi^\ast)$ by $\sigma_i^{\ast\mathcal{N}}(s_1,\dots,s_i)$ and $\sigma_i^\ast(s_1,\dots,s_i\mid\mathcal{N}^C,\Pi^\ast)$ by $\sigma_i^{\ast C}(s_1,\dots,s_i)$. Similar abbreviations apply to other related notation.

First, we provide a lemma to establish weak optimality.

\begin{lemma}\label{lem: lemma for weak inequality}
    The following results hold for any network $\mathcal{N}$.
    \begin{enumerate}[(i)]
        \item
            For any history $\bm{a}=(a_i)_{i\in\mathcal{N}_N}$ satisfying $\sigma^{\ast\mathcal{N}-1}_N (\bm{a}) \cap \bm{S}_1^{N-1} \neq \emptyset$ and any $s_N\in S^\ast$,
            $\mathbb{E}_\theta[u(\sigma_N^{\ast\mathcal{N}}(s_1,\dots,s_N), \theta) \mid \sigma^{\ast\mathcal{N}-1}_N (\bm{a}) \cap \bm{S}_1^{N-1}, s_N] \leq \mathbb{E}_\theta[u(\sigma_N^{\ast C}(s_1,\dots,s_N), \theta) \mid \sigma^{\ast\mathcal{N}-1}_N (\bm{a}) \cap \bm{S}_1^{N-1}, s_N]$
        
        \item
            For any history $\bm{a}=(a_i)_{i\in\mathcal{N}_N}$ satisfying $\sigma^{\ast\mathcal{N}-1}_N (\bm{a}) \cap \bm{S}_2^{N-1} \neq \emptyset$ and any $s_N\in S^\ast$,
            $\mathbb{E}_\theta[u(\sigma_N^{\ast\mathcal{N}}(s_1,\dots,s_N), \theta) \mid \sigma^{\ast\mathcal{N}-1}_N (\bm{a}) \cap \bm{S}_2^{N-1}, s_N] \leq \mathbb{E}_\theta[u(\sigma_N^{\ast C}(s_1,\dots,s_N), \theta) \mid \sigma^{\ast\mathcal{N}-1}_N (\bm{a}) \cap \bm{S}_2^{N-1}, s_N]$
        
        \item
            For any history $\bm{a}=(a_i)_{i\in\mathcal{N}_N}$ satisfying $\sigma^{\ast\mathcal{N}-1}_N (\bm{a}) \cap \bm{S}_3^{N-1} \neq \emptyset$ and any $s_N\in S^\ast$,
            $\mathbb{E}_\theta[u(\sigma_N^{\ast\mathcal{N}}(s_1,\dots,s_N), \theta) \mid \sigma^{\ast\mathcal{N}-1}_N (\bm{a}) \cap \bm{S}_3^{N-1}, s_N] \leq \mathbb{E}_\theta[u(\sigma_N^{\ast C}(s_1,\dots,s_N), \theta) \mid \sigma^{\ast\mathcal{N}-1}_N (\bm{a}) \cap \bm{S}_3^{N-1}, s_N]$
        
        \item
            For any history $\bm{a}=(a_i)_{i\in\mathcal{N}_N}$ satisfying $\sigma^{\ast\mathcal{N}-1}_N (\bm{a}) \cap \bm{S}_4^{N-1} \neq \emptyset$ and any $s_N\in S^\ast$,
            $\mathbb{E}_\theta[u(\sigma_N^{\ast\mathcal{N}}(s_1,\dots,s_N), \theta) \mid \sigma^{\ast\mathcal{N}-1}_N (\bm{a}) \cap \bm{S}_4^{N-1}, s_N] \leq \mathbb{E}_\theta[u(\sigma_N^{\ast C}(s_1,\dots,s_N), \theta) \mid \sigma^{\ast\mathcal{N}-1}_N (\bm{a}) \cap \bm{S}_4^{N-1}, s_N]$ 
    \end{enumerate}
\end{lemma}

\begin{proof}[Proof of Lemma \ref{lem: lemma for weak inequality}]
    \proofpart{Part (i)}
        Let $\tilde{\bm{S}}:=\sigma^{\ast\mathcal{N}-1}_N (\bm{a}) \cap \bm{S}_1^{N-1}$. 
        By Lemma \ref{lem: N^C detects rare signals the most} (i), $\tilde{\bm{S}}$ satisfies the condition of Lemma \ref{lem: equilibrium characterization} (i). 
        Thus, $\sigma_N^{\ast\mathcal{N}}(\tilde{\bm{S}},s_N) =  \sigma_{\{ h,h^\ast \}} (s_N)$. 
        Therefore, 
    \begin{align*}
    \mathbb{E}_\theta\bigl[u(\sigma_N^{\ast\mathcal{N}}(s_1,\dots,s_N),\theta)\mid \tilde{\bm{S}},s_N\bigr]
    &\leq \mathbb{E}_\theta\bigl[u(\sigma_N^\ast(\tilde{\bm{S}},s_N),\theta)\mid \tilde{\bm{S}},s_N\bigr]\\
    &= \mathbb{E}_\theta\bigl[u(\sigma_{\{h,h^\ast\}}(s_N),\theta)\mid \tilde{\bm{S}},s_N\bigr]\\
    &= \mathbb{E}_\theta\bigl[u(\sigma_N^{\ast C}(s_1,\dots,s_N),\theta)\mid \tilde{\bm{S}},s_N\bigr],
    \end{align*}
    where the inequality holds because $\sigma_N^\ast(\tilde{\bm{S}},s_N)$ maximizes the expected payoff given $\tilde{\bm{S}}$ and $s_N$, the first equality follows from $\sigma_N^\ast(\tilde{\bm{S}},s_N)=\sigma_{\{h,h^\ast\}}(s_N)$, and the second equality follows from Corollary~\ref{cor: equilibrium strategy under N^C}~(i) and $\tilde{\bm{S}}\subset\bm{S}_1^{N-1}$.
    
    \proofpart{Parts (ii) and (iii)} Similar arguments to the proof of Part (i) establish Parts (ii) and (iii). 

    \proofpart{Part (iv)}
            Let $\tilde{\bm{S}}:=\sigma^{\ast\mathcal{N}-1}_N (\bm{a}) \cap \bm{S}_4^{N-1}$.
            Since $\tilde{\bm{S}}\subset\bm{S}_4^{N-1}$, every element of $\tilde{\bm{S}}$ contains $l^\ast$, so $\tilde{\bm{S}}\cap\{l,h,h^\ast\}^{N-1}=\emptyset$.
            Thus, Lemma~\ref{lem: equilibrium characterization}~(iv) implies $\sigma_N^\ast(\tilde{\bm{S}},s_N)=\sigma_\emptyset(s_N)=0$.
            Therefore,
            \begin{align*}
            \mathbb{E}_\theta\bigl[u(\sigma_N^{\ast\mathcal{N}}(s_1,\dots,s_N),\theta)\mid\tilde{\bm{S}},s_N\bigr]
            &\leq \mathbb{E}_\theta\bigl[u(\sigma_N^\ast(\tilde{\bm{S}},s_N),\theta)\mid\tilde{\bm{S}},s_N\bigr]\\
            &= \mathbb{E}_\theta\bigl[u(\sigma_\emptyset(s_N),\theta)\mid\tilde{\bm{S}},s_N\bigr]\\
            &= \mathbb{E}_\theta\bigl[u(\sigma_N^{\ast C}(s_1,\dots,s_N),\theta)\mid\tilde{\bm{S}},s_N\bigr],
            \end{align*}
            where the inequality holds because $\sigma_N^\ast(\tilde{\bm{S}},s_N)$ maximizes the expected payoff given $\tilde{\bm{S}}$ and $s_N$, the first equality follows from $\sigma_N^\ast(\tilde{\bm{S}},s_N)=\sigma_\emptyset(s_N)=0$, and the second equality follows from Corollary~\ref{cor: equilibrium strategy under N^C} and $\tilde{\bm{S}}\subset\bm{S}_4^{N-1}$.

\end{proof}

Next, we provide a lemma to establish strict optimality.

\begin{lemma}\label{lem: lemma for strict inequality}
    \begin{enumerate}[(i)]
        \item 
        Suppose that there exists $i\neq 1,N$ such that $i-1\not\in\mathcal{N}_i$ under network $\mathcal{N}$.
        There exists $(s_1,\dots,s_{N-1})\in S^{\ast N-1}$ and $s_N\in S^\ast$ such that $\sigma_N^{\ast\mathcal{N}}(s_1,\dots,s_N)\neq\sigma_N^{\ast C}(s_1,\dots,s_N)$.

        \item 
        Suppose that $i-1\in\mathcal{N}_i$ for all $i\neq 1,N$ and there exists $i\neq N$ such that $i\not\in\mathcal{N}_N$ under network $\mathcal{N}$.
        There exists $(s_1,\dots,s_{N-1})\in S^{\ast N-1}$ and $s_N\in S^\ast$ such that $\sigma_N^{\ast\mathcal{N}}(s_1,\dots,s_N)\neq\sigma_N^{\ast C}(s_1,\dots,s_N)$.

        \item 
        Suppose that $i-1\in\mathcal{N}_i$ for all $i\neq 1,N$ and $i\in\mathcal{N}_N$ for all $i\neq N$ under network $\mathcal{N}$ and $\mathcal{N}\neq \mathcal{N}^C$.
        There exists $(s_1,\dots,s_{N-1})\in S^{\ast N-1}$ and $s_N\in S^\ast$ such that $\sigma_N^{\ast\mathcal{N}}(s_1,\dots,s_N)\neq\sigma_N^{\ast C}(s_1,\dots,s_N)$. 
    \end{enumerate}
\end{lemma}

\begin{proof}[Proof of Lemma \ref{lem: lemma for strict inequality}]
    \proofpart{Part (i)}
    Take $\mathcal{N}$ satisfying $i-1\not\in\mathcal{N}_i$ for some $i\neq 1,N$ arbitrarily.
    Let $i^\ast:=\min\{i\neq 1,N \mid i-1\notin\mathcal{N}_i\}$ and consider the signal profile $\hat{\bm{s}}\in\bm{S}_3^{N-1}$ defined by $\hat{s}_{i^\ast-1}=l$, $\hat{s}_{i^\ast}=h^\ast$, and $\hat{s}_k=h$ for all other $k$.
    We show $\sigma_N^{\ast C}(\hat{\bm{s}}, l) = 1$ and $\sigma_N^{\ast\mathcal{N}}(\hat{\bm{s}}, l) = 0$, which establishes the claim with $(s_1,\dots,s_{N-1},s_N)=(\hat{\bm{s}},l)$.
    
    Since $\hat{\bm{s}}\in\bm{S}_3^{N-1}$, Corollary \ref{cor: equilibrium strategy under N^C} (iii) implies $\sigma_N^{\ast C}(\hat{\bm{s}},l)=1$.
    
    We show $\sigma_N^{\ast\mathcal{N}}(\hat{\bm{s}},l)=0$. Let $\hat{\bm{a}}$ be agent $N$'s history generated by $\hat{\bm{s}}$ under $\mathcal{N}$, and let $\hat{\bm{S}}:=\sigma_N^{\ast\mathcal{N}-1}(\hat{\bm{a}})$. It suffices to show $\hat{\bm{S}}\cap\{h\}^{N-1}=\emptyset$ and $\hat{\bm{S}}\cap\{l,h\}^{N-1}\neq\emptyset$, which together imply $\sigma_N^{\ast\mathcal{N}}(\hat{\bm{s}},l)=\sigma_{h^\ast}(l)=0$ by Lemma~\ref{lem: equilibrium characterization}~(ii).
    
    We first show $\hat{\bm{S}}\cap\{h\}^{N-1}=\emptyset$. Since $(\hat{s}_1,\dots,\hat{s}_{i^\ast-2})\in\{h\}^{i^\ast-2}\subset\bm{S}_1^{i^\ast-2}$, Lemma~\ref{lem: equilibrium outcomes under particular signals}~(i) gives $\hat{a}_{i^\ast-1}=\sigma_{\{h,h^\ast\}}(l)=0$. For any $(s_1,\dots,s_{N-1})\in\{h\}^{N-1}$, however, the same lemma gives $\sigma_{i^\ast-1}^{\ast\mathcal{N}}(h,\dots,h)=\sigma_{\{h,h^\ast\}}(h)=1\neq 0=\hat{a}_{i^\ast-1}$, so $(h,\dots,h)\notin\hat{\bm{S}}$ and hence $\hat{\bm{S}}\cap\{h\}^{N-1}=\emptyset$.
    
    We next show $\hat{\bm{S}}\cap\{l,h\}^{N-1}\neq\emptyset$ by verifying $\hat{\bm{s}}^\prime\in\hat{\bm{S}}$, where $\hat{\bm{s}}^\prime$ is defined by $\hat{s}^\prime_{i^\ast}=h$ and $\hat{s}^\prime_k=\hat{s}_k$ for all $k\neq i^\ast$. 
    We show $\sigma_k^{\ast\mathcal{N}}(\hat{\bm{s}}^\prime)=\hat{a}_k$ for all $k=1,\dots,N-1$. 
    By Lemma~\ref{lem: equilibrium outcomes under particular signals}~(i), agent $k<i^\ast-1$ chooses action $1$ and agent $i^\ast-1$ chooses action $0$ under both $\hat{\bm{s}}$ and $\hat{\bm{s}}^\prime$. 
    Since $i^\ast-1\notin\mathcal{N}_{i^\ast}$, agent $i^\ast$ observes only agents who chose $1$, so by Lemma~\ref{lem: equilibrium characterization}~(i) she chooses action $1$ under both $\hat{\bm{s}}$ and $\hat{\bm{s}}^\prime$. 
    To show $\sigma_k^{\ast\mathcal{N}}(\hat{\bm{s}}^\prime)=\hat{a}_k$ for all $k>i^\ast$ by induction on $k$, suppose that agent $m$ chooses action $\hat{a}_m$ under $\hat{\bm{s}}^\prime$ for all $m<k$; then agent $k$ observes the same history and receives the same signal $h$ under both $\hat{\bm{s}}$ and $\hat{\bm{s}}^\prime$, so she chooses the same action, i.e., $\sigma_k^{\ast\mathcal{N}}(\hat{\bm{s}}^\prime)=\hat{a}_k$. 
    Therefore $\hat{\bm{s}}^\prime\in\hat{\bm{S}}\cap\{l,h\}^{N-1}$, which completes the proof.

    \proofpart{Part (ii)}
    Take $\mathcal{N}$ satisfying $i-1\in\mathcal{N}_i$ for all $i\neq 1,N$ and $i\notin\mathcal{N}_N$ for some $i\neq N$ arbitrarily.
    Let $i^\ast:=\min\{i\neq N\mid i\notin\mathcal{N}_N\}$ and consider the signal profile $\hat{\bm{s}}\in\bm{S}_3^{N-1}$ defined by $\hat{s}_k=h$ for $k<i^\ast$, $\hat{s}_{i^\ast}=l$, and $\hat{s}_k=h^\ast$ for $k>i^\ast$.
    We show $\sigma_N^{\ast C}(\hat{\bm{s}}, l) = 1$ and $\sigma_N^{\ast\mathcal{N}}(\hat{\bm{s}}, l) = 0$, which establishes the claim with $(s_1,\dots,s_{N-1},s_N)=(\hat{\bm{s}},l)$.

Since $\hat{\bm{s}}\in\bm{S}_3^{N-1}$, Corollary~\ref{cor: equilibrium strategy under N^C}~(iii) implies $\sigma_N^{\ast C}(\hat{\bm{s}},l)=1$.

    We show $\sigma_N^{\ast\mathcal{N}}(\hat{\bm{s}},l)=0$.
    Let $\hat{\bm{a}}=(\hat{a}_j)_{j\in\mathcal{N}_N}$ be agent $N$'s history generated by $\hat{\bm{s}}$ under $\mathcal{N}$, and let $\hat{\bm{S}}:=\sigma_N^{\ast\mathcal{N}-1}(\hat{\bm{a}})$.
    It suffices to show $\hat{\bm{S}}\cap\{h\}^{N-1}\neq\emptyset$, which by Lemma~\ref{lem: equilibrium characterization}~(i) implies $\sigma_N^{\ast\mathcal{N}}(\hat{\bm{s}},l)=\sigma_{\{h,h^\ast\}}(l)=0$.

    We first show $\hat{a}_j=1$ for all $j\in\mathcal{N}_N$.
    For $j\in\mathcal{N}_N$ with $j<i^\ast$,  $(\hat{s}_1,\dots,\hat{s}_{j-1})=(h,\dots,h)\in\bm{S}_1^{j-1}$, so Lemma~\ref{lem: equilibrium outcomes under particular signals}~(i) gives $\hat{a}_j=\sigma_{\{h,h^\ast\}}(h)=1$.
    For $j\in\mathcal{N}_N$ with $j>i^\ast$, $\hat{s}_j=h^\ast$.
    Since  $(\hat{s}_1,\dots,\hat{s}_{j-1})\in\{l,h,h^\ast\}^{j-1}$,  $\sigma_j^{\ast-1}((\hat{a}_k)_{k\in\mathcal{N}_j})\cap\{l,h,h^\ast\}^{j-1}\neq\emptyset$.
    Lemma~\ref{lem: equilibrium characterization} therefore places agent $j$ in case (i), (ii), or (iii), each of which gives action $1$ for signal $h^\ast$, so $\hat{a}_j=1$.
    Since $i^\ast\notin\mathcal{N}_N$ by definition, $\hat{a}_j=1$ for all $j\in\mathcal{N}_N$.

    By Lemma~\ref{lem: equilibrium outcomes under particular signals}~(i), $(h,\dots,h)\in\{h\}^{N-1}$ also generates the same history for agent $N$ in $\mathcal{N}$, so $(h,\dots,h)\in\hat{\bm{S}}$ and hence $\hat{\bm{S}}\cap\{h\}^{N-1}\neq\emptyset$.
    Lemma~\ref{lem: equilibrium characterization}~(i) then implies $\sigma_N^{\ast\mathcal{N}}(\hat{\bm{s}},l)=\sigma_{\{h,h^\ast\}}(l)=0$, which completes the proof.

    \proofpart{Part (iii)}
    Take $\mathcal{N}$ satisfying $i-1\in\mathcal{N}_i$ for all $i\neq 1,N$, $i\in\mathcal{N}_N$ for all $i\neq N$, and $\mathcal{N}\neq\mathcal{N}^C$ arbitrarily.
    Since $\mathcal{N}\neq\mathcal{N}^C$ and $\mathcal{N}_N=\{1,\dots,N-1\}$, there exists $j\in\{2,\dots,N-1\}$ with $\mathcal{N}_j\neq\{1,\dots,j-1\}$.
    Let $j^\ast:=\min\{j\in\{2,\dots,N-1\}\mid\mathcal{N}_j\neq\{1,\dots,j-1\}\}$ and $i^\ast:=\min\{k<j^\ast\mid k\notin\mathcal{N}_{j^\ast}\}$.
    Since $j^\ast-1\in\mathcal{N}_{j^\ast}$ by assumption, we have $i^\ast\leq j^\ast-2$, and by minimality of $j^\ast$, all agents $k<j^\ast$ satisfy $\mathcal{N}_k=\{1,\dots,k-1\}$.
    Consider the signal profile $\hat{\bm{s}}\in\bm{S}_4^{N-1}$ defined by $\hat{s}_{i^\ast}=l$, $\hat{s}_{j^\ast}=l^\ast$, $\hat{s}_k=h^\ast$ for $i^\ast<k<j^\ast$, and $\hat{s}_k=h$ otherwise.
    We show $\sigma_N^{\ast C}(\hat{\bm{s}},h^\ast)=0$ and $\sigma_N^{\ast\mathcal{N}}(\hat{\bm{s}},h^\ast)=1$, which establishes the claim with $(s_1,\dots,s_{N-1},s_N)=(\hat{\bm{s}},h^\ast)$. 
    
    Since $\hat{\bm{s}}\in\bm{S}_4^{N-1}$, Corollary~\ref{cor: equilibrium strategy under N^C} implies $\sigma_N^{\ast C}(\hat{\bm{s}},h^\ast)=\sigma_\emptyset(h^\ast)=0$.

    We show $\sigma_N^{\ast\mathcal{N}}(\hat{\bm{s}},h^\ast)=1$.
    Let $\hat{\bm{a}}=(a_i)_{i=1}^{N-1}$ be the action profile generated by $\hat{\bm{s}}$ under $\mathcal{N}$, and let $\hat{\bm{S}}:=\sigma_N^{\ast\mathcal{N}-1}(\hat{\bm{a}})$.
    It suffices to show $\hat{\bm{S}}\cap\{l,h,h^\ast\}^{N-1}\neq\emptyset$ by Lemma~\ref{lem: equilibrium characterization}.

    We first derive $\hat{a}_i$ for $i < j^\ast$. 
    Since $(\hat{s}_1,\dots,\hat{s}_{i})=(h,\dots,h)\in\bm{S}_1^{i}$ for any $i<i^\ast$,
    Lemma~\ref{lem: equilibrium outcomes under particular signals}~(i) implies $\hat{a}_{i}=\sigma_{\{h,h^\ast\}}(h)=1$ for any $i<i^\ast$ and $\hat{a}_{i^\ast}=\sigma_{\{h,h^\ast\}}(l)=0$. 
    For $i$ with $i^\ast < i < j^\ast$, since $(\hat{s}_1,\dots,\hat{s}_{i-1})\notin\bm{S}_4^{i-1}$, 
     Lemma \ref{lem: equilibrium characterization} and \ref{lem: N^C detects rare signals the most} imply $\hat{a}_{i}=1$ as agent $i$ receives signal $h^\ast$.

    We now show $\hat{\bm{S}}\cap\{l,h,h^\ast\}^{N-1}\neq\emptyset$ by verifying $\hat{\bm{s}}^{\prime}\in\hat{\bm{S}}$, where $\hat{\bm{s}}^{\prime} \in \{l,h,h^\ast\}^{N-1}$ is defined by $\hat{s}^{\prime}_{j^\ast}=l$ and $\hat{s}^{\prime}_k=\hat{s}_k$ for all $k\neq j^\ast$.
    We show $\sigma_k^{\ast\mathcal{N}}(\hat{\bm{s}}^{\prime})=\hat{a}_k$ for all $k=1,\dots,N-1$.
    For $k<j^\ast$, $\hat{s}_k=\hat{s}_k^\prime$, so $\sigma_k^{\ast\mathcal{N}}(\hat{\bm{s}}^{\prime})=\hat{a}_k$.
    For $k=j^\ast$, agent $j^\ast$ observes the same history $(\hat{a}_m)_{m\in\mathcal{N}_{j^\ast}}$.
    Recall that $\hat{a}_m =1$ for all $m\in\{ 1,\dots,i^\ast -1,i^\ast +1,\dots,j^\ast-1 \}$.
    Thus, we have $\hat{a}_m=1$ for all $m\in\mathcal{N}_j^\ast$. 
    Since this history is also generated by $(h,\dots,h)$, by Lemma~\ref{lem: equilibrium outcomes under particular signals}~(i). Lemma~\ref{lem: equilibrium characterization}~(i) then implies $\sigma_{j^\ast}^{\ast\mathcal{N}}(\hat{\bm{s}}^{\prime})=\sigma_{\{h,h^\ast\}}(l)=0=\hat{a}_{j^\ast}$. 
    To show $\sigma_{k}^{\ast\mathcal{N}}(\hat{\bm{s}}^{\prime})=\hat{a}_{k}$ for all $k>j^\ast$ by induction, fix $k>j^\ast$ and suppose that this equation holds for all $m<k$. Then, agent $k$ observes the same history and receives the same signal under both $\hat{\bm{s}}$ and $\hat{\bm{s}}^{\prime}$. Thus, $\sigma_{k}^{\ast\mathcal{N}}(\hat{\bm{s}}^{\prime})=\hat{a}_{k}$.
    Now, we complete to establish $\hat{\bm{s}}^{\prime}\in\hat{\bm{S}}$, which completes the proof.
\end{proof}

\subsection{Proof of Theorem \ref{thm: complete}}
\begin{proof}[Proof of Theorem \ref{thm: complete}]
    Take any $\mathcal{N}^\prime \neq\mathcal{N}^C$ satisfying the following two conditions: (I) there exists $i\neq 1,N$ such that $\mathcal{N}^\prime_i \neq \{1,\dots,i-1 \}$; (II) $\mathcal{N}^\prime_N \neq \{1,\dots,N-1 \}$. 
    Then, we can define $\mathcal{N} \neq \mathcal{N}^\prime, \mathcal{N}^C$ by $\mathcal{N}_N=\{1,\dots,N-1\}$ and $\mathcal{N}_i=\mathcal{N}^\prime_i$ for all $i<N$.
    Lemma~\ref{lem: ignoring principle} implies $V_N(\mathcal{N}^\prime,\Pi^\ast)\leq V_N(\mathcal{N},\Pi^\ast)$.
    Therefore, it suffices to establish $V_N(\mathcal{N},\Pi^\ast)<V_N(\mathcal{N}^C,\Pi^\ast)$ for any $\mathcal{N}$ that violates at least one of Condition (I) and (II). Fix such a network $\mathcal{N}$ arbitrarily.

    Let $A_\mathcal{N}$ denote the set of agent $N$'s histories $\bm{a}=(a_i)_{i \in \mathcal{N}_N}$ that can arise under $\mathcal{N}$, and for each $\bm{a}\in A_\mathcal{N}$ and $k\in\{1,2,3,4\}$ let $\tilde{\bm{S}}_k(\bm{a}):=\sigma^{\ast\mathcal{N}-1}_N(\bm{a})\cap\bm{S}_k^{N-1}$.
    The collection $\{\tilde{\bm{S}}_k(\bm{a}):\bm{a}\in A_\mathcal{N},\,k\in\{1,2,3,4\},\,\tilde{\bm{S}}_k(\bm{a})\neq\emptyset\}$ forms a partition of $S^{\ast N-1}$.
    Therefore,
    \begin{align*}
    V_N(\mathcal{N},\Pi^\ast)
    &= \sum_{s_N\in S^\ast}\sum_{\bm{a}\in A_\mathcal{N}}\sum_{k=1}^{4}
       \mathbb{P}\bigl(\tilde{\bm{S}}_k(\bm{a}),s_N\bigr)\cdot
       \mathbb{E}_\theta\bigl[u(\sigma_N^{\ast\mathcal{N}}(s_1,\dots,s_N),\theta)\mid\tilde{\bm{S}}_k(\bm{a}),s_N\bigr]\\
    &\leq \sum_{s_N\in S^\ast}\sum_{\bm{a}\in A_\mathcal{N}}\sum_{k=1}^{4}
       \mathbb{P}\bigl(\tilde{\bm{S}}_k(\bm{a}),s_N\bigr)\cdot
       \mathbb{E}_\theta\bigl[u(\sigma_N^{\ast C}(s_1,\dots,s_N),\theta)\mid\tilde{\bm{S}}_k(\bm{a}),s_N\bigr]\\
    &= V_N(\mathcal{N}^C,\Pi^\ast),
    \end{align*}
    where the inequality follows from Lemma~\ref{lem: lemma for weak inequality} applied to each term, and the first and last equalities hold by the law of iterated expectations.

    To establish the strict inequality $V_N(\mathcal{N},\Pi^\ast) < V_N(\mathcal{N}^C,\Pi^\ast)$, it suffices to find a signal profile $(s_1,\dots,s_{N-1})\in S^{\ast N-1}$ and $s_N\in S^\ast$ such that $\sigma_N^{\ast\mathcal{N}}(s_1,\dots,s_N)\neq\sigma_N^{\ast C}(s_1,\dots,s_N)$.
    Indeed, suppose such a profile exists, and let $k\in\{1,2,3,4\}$ be the index satisfying $(s_1,\dots,s_{N-1})\in\bm{S}_k^{N-1}$, and let $\bm{a}\in A_\mathcal{N}$ be the history generated by $(s_1,\dots,s_{N-1})$ under $\mathcal{N}$.
    Since $\sigma_N^{\ast\mathcal{N}}(s_1,\dots,s_N)\neq\sigma_N^{\ast C}(s_1,\dots,s_N)=\sigma_N^\ast(\tilde{\bm{S}}_k(\bm{a}),s_N)$ by Corollary \ref{cor: equilibrium strategy under N^C} and Lemma~\ref{lem: equilibrium characterization} guarantees that $\sigma_N^\ast(\tilde{\bm{S}}_k(\bm{a}),s_N)$ uniquely  maximizes the expected payoff given $\tilde{\bm{S}}_k(\bm{a})$ and $s_N$, we have
    \begin{align*}
    \mathbb{E}_\theta\bigl[u(\sigma_N^{\ast\mathcal{N}}(s_1,\dots,s_N),\theta)\mid\tilde{\bm{S}}_k(\bm{a}),s_N\bigr]
    &< \mathbb{E}_\theta\bigl[u(\sigma_N^\ast(\tilde{\bm{S}}_k(\bm{a}),s_N),\theta)\mid\tilde{\bm{S}}_k(\bm{a}),s_N\bigr]\\
    &= \mathbb{E}_\theta\bigl[u(\sigma_N^{\ast C}(s_1,\dots,s_N),\theta)\mid\tilde{\bm{S}}_k(\bm{a}),s_N\bigr].
    \end{align*}
    Since $\mathbb{P}(\tilde{\bm{S}}_k(\bm{a}),s_N)>0$, this strict inequality for one term implies $V_N(\mathcal{N},\Pi^\ast) < V_N(\mathcal{N}^C,\Pi^\ast)$.

    Since $\mathcal{N}$ violates at least one of Condition (I) and (II), we can apply one of Lemma \ref{lem: lemma for strict inequality} (i), (ii), and (iii) to find a signal profile $(s_1,\dots,s_{N-1})\in S^{\ast N-1}$ and $s_N\in S^\ast$ such that $\sigma_N^{\ast\mathcal{N}}(s_1,\dots,s_N)\neq\sigma_N^{\ast C}(s_1,\dots,s_N)$, which completes the proof.
\end{proof}

\subsection{Proof of Proposition \ref{prop: pareto dominate}}
\begin{proof}[Proof of Proposition \ref{prop: pareto dominate}]
    Fix any agent $j\leq N$ and network $\mathcal{N}$.
    Let $A_\mathcal{N}$ denote the set of agent $j$'s histories $\bm{a}=(a_i)_{i \in \mathcal{N}_j}$ that can arise under $\mathcal{N}$, and for each $\bm{a}\in A_\mathcal{N}$ and $k\in\{1,2,3,4\}$.
    Let $\tilde{\bm{S}}_k(\bm{a}):=\sigma^{\ast\mathcal{N}-1}_j(\bm{a})\cap\bm{S}_k^{j-1}$.
    The collection $\{\tilde{\bm{S}}_k(\bm{a})\mid \bm{a}\in A_\mathcal{N},\,k\in\{1,2,3,4\},\,\tilde{\bm{S}}_k(\bm{a})\neq\emptyset\}$ forms a partition of $S^{\ast j-1}$.
    Since the corresponding result in Lemma \ref{lem: lemma for weak inequality} for agent $j$ for agent $j$ can be established by the same argument as in the proof of that lemma, we have
    \begin{align*}
    V_j(\mathcal{N},\Pi^\ast)
    &= \sum_{s_j\in S^\ast}\sum_{\bm{a}\in A_\mathcal{N}}\sum_{k=1}^{4}
       \mathbb{P}\bigl(\tilde{\bm{S}}_k(\bm{a}),s_j\bigr)\cdot
       \mathbb{E}_\theta\bigl[u(\sigma_j^{\ast\mathcal{N}}(s_1,\dots,s_j),\theta)\mid\tilde{\bm{S}}_k(\bm{a}),s_j\bigr]\\
    &\leq \sum_{s_j\in S^\ast}\sum_{\bm{a}\in A_\mathcal{N}}\sum_{k=1}^{4}
       \mathbb{P}\bigl(\tilde{\bm{S}}_k(\bm{a}),s_j\bigr)\cdot
       \mathbb{E}_\theta\bigl[u(\sigma_j^{\ast C}(s_1,\dots,s_j),\theta)\mid\tilde{\bm{S}}_k(\bm{a}),s_j\bigr]\\
    &= V_j(\mathcal{N}^C,\Pi^\ast).
    \end{align*} 
\end{proof}

\section{Proof of Proposition \ref{prop:independent-optimum}}
\label{app:guinea-pigs-proof}

Write $q=(1-\varepsilon)(1-p)$. Each guinea pig independently chooses $1$ with probability $1-q$ in state $H$ and $q$ in state $L$.
Thus, conditional on $H$, the number of actions equal to $1$ among the first $M$ agents has distribution $\operatorname{Bin}(M,1-q)$.

\begin{lemma} \label{lem:independent-payoff}
For $M\in\{1,\ldots,N-1\}$, under $\mathcal N^M$,
\begin{align*}
    V_N(\mathcal N^M,\Pi^{\varepsilon,p})
    =1-(1-\varepsilon)^{N-M}\left[
    \sum_{j=0}^{\lfloor\frac{M-1}{2}\rfloor}\binom{M}{j}(1-q)^j q^{M-j}
    +(1-p)\binom{M}{M/2}[q(1-q)]^{M/2}\right].
\end{align*}
Here $\binom{M}{M/2}=0$ for odd $M$.
For every integer $k\geq1$ with $2k+1\leq N-1$,
\begin{align*}
    V_N(\mathcal N^{2k},\Pi^{\varepsilon,p})-V_N(\mathcal N^{2k+1},\Pi^{\varepsilon,p})
    =\varepsilon(1-\varepsilon)^{N-2k-1}\sum_{j=0}^{k-1}\binom{2k}{j}(1-q)^j q^{2k-j}>0.
\end{align*}
\end{lemma}

\begin{proof}[Proof of Lemma \ref{lem:independent-payoff}]
Fix $M$. By conditional independence, an initial history with $j$ actions equal to $1$ gives public likelihood ratio $[(1-q)/q]^{2j-M}$ for each $j\in[0,M]$.
Since $(1-q)/q>p/(1-p)$, a strict initial majority determines the next action unless the agent receives the conclusive signal for the opposite state.
If the majority favors $1$, another action $1$ multiplies the public likelihood ratio by $1/(1-\varepsilon)$; if it favors $0$, another action $0$ multiplies it by $1-\varepsilon$.
Thus, following the majority strengthens the same strict inequality, while the first contrary action reveals the state and makes all subsequent actions correct.

Conditional on $H$, an incorrect initial majority therefore leads to a terminal error with probability $(1-\varepsilon)^{N-M}$, and a correct majority leads to no error.
After an initial tie, agent $M+1$ acts on her private signal and chooses $0$ with probability $q$.
Her action creates public likelihood ratio $(1-q)/q$ or its reciprocal, and hence the same argument gives error probability $q(1-\varepsilon)^{N-M-1}=(1-p)(1-\varepsilon)^{N-M}$.
All choices at positive-probability information sets are strict, and thus interchanging the states, signals, and actions preserves the strategy profile and gives $\alpha_N^H(0)=\alpha_N^L(1)$.
Weighting these cases by the binomial probabilities proves the payoff formula.

For the parity comparison, under $\mathcal N^{2k+1}$ a majority of zeros arises exactly when the first $2k$ actions have such a majority, or they tie and agent $2k+1$ chooses $0$. Hence,
\begin{align*}
    \sum_{j=0}^{k}\binom{2k+1}{j}(1-q)^j q^{2k+1-j}
    =\sum_{j=0}^{k-1}\binom{2k}{j}(1-q)^j q^{2k-j}
    +q\binom{2k}{k}[q(1-q)]^k.
\end{align*}
Substituting into the payoff formula and using $q=(1-\varepsilon)(1-p)$ cancels the tie terms and yields the stated difference.
It is strictly positive because $\varepsilon>0$ and the all-zero initial history has probability $q^{2k}>0$ in state $H$.
\end{proof}

Set $z=1/(1-\varepsilon)-(1-p)=(1-q)/(1-\varepsilon)>p$, and define, for $k\geq0$,
\begin{align*}
    J_k(z;p)=\sum_{j=0}^{k-1}\binom{2k}{j}z^j(1-p)^{2k-j}+(1-p)\binom{2k}{k}[z(1-p)]^k.
\end{align*}
Empty sums are zero, and hence $J_0(z;p)=1-p$.
Substituting $q=(1-\varepsilon)(1-p)$ and $1-q=(1-\varepsilon)z$ into the payoff formula gives
\begin{align*}
    1-V_N(\mathcal N^1,\Pi^{\varepsilon,p})&=(1-\varepsilon)^N J_0(z;p),\\
    1-V_N(\mathcal N^{2k},\Pi^{\varepsilon,p})&=(1-\varepsilon)^N J_k(z;p),\qquad 1\leq k\leq K,
\end{align*}
where $K=\lfloor(N-1)/2\rfloor$.
Thus, $k=0$ represents $M=1$, and $k\geq1$ represents $M=2k$.

\begin{lemma} \label{lem:independent-ratios}
Fix $p\in(1/2,1)$.
For $z>p$, the sequence $J_k(z;p)$ is strictly log-convex: $[J_{k+1}(z;p)]^2<J_k(z;p)J_{k+2}(z;p)$ for $k\geq0$.
For each $k\geq0$, $J_{k+1}(z;p)/J_k(z;p)$ is strictly increasing in $z>0$, is below one at $z=p$, and diverges as $z\to\infty$.
For $z>p$, $\lim_{k\to\infty}J_{k+1}(z;p)/J_k(z;p)=4(1-p)z$.
\end{lemma}

\begin{proof}[Proof of Lemma \ref{lem:independent-ratios}]
Fix $z>p$, set $\varepsilon=1-1/[z+(1-p)]$, and write $q=(1-\varepsilon)(1-p)\in(0,1-p)$.
For $k\geq0$, define
\begin{align*}
    T_k(q)=\sum_{j=k+1}^{2k+1}\binom{2k+1}{j}q^j(1-q)^{2k+1-j},\qquad
    C_k(q)=\binom{2k}{k}[q(1-q)]^k.
\end{align*}
Conditioning the last of $2k+1$ independent Bernoulli draws on the first $2k$ draws gives
\begin{align*}
    (1-\varepsilon)^{2k}J_k(z;p)=T_k(q)+(1-p-q)C_k(q).
\end{align*}
Adding two draws changes whether successes form a majority only when the original count is $k$ or $k+1$. Thus,
\begin{align*}
    T_k(q)-T_{k+1}(q)
    &=\binom{2k+1}{k+1}q^{k+1}(1-q)^{k+2}-\binom{2k+1}{k}q^{k+2}(1-q)^{k+1}\\
    &=\frac{1-2q}{2}C_{k+1}(q).
\end{align*}
Summing these differences from $k$ to $m\geq k$ cancels the intermediate terms:
\begin{align*}
    T_k(q)-T_{m+1}(q)
    =\sum_{j=k}^{m}[T_j(q)-T_{j+1}(q)]
    =\frac{1-2q}{2}\sum_{j=k+1}^{m+1}C_j(q).
\end{align*}
Since $q<1/2$, the law of large numbers gives $T_{m+1}(q)\to0$ as $m\to\infty$. Taking this limit yields
\begin{align*}
    T_k(q)=\frac{1-2q}{2}\sum_{j=k+1}^{\infty}C_j(q).
\end{align*}
Consequently,
\begin{align*}
    (1-\varepsilon)^{2k}J_k(z;p)
    &=(1-p-q)C_k(q)+\frac{1-2q}{2}\sum_{r=1}^{\infty}C_{k+r}(q)\\
    &=\sum_{r=0}^{\infty}w_r C_{k+r}(q),
\end{align*}
where $w_0=1-p-q=\varepsilon(1-p)>0$ and $w_r=(1-2q)/2>0$ for $r\geq1$.
These weights do not depend on $k$.
Direct calculation gives
\begin{align*}
    \frac{C_{k+1}(q)}{C_k(q)}
    =\frac{(2k+2)(2k+1)}{(k+1)^2}q(1-q)
    =\left(4-\frac{2}{k+1}\right)q(1-q).
\end{align*}
This ratio strictly increases in $k$ and is below $4q(1-q)<1$.
Thus, $C_{k+1}(q)^2<C_k(q)C_{k+2}(q)$, and the weighted sums above converge.
Suppressing the argument $q$, this strict inequality and Cauchy--Schwarz give
\begin{align*}
    \bigl[(1-\varepsilon)^{2k+2}J_{k+1}\bigr]^2
    &=\left(\sum_{r=0}^{\infty}w_r C_{k+r+1}\right)^2\\
    &<\left(\sum_{r=0}^{\infty}w_r\sqrt{C_{k+r}C_{k+r+2}}\right)^2\\
    &\leq\left(\sum_{r=0}^{\infty}w_r C_{k+r}\right)
    \left(\sum_{r=0}^{\infty}w_r C_{k+r+2}\right)=(1-\varepsilon)^{4k+4}J_kJ_{k+2}.
\end{align*}
Thus, $J_{k+1}^2<J_kJ_{k+2}$, and hence $J_{k+1}(z;p)/J_k(z;p)$ strictly increase in $k$.

To obtain their limit, note that, for every $r\geq0$,
\begin{align*}
    \left(4-\frac{2}{k+1}\right)q(1-q)
    \leq\frac{C_{k+r+1}}{C_{k+r}}<4q(1-q).
\end{align*}
Multiplying by $w_r C_{k+r}$, summing over $r\geq0$, and dividing by $\sum_{r=0}^{\infty}w_r C_{k+r}$ yields
\begin{align*}
    \left(4-\frac{2}{k+1}\right)q(1-q)
    \leq(1-\varepsilon)^2\frac{J_{k+1}}{J_k}<4q(1-q).
\end{align*}
Since $q(1-q)/(1-\varepsilon)^2=(1-p)z$, we obtain
\begin{align*}
    \left(4-\frac{2}{k+1}\right)(1-p)z
    \leq\frac{J_{k+1}}{J_k}<4(1-p)z.
\end{align*}
Both bounds converge to $4(1-p)z$ as $k\to\infty$, proving the stated limit.

Next, we establish monotonicity in $z$.
For $k=0$, $J_1/J_0=(1-p)(1+2z)$.
For $k\geq1$, let $a_j$ and $b_j$ be the coefficients of $z^j$ in $J_{k+1}$ and $J_k$, respectively, with $b_{k+1}=0$.
For $0\leq j\leq k$,
\begin{align*}
    \frac{a_j}{b_j}=\begin{cases}
        (1-p)^2\dfrac{(2k+2)(2k+1)}{(2k+2-j)(2k+1-j)}, & 0\leq j<k,\\[4pt]
        (1-p)\dfrac{(2k+2)(2k+1)}{(k+2)(k+1)}, & j=k.
    \end{cases}
\end{align*}
These ratios strictly increase with $j$; the last exceeds the preceding one by the factor $(k+3)/[(1-p)(k+1)]>1$.
Since $a_{k+1}>0$, it follows that, for $z>0$,
\begin{align*}
    (\partial_z J_{k+1})J_k-J_{k+1}(\partial_z J_k)
    =\sum_{0\leq j<i\leq k+1}(i-j)(a_i b_j-a_j b_i)z^{i+j-1}>0.
\end{align*}
Hence, $J_{k+1}/J_k$ is strictly increasing in $z$.
At $z=p$, the binomial-tail decomposition gives $J_k(p;p)=T_k(1-p)$, and hence
\begin{align*}
    J_k(p;p)-J_{k+1}(p;p)=\frac{2p-1}{2}C_{k+1}(1-p)>0.
\end{align*}
This evaluates the polynomials at the boundary and does not require an equilibrium calculation at $\varepsilon=0$.
Finally, $J_{k+1}$ and $J_k$ have degrees $k+1$ and $k$ with positive leading coefficients, and thus their ratio diverges as $z\to\infty$.
\end{proof}

\begin{proof}[Proof of Proposition \ref{prop:independent-optimum}]
For fixed $p$, $z=1/(1-\varepsilon)-(1-p)$ increases continuously from $p$ to infinity as $\varepsilon$ increases from zero to one.
Lemma \ref{lem:independent-ratios} therefore gives a unique $\varepsilon_k(p)\in(0,1)$ at which $J_{k+1}(z;p)=J_k(z;p)$, with $J_{k+1}<J_k$ below the cutoff and $J_{k+1}>J_k$ above it.
At $\varepsilon=\varepsilon_k(p)$, strict log-convexity gives $J_{k+2}/J_{k+1}>J_{k+1}/J_k=1$, and hence $\varepsilon_{k+1}(p)<\varepsilon_k(p)$.
Solving $J_1/J_0=(1-p)(1+2z)=1$ gives $z=p/[2(1-p)]$ and the stated formula for $\varepsilon_0(p)$.

For any $N\geq3$, let $K=\lfloor(N-1)/2\rfloor$.
Lemma \ref{lem:independent-payoff} excludes every odd $M>1$, and the error formulas above reduce the remaining comparison to minimizing $J_k$ over $0\leq k\leq K$.
Since the adjacent ratios strictly increase in $k$, $J_k$ decreases while these ratios are below one and increases after they exceed one.
If $J_{k+1}/J_k=1$ for $k<K$, exactly $k$ and $k+1$ minimize $J_k$.
The ordered cutoffs therefore give the stated optimal sets, including the ties at the cutoffs.
For odd $N$, $2K=N-1$; for even $N$, $2K=N-2$, which also proves the extension in the footnote.

We next prove (i).
Using the binomial-tail decomposition and difference in the proof of Lemma \ref{lem:independent-ratios}, we obtain
\begin{align*}
    D_k(\varepsilon,p)
    &:=\frac{(1-\varepsilon)^{2k+2}[J_{k+1}(z;p)-J_k(z;p)]}{q\binom{2k}{k}[q(1-q)]^k}\\
    &=\varepsilon(2-\varepsilon)\sum_{j=0}^k\frac{\binom{2k+1}{k+1+j}}{\binom{2k}{k}}\left(\frac{q}{1-q}\right)^j
    -(2p-1)\frac{2k+1}{k+1}(1-q)-\varepsilon(1-\varepsilon).
\end{align*}
The denominator is positive, and hence $D_k$ has the sign of $J_{k+1}-J_k$.
For fixed $\varepsilon$, increasing $p$ decreases $q$, weakly decreases the finite sum, and strictly increases $(2p-1)(1-q)$.
Thus, $D_k$ strictly decreases in $p$.
Starting at $D_k(\varepsilon_k(p),p)=0$, an increase in $p$ makes this expression negative; the unique cutoff must therefore increase.

For the endpoint limits, fix any $\varepsilon\in(0,1)$.
The finite sum is at least one, and thus
\begin{align*}
    & \lim_{p\searrow1/2}D_k(\varepsilon,p)\geq\varepsilon>0,\\
    & \lim_{p\nearrow1}D_k(\varepsilon,p)=-(1-\varepsilon)^2\frac{2k+1}{k+1}-\varepsilon(1-\varepsilon)<0.
\end{align*}
Hence, $\varepsilon_k(p)<\varepsilon$ near $p=1/2$ and $\varepsilon_k(p)>\varepsilon$ near $p=1$.
Since $\varepsilon$ was arbitrary, the cutoff converges to zero and one at the respective endpoints.

Finally, we prove (ii).
At $z=1/[4(1-p)]>p$, the adjacent ratios strictly increase to one, and hence every finite ratio is below one.
The corresponding parameter is
\begin{align*}
    \varepsilon_\infty(p)=1-\frac{4(1-p)}{1+4(1-p)^2}
    =\frac{(2p-1)^2}{1+4(1-p)^2},
\end{align*}
which is below every $\varepsilon_k(p)$.
For any $\varepsilon>\varepsilon_\infty(p)$, the ratio limit $4(1-p)z$ exceeds one. Thus, $J_{k+1}/J_k>1$ and hence $\varepsilon_k(p)<\varepsilon$ for all sufficiently large $k$.
These bounds prove $\varepsilon_k(p)\to\varepsilon_\infty(p)$.

At or below $\varepsilon_\infty(p)$, every adjacent ratio is below one, and the largest feasible index $K$ uniquely minimizes $J_k$.
Above $\varepsilon_\infty(p)$, the ratios eventually exceed one, and thus the minimizers over all $k\geq0$ form a fixed set of one index or two adjacent indices.
Once $K$ includes these indices, increasing $N$ does not change the optimal set.
\end{proof}

\singlespacing 
\bibliography{Reference}

@article{smith2000pathological,
  title={Pathological outcomes of observational learning},
  author={Smith, Lones and S{\o}rensen, Peter},
  journal={Econometrica},
  volume={68},
  number={2},
  pages={371--398},
  year={2000},
  publisher={Wiley Online Library}
}

@article{banerjee1992simple,
  title={A simple model of herd behavior},
  author={Banerjee, Abhijit V},
  journal={The Quarterly Journal of Economics},
  volume={107},
  number={3},
  pages={797--817},
  year={1992},
  publisher={MIT Press}
}

@article{bikhchandani1992theory,
  title={A theory of fads, fashion, custom, and cultural change as informational cascades},
  author={Bikhchandani, Sushil and Hirshleifer, David and Welch, Ivo},
  journal={Journal of Political Economy},
  volume={100},
  number={5},
  pages={992--1026},
  year={1992},
  publisher={The University of Chicago Press}
}

@article{ccelen2004observational,
  title={Observational learning under imperfect information},
  author={{\c{C}}elen, Bo{\u{g}}a{\c{c}}han and Kariv, Shachar},
  journal={Games and Economic Behavior},
  volume={47},
  number={1},
  pages={72--86},
  year={2004},
  publisher={Elsevier}
}

@article{sgroi2002optimizing,
  title={Optimizing information in the herd: Guinea pigs, profits, and welfare},
  author={Sgroi, Daniel},
  journal={Games and Economic Behavior},
  volume={39},
  number={1},
  pages={137--166},
  year={2002},
  publisher={Elsevier}
}

@article{acemoglu2011bayesian,
  title={Bayesian learning in social networks},
  author={Acemoglu, Daron and Dahleh, Munther A and Lobel, Ilan and Ozdaglar, Asuman},
  journal={The Review of Economic Studies},
  volume={78},
  number={4},
  pages={1201--1236},
  year={2011},
  publisher={Oxford University Press}
}

@article{foster1995learning,
  title={Learning by doing and learning from others: Human capital and technical change in agriculture},
  author={Foster, Andrew D and Rosenzweig, Mark R},
  journal={Journal of Political Economy},
  volume={103},
  number={6},
  pages={1176--1209},
  year={1995},
  publisher={The University of Chicago Press}
}

@article{munshi2004social,
  title={Social learning in a heterogeneous population: Technology diffusion in the Indian Green Revolution},
  author={Munshi, Kaivan},
  journal={Journal of Development Economics},
  volume={73},
  number={1},
  pages={185--213},
  year={2004},
  publisher={Elsevier}
}

@article{slezak2000effect,
  title={The effect of organizational form on information flow and decision quality: Informational cascades in group decision making},
  author={Slezak, Steve L and Khanna, Naveen},
  journal={Journal of Economics \& Management Strategy},
  volume={9},
  number={1},
  pages={115--156},
  year={2000},
  publisher={Wiley Online Library}
}

@article{rosenberg2019efficiency,
  title={On the efficiency of social learning},
  author={Rosenberg, Dinah and Vieille, Nicolas},
  journal={Econometrica},
  volume={87},
  number={6},
  pages={2141--2168},
  year={2019},
  publisher={Wiley Online Library}
}

@article{ali2018herding,
  title={Herding with costly information},
  author={Ali, S Nageeb},
  journal={Journal of Economic Theory},
  volume={175},
  pages={713--729},
  year={2018},
  publisher={Elsevier}
}

@article{kartik2024beyond,
  title={Beyond unbounded beliefs: How preferences and information interplay in social learning},
  author={Kartik, Navin and Lee, SangMok and Liu, Tianhao and Rappoport, Daniel},
  journal={Econometrica},
  volume={92},
  number={4},
  pages={1033--1062},
  year={2024},
  publisher={Wiley Online Library}
}

@article{gale2003bayesian,
  title={Bayesian learning in social networks},
  author={Gale, Douglas and Kariv, Shachar},
  journal={Games and Economic Behavior},
  volume={45},
  number={2},
  pages={329--346},
  year={2003},
  publisher={Elsevier}
}

@article{arieli2019multidimensional,
  title={Multidimensional social learning},
  author={Arieli, Itai and Mueller-Frank, Manuel},
  journal={The Review of Economic Studies},
  volume={86},
  number={3},
  pages={913--940},
  year={2019},
  publisher={Oxford University Press}
}

@article{arieli2021general,
  title={A general analysis of sequential social learning},
  author={Arieli, Itai and Mueller-Frank, Manuel},
  journal={Mathematics of Operations Research},
  volume={46},
  number={4},
  pages={1235--1249},
  year={2021},
  publisher={INFORMS}
}

@article{lobel2015information,
  title={Information diffusion in networks through social learning},
  author={Lobel, Ilan and Sadler, Evan},
  journal={Theoretical Economics},
  volume={10},
  number={3},
  pages={807--851},
  year={2015},
  publisher={Wiley Online Library}
}

@article{bikhchandani2024information,
  title={Information cascades and social learning},
  author={Bikhchandani, Sushil and Hirshleifer, David and Tamuz, Omer and Welch, Ivo},
  journal={Journal of Economic Literature},
  volume={62},
  number={3},
  pages={1040--1093},
  year={2024},
  publisher={American Economic Association 2014 Broadway, Suite 305, Nashville, TN 37203-2425}
}

@article{hann2018speed,
  title={The speed of sequential asymptotic learning},
  author={Hann-Caruthers, Wade and Martynov, Vadim V and Tamuz, Omer},
  journal={Journal of Economic Theory},
  volume={173},
  pages={383--409},
  year={2018},
  publisher={Elsevier}
}

@incollection{golub2016,
  author = {Golub, Benjamin and Sadler, Evan},
  isbn = {9780199948277},
  title = "Learning in social networks",
  booktitle = "The Oxford Handbook of the Economics of Networks",
  publisher = {Oxford University Press},
  address = "Oxford",
  year = {2016},
  month = {04},
  Pages = "504–542",
}

@article{mossel2015strategic,
  title={Strategic learning and the topology of social networks},
  author={Mossel, Elchanan and Sly, Allan and Tamuz, Omer},
  journal={Econometrica},
  volume={83},
  number={5},
  pages={1755--1794},
  year={2015},
  publisher={Wiley Online Library}
}

@article{dasaratha2019aggregative,
  title={Aggregative efficiency of bayesian learning in networks},
  author={Dasaratha, Krishna and He, Kevin},
  journal={arXiv preprint arXiv:1911.10116},
  year={2019}
}

@article{ali2018role,
  title={On the role of responsiveness in rational herds},
  author={Ali, S Nageeb},
  journal={Economics Letters},
  volume={163},
  pages={79--82},
  year={2018},
  publisher={Elsevier}
}

@article{golub2010naive,
  title={Naive learning in social networks and the wisdom of crowds},
  author={Golub, Benjamin and Jackson, Matthew O},
  journal={American Economic Journal: Microeconomics},
  volume={2},
  number={1},
  pages={112--149},
  year={2010},
  publisher={American Economic Association}
}

@article{dasaratha2023learning,
  title={Learning from neighbours about a changing state},
  author={Dasaratha, Krishna and Golub, Benjamin and Hak, Nir},
  journal={Review of Economic Studies},
  volume={90},
  number={5},
  pages={2326--2369},
  year={2023},
  publisher={Oxford University Press US}
}

@article{ellison1993rules,
  title={Rules of thumb for social learning},
  author={Ellison, Glenn and Fudenberg, Drew},
  journal={Journal of Political Economy},
  volume={101},
  number={4},
  pages={612--643},
  year={1993},
  publisher={The University of Chicago Press}
}

@article{ellison1995word,
  title={Word-of-mouth communication and social learning},
  author={Ellison, Glenn and Fudenberg, Drew},
  journal={The Quarterly Journal of Economics},
  volume={110},
  number={1},
  pages={93--125},
  year={1995},
  publisher={Oxford University Press}
}

@article{bala1998learning,
  title={Learning from neighbours},
  author={Bala, Venkatesh and Goyal, Sanjeev},
  journal={The Review of Economic Studies},
  volume={65},
  number={3},
  pages={595--621},
  year={1998},
  publisher={Oxford University Press}
}

@article{bala2001conformism,
  title={Conformism and diversity under social learning},
  author={Bala, Venkatesh and Goyal, Sanjeev},
  journal={Economic Theory},
  volume={17},
  number={1},
  pages={101--120},
  year={2001},
  publisher={Springer}
}

@article{demarzo2003persuasion,
  title={Persuasion bias, social influence, and unidimensional opinions},
  author={DeMarzo, Peter M and Vayanos, Dimitri and Zwiebel, Jeffrey},
  journal={The Quarterly Journal of Economics},
  volume={118},
  number={3},
  pages={909--968},
  year={2003},
  publisher={Oxford University Press}
}

@article{jadbabaie2012nonbayesian,
  title={Non-Bayesian social learning},
  author={Jadbabaie, Ali and Molavi, Pooya and Sandroni, Alvaro and Tahbaz-Salehi, Alireza},
  journal={Games and Economic Behavior},
  volume={76},
  number={1},
  pages={210--225},
  year={2012},
  publisher={Elsevier}
}

@article{molavi2018theory,
  title={A theory of non-Bayesian social learning},
  author={Molavi, Pooya and Tahbaz-Salehi, Alireza and Jadbabaie, Ali},
  journal={Econometrica},
  volume={86},
  number={2},
  pages={445--490},
  year={2018},
  publisher={Wiley Online Library}
}

@article{eyster2010naive,
  title={Naive herding in rich-information settings},
  author={Eyster, Erik and Rabin, Matthew},
  journal={American Economic Journal: Microeconomics},
  volume={2},
  number={4},
  pages={221--243},
  year={2010},
  publisher={American Economic Association}
}

@article{bohren2016informational,
  title={Informational herding with model misspecification},
  author={Bohren, J Aislinn},
  journal={Journal of Economic Theory},
  volume={163},
  pages={222--247},
  year={2016},
  publisher={Elsevier}
}

@article{dasaratha2020naive,
  title={Network structure and naive sequential learning},
  author={Dasaratha, Krishna and He, Kevin},
  journal={Theoretical Economics},
  volume={15},
  number={2},
  pages={415--444},
  year={2020},
  publisher={Wiley Online Library}
}

@article{dasaratha2021experiment,
  title={An experiment on network density and sequential learning},
  author={Dasaratha, Krishna and He, Kevin},
  journal={Games and Economic Behavior},
  volume={128},
  pages={182--192},
  year={2021},
  publisher={Elsevier}
}

@article{frick2020misinterpreting,
  title={Misinterpreting others and the fragility of social learning},
  author={Frick, Mira and Iijima, Ryota and Ishii, Yuhta},
  journal={Econometrica},
  volume={88},
  number={6},
  pages={2281--2328},
  year={2020},
  publisher={Wiley Online Library}
}

@article{frick2023belief,
  title={Belief convergence under misspecified learning: A martingale approach},
  author={Frick, Mira and Iijima, Ryota and Ishii, Yuhta},
  journal={The Review of Economic Studies},
  volume={90},
  number={2},
  pages={781--814},
  year={2023},
  publisher={Oxford University Press}
}

@article{arieli2025hazards,
  title={The hazards and benefits of condescension in social learning},
  author={Arieli, Itai and Babichenko, Yakov and M{\"u}ller, Stephan and Pourbabaee, Farzad and Tamuz, Omer},
  journal={Theoretical Economics},
  volume={20},
  number={1},
  pages={27--56},
  year={2025},
  publisher={Wiley Online Library}
}

@article{golub2012homophily,
  title={How homophily affects the speed of learning and best-response dynamics},
  author={Golub, Benjamin and Jackson, Matthew O},
  journal={The Quarterly Journal of Economics},
  volume={127},
  number={3},
  pages={1287--1338},
  year={2012},
  publisher={Oxford University Press}
}

@article{bohren2021learning,
  title={Learning with heterogeneous misspecified models: Characterization and robustness},
  author={Bohren, J Aislinn and Hauser, Daniel N},
  journal={Econometrica},
  volume={89},
  number={6},
  pages={3025--3077},
  year={2021},
  publisher={Wiley Online Library}
}

@article{lorenz2011social,
  title={How social influence can undermine the wisdom of crowd effect},
  author={Lorenz, Jan and Rauhut, Heiko and Schweitzer, Frank and Helbing, Dirk},
  journal={Proceedings of the National Academy of Sciences},
  volume={108},
  number={22},
  pages={9020--9025},
  year={2011},
  doi={10.1073/pnas.1008636108}
}

@article{muchnik2013social,
  title={Social influence bias: A randomized experiment},
  author={Muchnik, Lev and Aral, Sinan and Taylor, Sean J.},
  journal={Science},
  volume={341},
  number={6146},
  pages={647--651},
  year={2013},
  doi={10.1126/science.1240466}
}

@article{becker2017network,
  title={Network dynamics of social influence in the wisdom of crowds},
  author={Becker, Joshua and Brackbill, Devon and Centola, Damon},
  journal={Proceedings of the National Academy of Sciences},
  volume={114},
  number={26},
  pages={E5070--E5076},
  year={2017},
  doi={10.1073/pnas.1615978114}
}

@article{cai2009observational,
  title={Observational Learning: Evidence from a Randomized Natural Field Experiment},
  author={Cai, Hongbin and Chen, Yuyu and Fang, Hanming},
  journal={American Economic Review},
  volume={99},
  number={3},
  pages={864--882},
  year={2009},
  doi={10.1257/aer.99.3.864}
}

@article{eyster2014extensive,
  title={Extensive imitation is irrational and harmful},
  author={Eyster, Erik and Rabin, Matthew},
  journal={The Quarterly Journal of Economics},
  volume={129},
  number={4},
  pages={1861--1898},
  year={2014},
  doi={10.1093/qje/qju021}
}
\addcontentsline{toc}{section}{Reference}

\end{document}